\documentclass[11pt,a4paper]{article}

\usepackage{url}

\usepackage{amsmath, amssymb, amsfonts, nicefrac,soul,amsthm,nicefrac,mathabx,mathtools,dsfont}
\usepackage{titlesec}
\usepackage[outercaption]{sidecap}
\usepackage{enumitem}
\usepackage{color}
\usepackage{setspace}
\usepackage{tikz}

\usepackage{mathptmx}
\usepackage{times}
\usepackage{fullpage}
\usepackage{algorithm, algorithmic}
\usepackage[capitalize, nameinlink]{cleveref}
\crefdefaultlabelformat{#2\textbf{#1}#3} % <-- Only #1 in \textbf
\crefname{section}{\textbf{Section}}{\textbf{Sections}}
\Crefname{section}{\textbf{Section}}{\textbf{Sections}}
\newtheorem{theorem}{Theorem}
\newtheorem{definition}[theorem]{Definition}
\newtheorem{proposition}[theorem]{Proposition}
\newtheorem{lemma}[theorem]{Lemma}
\newtheorem{claim}[theorem]{Claim}
\newtheorem{corollary}[theorem]{Corollary}
\newtheorem{fact}[theorem]{Fact}

\newcommand{\Lap}[1]{\mathop{\mathsf{Lap(#1)}}}
\newcommand{\poly}{\mathop{\mathrm{poly}}}

\newcommand{\R}{\mathbb{R}}
\newcommand{\N}{\mathbb{N}}

\renewcommand{\epsilon}{\varepsilon}
\newcommand{\calP}{\mathcal{P}}

\newcommand{\TD}{\mathrm{TD}}

\newcommand{\diam}{\mathrm{diam}}
\newcommand{\G}{\mathcal{G}}
\newcommand{\TDC}{\mathsf{TDC}}
\mathchardef\mhyphen="2D
\newcommand{\lTDC}{\ell\mhyphen\mathsf{TDC}}	
\newcommand{\B}{\mathsf{B}}
\newcommand{\Bopt}{\mathsf{B}^*}
\newcommand{\vol}{\mathrm{vol}}
\newcommand{\CH}{\mathsf{CH}}
\newcommand{\calS}{\mathcal{S}}
\newcommand{\width}{\mathrm{width}}

\allowdisplaybreaks

\title{Private Approximations of a Convex Hull in Low Dimensions}
\author{
	Yue Gao\thanks{Y.G.~gratefully acknowledges funding from the Natural Sciences and Engineering Research Council of Canada (NSERC) and from the Alberta Machine Intelligence Institute (Amii).}\\
	Dept.~of Computing Science\\
	University of Alberta\\
	Edmonton, AB Canada\\
	\texttt{gao12@ualberta.ca}
	\and
	Or Sheffet\thanks{O.S.~gratefully acknowledges the Natural Sciences and Engineering Research Council of Canada (NSERC) for supporting O.S.~with grant \#2017–06701, which also helped fund Y.G.~when she was advised by O.S.~at the University of Alberta. Much of this work was done when O.S.~was a visitor at the ``Data Privacy: Foundations and Application'' program held in spring 2019 at the Simons Institute for the Theory of Computing, UC Berkeley.}\\
	Faculty of Engineering\\
	Bar-Ilan University\\
	Ramat-Gan, Israel\\
	\texttt{or.sheffet@biu.ac.il}
}
\date{}
\begin{document}

\begin{titlepage}
\maketitle 

\begin{abstract}
	We give the first differentially private algorithms that estimate a variety of geometric features of points in the Euclidean space, such as diameter, width, volume of convex hull, min-bounding box, min-enclosing ball etc. Our work relies heavily on the notion of \emph{Tukey-depth}. Instead of (non-privately) approximating the convex-hull of the given set of points $P$, our algorithms approximate the geometric features of the $\kappa$-Tukey region induced by $P$ (all points of Tukey-depth $\kappa$ or greater). Moreover, our approximations are all bi-criteria: for any geometric feature $\mu$ our $(\alpha,\Delta)$-approximation is a value ``sandwiched'' between $(1-\alpha)\mu(D_P(\kappa))$ and $(1+\alpha)\mu(D_P(\kappa-\Delta))$.
	
	Our work is aimed at producing a \emph{$(\alpha,\Delta)$-kernel of $D_P(\kappa)$}, namely a set $\calS$ such that (after a shift) it holds that $(1-\alpha)D_P(\kappa)\subset \CH(\calS) \subset (1+\alpha)D_P(\kappa-\Delta)$. We show that an analogous notion of a bi-critera approximation of a directional kernel, as originally proposed by~\cite{AgarwalHV04}, \emph{fails} to give a kernel, and so we result to subtler notions of approximations of projections that do yield a kernel. First, we give differentially private algorithms that find $(\alpha,\Delta)$-kernels for a ``fat'' Tukey-region. Then, based on a private approximation of the min-bounding box, we find a transformation that does turn $D_P(\kappa)$ into a ``fat'' region \emph{but only if} its volume is proportional to the volume of $D_P(\kappa-\Delta)$. Lastly, we give a novel private algorithm that finds a depth parameter $\kappa$ for which the volume of $D_P(\kappa)$ is comparable to $D_P(\kappa-\Delta)$. We hope this work leads to the further study of the intersection of differential privacy and computational geometry.
\end{abstract}

\thispagestyle{empty}

\end{titlepage}
%\setcounter{page}{14}
%\appendix
\section{Introduction}
\label{sec:intro}

%Many datasets, of low dimension, contain sensitive data. Analyzing them while maintaining their privacy is vital
With modern day abundance of data, there are numerous datasets that hold the sensitive and personal details of individuals, yet collect only a few features per user. Examples of such \emph{low-dimensional} datasets include locations (represented as points on the $2D$-plane), medical data composed of only a few measurements (e.g.~\cite{Sanuki2017CufflessCB, Weiss2019SmartphoneAS}), or high-dimensional data restricted to a small subset of features (often selected for the purpose of data-visualization). It is therefore up to us to make sure that the analyses of such sensitive datasets do not harm the privacy of their participants. Differentially private algorithms~\cite{DworkMNS06,DworkKMMN06} alleviate such privacy concerns as they guarantee that the presence or absence of any single individual in the dataset has only a limited affect on any outcome.

Often (again, commonly motivated by visualization), understanding the geometric features of such low-dimensional datasets is a key step in their analysis. 
%The rich field of \emph{computational geometry} is centered on the design of algorithms solving such geometric problems. 
Yet, to this day, very little work has been done to establish differentially private algorithms that approximate the data's geometrical features. This should not come as a surprise seeing as most geometric features~--- such as diameter, width,\footnote{The min gap between two hyperplanes that ``sandwich'' the data.} volume of convex-hull, min-bounding ball radius, etc.~--- are highly sensitive to the presence / absence of a single datum. Moreover, while it is known that differential privacy generalizes~\cite{DworkFHPRR15,BassilyNSSSU16}, geometrical properties often do not: if the dataset $P$ is composed on $n$ i.i.d draws from a distribution $\calP$ then it might still be likely that, say, $\diam(P)$ and $\diam(\calP)$ are quite different.\footnote{For example, consider $\calP$ as a uniform distribution over $2n$ discrete points whose diameter greatly shrinks unless two specific points are drawn into $P$.}

%Little intersection between DP and CG, high sensitivity

%Median - canonical way to overcome high sensitivity queries
But differential privacy has already overcome the difficulty of large sensitivity in many cases, the leading example being the median~--- despite the fact that the median may vary greatly with the addition of a new entry into the data, we are still capable of privately approximating the median. The crux in differentially private median approximation is that the quality of the approximation is not measured by the actual distance between the true input-median and the result of the algorithm, but rather by the probability mass of the input's CDF ``sandwiched'' between the true median and the output of the private algorithm. 
A similar effect takes place in our work. While we deal with geometric concepts that exhibit large sensitivity, we formulate \emph{robust} approximation guarantees of these concepts, guarantees that do generalize when the data is drawn i.i.d.~from some unknown distribution. And much like in previous works in differential privacy~\cite{BeimelMNS19,KaplanSS20}, our approximation rely heavily on the notion of the \emph{depth} of a point.

Specifically, our approximation guarantees are with respect to \emph{Tukey depth}~\cite{Tukey75}. Roughly speaking (see Section~\ref{sec:preliminaries}), a point $x$ has Tukey depth $\kappa$ w.r.t.~a dataset $P$, denoted $\TD(x,P)=\kappa$, if the smallest set $S\subset P$ one needs to remove from $P$ so that some hyperplane separates $x$ from $P\setminus S$ has cardinality $\kappa$. This also allows us to define the \emph{$\kappa$-Tukey region} $D_P(\kappa) = \{x\in \R^d:~ \TD(x,P)\geq \kappa\}$. So, for example, $D_P(0)=\R^d$ and $D_P(1)=\CH(P)$ (the convex-hull of $P$). 
It follows from the definition that for any $1\leq \kappa_1 \leq \kappa_2$ we have $\CH(P) = D_P(1) \supset D_P(\kappa_1)\supset D_P(\kappa_2)$. It is known that for any dataset $P$ and depth $\kappa$ the Tukey-region $D_P(\kappa)$ is a convex polytope, and moreover (see~\cite{Edelsbrunner87}) that for any $P$ of size $n$ it holds that $D_P(n/(d+1))\neq\emptyset$. Moreover, there exists efficient algorithms (in low-dimensions) that find $D_P(\kappa)$.

One property of the Tukey depth, a pivotal property that enables differentially private approximations, is that it exhibits low-sensitivity at any given point. As noted by~\cite{BeimelMNS19}, it follows from the very definition of Tukey-depth that if we add or remove any single datapoint to/from $P$, then the depth of any given $x\in \R^d$ changes by no more than $1$. And so, in this work, we give bi-criteria approximations of key geometric features of $D_P(\kappa)$~--- where the quality of the approximation is measured both by a multiplicative factor and with respect to a \emph{shallower} Tukey region. Given a measure $\mu$ of the convex polytope $D_P(\kappa)$, such as diameter, width, volume etc., we return a $(\alpha,\Delta)$-approximation of $\mu$~--- a value lower bounded by $(1-\alpha)\mu(D_P(\kappa))$ and upper-bounded by $(1+\alpha)\mu(D_P(\kappa-\Delta))$. This implies that the quality of the approximation depends on \emph{both} the approximation parameters fed into the algorithm and also on the ``niceness'' properties of the data. For datasets where $\mu(D_P(\kappa-\Delta))\approx \mu(D_P(\kappa))$, our $(\alpha,\Delta)$-approximation is a good approximation of $\mu(D_P(\kappa))$, but for datasets where $\mu(D_P(\kappa-\Delta))\gg \mu(D_P(\kappa))$ our guarantee is rather weak. Note that no differentially private algorithm can correctly report for all $P$ whether $\mu(D_P(\kappa))$ and $\mu(D_P(\kappa-\Delta))$ are / are-not similar seeing as, as Figure~\ref{fig:tukey_volatility} shows, such proximity can be highly affected by the existence of a single datum in $P$.
Again, this is very much in  line with private approximations of the median~\cite{NissimRS07, BeimelNS13b}. Moreover, referring to the earlier discuss about generalizability~--- in the case where $P$ is drawn from a distribution $\calP$, it is known that $\forall x\in\R^d,~~|\frac 1 n\TD(x,P) - \TD(x,\calP)|= O(\sqrt{\frac{d\log(n)}{n}})$~\cite{BurrF17}, where $\TD(x,\calP)$ denotes the smallest measure $\calP$ places on any halfspace containing $x$. Thus, if $D_P(\kappa)$ and $D_P(\kappa-\Delta)$ vary drastically, then it follows that the distribution $\calP$ is ``volatile'' at depth $\frac{\kappa}n$.

\begin{figure}[t]
	\begin{minipage}[c]{0.4\textwidth}
		\caption{\label{fig:tukey_volatility}
			An example showing that $D_P(\kappa)$'s volume, width and proximity to shallower regions can be greatly affected by a single point in the input.}
	\end{minipage}
	\hfill
	\begin{minipage}[c]{0.6\textwidth}
		\centering
	\includegraphics[scale=0.3]{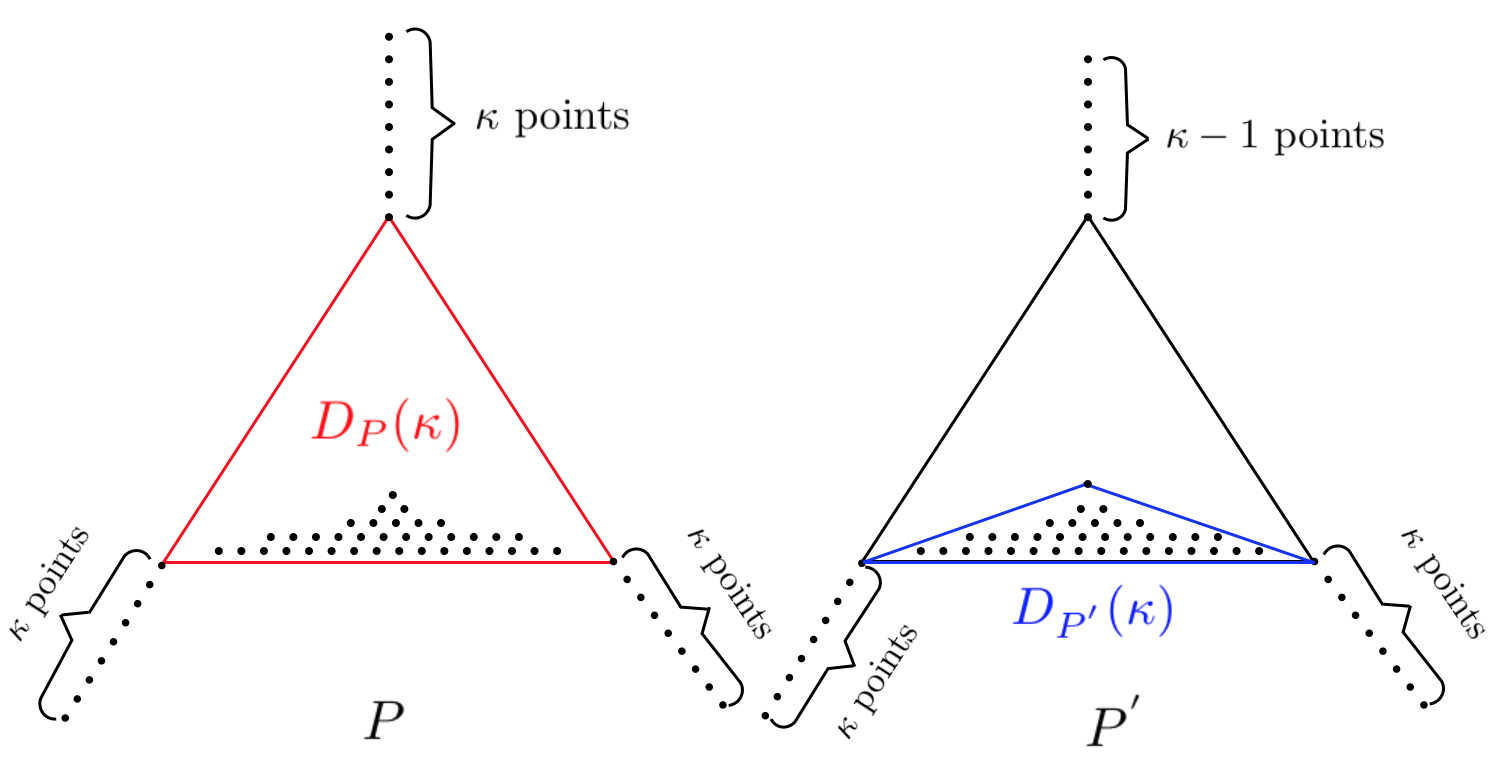}
	\end{minipage}
\end{figure}

Our main goal in this work is to produce a \emph{kernel} for $D_P(\kappa)$. Non privately, a $\alpha$-kernel~\cite{AgarwalHV04} of a dataset $P$ is a set $\calS\subset P$ where for any direction $u$ it holds that $(1-\alpha)\max_{p,q\in P}\langle p-q,u\rangle \leq\linebreak \max_{p,q\in \calS}\langle p-q,u\rangle \leq \max_{p,q\in P} \langle p-q,u\rangle$. Agarwal et al~\cite{AgarwalHV04} showed that for any $P$ there exists such a kernel whose size is $(\nicefrac 1 \alpha)^{O(d)}$. (Note how this implies that $|P|\gg (\nicefrac 1 \alpha)^{O(d)}$ since otherwise the non-private algorithm may as well output $P$ itself.) More importantly, the fact that $\calS$ is a $\alpha$-kernel implies that $(1-O(\alpha))\CH(P)\subset \CH(\calS)\subset \CH(P)$. It is thus tempting to define an analogous notion of $(\alpha,\Delta)$-kernel as ``for any direction $u$ we have  $(1-\alpha)\max_{p,q\in D_P(\kappa)}\langle p-q,u\rangle \leq \max_{p,q\in \calS}\langle p-q,u\rangle \leq (1+\alpha)\max_{p,q\in D_P(\kappa-\Delta)} \langle p-q,u\rangle$'' and hope that it yields that $(1-O(\alpha))D_P(\kappa)\subset \CH(\calS)\subset D_P(\kappa-\Delta)$. Alas, that is \emph{not} the case. Having $\calS\subset D_P(\kappa)$ turns out to be a crucial component in arguing about the containment of the convex-hulls, and the argument breaks without it. We give a counter example in a later discussion (Section~\ref{sec:DP_kernel_for_fat_Tukey_region}). Therefore, viewing this directional-width approximation property as means to an end, we define the notion of $(\alpha,\Delta)$-kernel directly w.r.t.~the containment of the convex bodies.
\begin{definition}
	\label{def:alpha_delta_kernel}
	Given a dataset $P$ and a depth parameter $\kappa$, a set $\calS$ is called a \emph{$(\alpha,\Delta)$-kernel for $D_P(\kappa)$} if there exists a point $c_1$ such that $(1-\alpha)(D_P(\kappa)-c_1)\subset \CH(\calS)-c_1$ and a point $c_2$ such that\linebreak \mbox{$\CH(\calS)-c_2 \subset (1+\alpha)(D_P(\kappa-\Delta)-c_2)$.}
\end{definition}
\noindent Note that in particular, a $(\alpha,\Delta)$-kernel gives the $(\alpha,\Delta)$-approximation of the projection along every direction $u$ proposed earlier (in quotation-marks above). In fact, a $(\alpha,\Delta)$-kernel yields $(\alpha,\Delta)$-approximations of many other properties of $D_P(\kappa)$, such as volume, min-bounding box, min-enclosing / max-enclosed ball radius, surface area, etc. Our work is the first to give a private approximation of \emph{any} of these concepts.

As it turns out, we are able to give a $(\alpha,\Delta)$-kernel of $D_P(\kappa)$ only when $D_P(\kappa)$ satisfies some ``niceness'' properties. We briefly describe the structure of our work to better explain these properties and how they relate. We begins with multiple preliminaries~--- in Section~\ref{sec:preliminaries} we establish background knowledge, and in Sections~\ref{sec:TDC} and~\ref{sec:approx_diam_and_width} we establish some basic privacy-preserving algorithms for tasks we require later.%, such as an efficient implementation of a private algorithm for finding a point in convex hull from~\cite{BeimelMNS19}, and private diameter and width approximations. 

Based on these rudimentary algorithms, we turn our attention towards the design of a private kernel approximation. In Section~\ref{sec:DP_kernel_for_fat_Tukey_region} we give our algorithm for finding a kernel, which works under the premise that the width of $D_P(\kappa)$ is large. This means that our goal is complete if we are able to assert, using a private algorithm, that $D_P(\kappa)$ has large width. So, in Section~\ref{sec:approx_bounding_box}, we give a private $(O(1),\Delta)$-approximation of the min-bounding box of $D_P(\kappa)$; and show that this box yields a transformation that turns $D_P(\kappa)$ into a region of large width, \emph{but only if} the volumes of $D_P(\kappa)$ and $D_P(\kappa-\Delta)$ are comparable. Finally, in Section~\ref{sec:finding_good_kappa}, we give an algorithm that finds a value of $\kappa$ for which is this premise about the volumes of $D_P(\kappa)$ and $D_P(\kappa-\Delta)$ holds, rendering us capable of privately finding a $(\alpha,\Delta)$-kernel for this particular $D_P(\kappa)$. We conclude in Section~\ref{sec:conclusion}, where we detail the applications of having a $(\alpha,\Delta)$-kernel and discuss many open problems.

Providing further details about the private approximation algorithms we introduce in this work requires that we first delve into some background details and introduce some key parameters.

\paragraph{The Setting: Low-Dimension and Small Granularity.} Differential privacy deals with the trade-offs between the privacy parameters, $\epsilon$ and $\delta$, and an algorithm's utility guarantee. Unlike the majority of works in differential privacy, we don't express these trade-offs based on the size $n$ of the data.\footnote{
Though $n$ comes into play in our work, both in requiring that for large enough $\kappa$ we have that $D_P(\kappa)\neq\emptyset$ and in bounding $\Delta$, since if $\Delta>n$ then it is trivial to give a $(\alpha,\Delta)$-kernel. Moreover, ideally we would have that $\Delta \leq \sqrt{dn\log(n)}$ so that both $D_P(\kappa)$ and $D_P(\kappa-\Delta)$ (roughly) represent the same Tukey-depth region w.r.t to the distribution the dataset was drawn from, based on the above-mentioned bounds of~\cite{BurrF17}.} Instead, in our work we upper bound the $\Delta$-term of a a private $(\alpha,\Delta)$-approximation as a function of the privacy- and accuracy-parameters, as well as additional two parameters. These two parameters are (i) the dimension, $d$, which we assume to be constant and so $n^{\poly(d)}$ is still considered efficient for our needs; and (ii) the granularity of the grid on which the data resides. In differential privacy, it is impossible to provide useful algorithms for certain basic tasks~\cite{BunNSV15} when the universe of possible entries is infinite. Therefore, we assume that the given input $P$ lies inside the hypercube $[0,1]^d$ and moreover~--- that its points reside on a grid $\G^d$ whose granularity is denoted as $\Upsilon$. This means that each coordinate of a point $p\in P$ can be described using $\upsilon = \log_2(1/\Upsilon)$ many bits. We assume here that $1/\Upsilon$ is large (say, all numbers are \texttt{int}s in $\texttt{C}$, so $\Upsilon=2^{-32}$), too large for the grid to be efficiently traversed. And so, for each $(\epsilon,\delta)$-differentially private algorithm we present, an algorithm that returns with a high probability of $1-\beta$ a $(\alpha,\Delta)$-approximation of some geometric feature of $D_P(\kappa)$, we upper bound the $\Delta$-term as a function of $(\alpha,\beta,\epsilon,\delta,d,\upsilon)$. (Of course, we must also have that $\kappa>\Delta$ otherwise the algorithm can simply return $[0,1]^d$.) In addition, we consider efficient any algorithm whose runtime is $(n\cdot\upsilon\cdot \epsilon^{-1}\cdot \alpha^{-1}\cdot\log(1/\beta\delta))^{\poly(d)}$.

Furthermore, Kaplan et al~\cite{KaplanSS20} gave a $\epsilon$-differentially private algorithm for detecting whether a given input $P$ has degenerated $\kappa$-Tukey regions~--- namely, $0$-volume polytopes that lie in some $j$-dimensional subspace ($j<d$). Moreover, if $D_P(\kappa)$ is degenerate, then the algorithm of~\cite{KaplanSS20} returns the affine subspace it lies in, so we can restrict our input to this subspace. We thus assume that as pre-processing to our algorithms, this detection algorithm was run and that we have also established that $n$ is sufficiently large so that for sufficiently large $\kappa$ it holds that $D_P(\kappa)$ is non-empty (otherwise we abort). And so, throughout this work we assume we deal with non-empty and non-$0$ volume Tukey regions.

\paragraph{Detailed Contribution and Organization.} 
First, in Section~\ref{sec:preliminaries} we survey some background in differential privacy and geometry. Our contributions are detailed in the remaining section and are as follows.\vspace{-2mm}
\begin{itemize}[leftmargin=*]
	\setlength\itemsep{0em}
	\item In Section~\ref{sec:TDC} we give an efficient implementation of the algorithm of~\cite{BeimelMNS19} for privately finding a point inside a convex hull. Beimel et al~\cite{BeimelMNS19} constructed a function for Tukey-Depth Completion~($\TDC$): given a prefix of $0\leq i<d$ coordinates, each $x\in\R$ is mapped to the max $\TD$ of a point whose first $i+1$ coordinates are the given prefix concatenated with $x$. Beimel et al showed that this $\TDC$-function is quasi-concave (details in Section~\ref{sec:TDC}) and so one can privately find $x\in \G$ with high $\TDC(x)$-value; so by repeating this process $d$ times one finds a point with high $\TD$. Unfortunately, the algorithm Beimel et al provide is inefficient, as it requires pre-processing involving the entire grid. We give an efficient algorithm for computing $\TDC$. We show that by first finding all non-empty $\kappa$-Tukey regions we can use LPs and so efficiently find the $\leq n$ points on the real line where the $\TDC$-function changes value. Therefore we can efficiently compute $\TDC(x)$ for any point on the line and find the $\max$-value of $\TDC$ on any interval efficiently. Next, applying an efficient algorithm for privately approximating quasi-concave function gives an overall  efficient algorithm with the same utility guarantee as in~\cite{BeimelMNS19}.
		
	We also show that the function that takes an additional parameter $\ell$ and maps $x$ to $\min\left\{\TDC(x),\TDC(x+\ell)\right\}$ is also quasi-concave and can also be computed efficiently. The two functions play an important role in the construction of following algorithms as we often rotate the space so that some direction $v$ aligns with first axis and then apply $\TDC$ to find a good extension of a particular coordinate along $v$ into a point inside a Tukey-region.
	
	\item In Section~\ref{sec:approx_diam_and_width} we give our efficient private algorithms for $(\alpha,\Delta^{\diam})$-diameter approximation and $(\alpha,\Delta^{\width})$-width approximation, as well as two similar algorithms for some specific tasks we require later. While the novelty of these algorithm lies in their premise, the algorithms themselves are quite standard and rely on the Sparse-Vector Technique.
	\item In Section~\ref{sec:DP_kernel_for_fat_Tukey_region} we give our private $(\alpha,\Delta)$-kernel approximation, that much like its non-private equivalent~\cite{AgarwalHV04}, requires some ``fatness'' condition. In fact, we have two somewhat different conditions. Our first algorithm requires a known (constant) lower bound on $\width(D_P(\kappa))$, and our second algorithm requires a known (constant) lower bound on the ratio $\frac{\width(D_P(\kappa))}{\diam(D_P(\kappa-\Delta))}$. More importantly, the resulting sets from each algorithm do not satisfy an analogous property to the non-private kernel definition of~\cite{AgarwalHV04}, but rather more intricate properties regarding projections along any direction. Thus, in Section~\ref{subsec:kernel_definitions}, prior to presenting the two algorithms, we prove that these two properties are sufficient for finding a $(\alpha,\Delta)$-kernel. These results may be of independent interest. 
	
	Clearly, there exists neighboring datasets where one dataset satisfies this lower bound on the width whereas its neighboring dataset violates the bound. 
	Thus, no private algorithm can tell for \emph{any} input whether $\width(D_P(\kappa))$ is large or not; but in Section~\ref{subsec:private_selection} we give a private algorithm that verifies whether some sufficient condition for large width holds.
	%Privately checking whether first condition, namely whether the width of the dataset is large, is simple (if $\width(D_P(\kappa+\Delta^{\width}))$ is large then we can privately assert that $\width(D_P(\kappa))$ is large based on the algorithm from Section~\ref{sec:approx_diam_and_width}). But privately checking for a sufficient condition under which the ratio of the width and the diameter is large is a more intricate task, which we detail in Section~\ref{subsec:private_selection}.
	
	\item We thus turn our attention to \emph{asserting} that the fatness condition required for the kernel-approximation algorithm does indeed hold. In Section~\ref{sec:approx_bounding_box} we give a private $(c,\Delta)$-approximation of the min bounding box problem~--- it returns a box that (a) contains $D_P(\kappa)$ and (b) with volume upper bounded by $c~\cdot~ \vol(D_P(\kappa-\Delta))$. Then, in Section~\ref{subsec:bounding_box_to_fatness} we argue that if $\vol(D_P(\kappa))\geq \vol(D_P(\kappa-\Delta))/2$ then this procedure also gives an affine transformation $T$ that makes $\width(T(D_P(\kappa))$ sufficiently large. (Obviously, if $\calS$ is a $(\alpha,\Delta)$-kernel for $T(D_P(\kappa))$ then $T^{-1}(\calS)$ is a $(\alpha,\Delta)$-kernel for $D_P(\kappa)$.)
	\item In Section~\ref{sec:finding_good_kappa} we give a private algorithm for finding a ``good'' value of $\kappa$, one for which it does hold that $\vol(D_P(\kappa))\geq \vol(D_P(\kappa-\Delta))/2$. We formulate a certain query $q$ where any $\kappa$ for which $q_P(\kappa)$ is large must also be a good $\kappa$, and then give a private algorithm for finding a $\kappa$ with a large $q_P(\kappa)$-value. The $\epsilon$-differentially private algorithm we give is actually rather novel~--- it is based on a combination of the Exponential-Mechanism with additive Laplace noise. Its privacy is a result of arguing that for any neighboring $P$ and $P'$ where $P'=P\cup\{x\}$ we can \emph{match} $\kappa$ with $\kappa+1$ so that $|q_P(\kappa)-q_{P'}(\kappa+1)|\leq 1$, and then using a few more observations that establish pure $\epsilon$-differential privacy (rather than $(\epsilon,\delta)$-DP). 
\end{itemize} 
Our work thus culminates in the following theorem.
\begin{theorem}
	\label{thm:main_thm}
	There exists an efficient $(\epsilon,\delta)$-differentially private algorithm, that for any sufficiently large dataset $P$, where $|P| \geq \tilde \Omega(d^4\upsilon\cdot\Delta)$, with probability $\geq1-\beta$ finds a $\kappa$ and a set $\calS$ such that $\calS$ is a $(\alpha,\Delta)$-kernel of $D_P(\kappa)$ where
	$\Delta = O(\frac {f(d)} \epsilon \cdot (\frac 1 \alpha)^{\frac d 2}\sqrt{\log(\frac 1 \delta)}\log(\frac{1}{\alpha\beta}))$ for some function $f(d)=2^{d^2\poly\log(d)}$.
\end{theorem}
\noindent In fact, it is also required that $\Delta\geq \Delta^{\rm BB}(d,\upsilon,\epsilon,\delta,\beta)$ where $\Delta^{\rm BB}$ is guarantee of the private min-bounding-box algorithm, as detailed in Theorem~\ref{thm:dp_bounding_box}; yet this lower-bound holds under a very large regime of parameters. Moreover, there also exists a different algorithm that returns a somewhat better guarantee, replacing $f(d)$ with some $\tilde f(d,\upsilon)$-function and reducing the dependency on $\alpha$ to $(1/\alpha)^{\frac{d-1}2}$, and more importantly~--- promising that $\calS \subset D_P(\kappa-\Delta)$.

\paragraph{Additional Works.} In addition to the two works~\cite{BeimelMNS19, KaplanSS20} that privately find a point inside a convex hull, it is also worth mentioning the works regarding privately approximating the diameter~\cite{NissimSV16,NissimS18} (they return a $O(1)$-approximation of the diameter that may miss a few points), and the work of~\cite{KaplanMMS19} that privately approximates a $k$-edges polygon yet requires a dataset of points where many lie inside the polygon and many lie \emph{outside} the polygon. No additional works that we know of lie in the intersection of differential privacy and computational geometry. Computational geometry, of course, is a rich fied of computer science replete with many algorithms for numerous tasks in geometry. Our work only give private analogs to (a few of) the algorithms of~\cite{Chan02,AgarwalHV04}, but there are far many more algorithms to be privatized and the reader is referred to~\cite{Har-Peled11} for a survey of the field. Many works deal with computing the Tukey-depth and the Tukey region~\cite{RousseeuwR98,LiuMM19}, and others give statistical convergence rates for the Tukey-depth when the data is composed of i.i.d draws from a distribution~\cite{ZuoS00b,BurrF17,Brunel19}.

\section{Preliminaries}
\label{sec:preliminaries}

\subsection{Geometry}
\label{subsec:prelim_geometry}

%Convex geometry, angle-cover, upper and lower-bounds of volume based on projections,
In this work we use $\langle\cdot,\cdot\rangle$ to denote the inner-product between two vectors in $\R^d$.
A closed half-space is defined by a vector $u$ and a scalar $\lambda$ and it is the set $\{x\in \R^d: \langle x,u\rangle \leq \lambda\}$. A polytope is the intersection of finitely many closed half-spaces. It is known that a polytope $S$ is a convex body: for any $x,y\in S$ and any $\lambda\in [0,1]$ it also holds that $\lambda x + (1-\lambda)y \in S$. Given a collection of points $P$, their convex-hull is the set of all points that can be written as a convex combination of the points in $P$. For a polytope $P$ and a point $x$ we define $P-x$ as the shift of $P$ by $x$ (namely $z\in P-x$ iff $\exists y\in P$ s.t. $z=y-x$), and we define by $cP$ the blow-up of $P$ by a scalar $c$. 

An inner product $\langle x, u\rangle = \|x\|\|u\|\cos(\angle(x,u))$ is also known a projection of $x$ onto the subspace spanned by $u$. A projection onto a subspace $\Pi^V$ maps any $x\in \R^d$ to its closest point in the subspace $V$. The following fact is well-known.
\begin{fact}
	\label{fact:volume_bounds_by_projections}
	Let $S$ be a convex body. Let $u$ be any vector and let $\Pi^{\perp u}$ be the projection onto the subspace orthogonal to $u$. Denote $\ell$ as the max-length of the intersection of $S$ with any affine line in direction $u$, and denote $A$ as the volume of the projection of $S$ onto the subspace orthogonal to $u$,  $A=\vol(\Pi^{\perp u}(S))$. Then $\frac {A\cdot \ell}d\leq \vol(S) \leq A\cdot \ell$.
\end{fact}
\noindent The fact follows from reshaping $S$ so that it is contained in the ``cylinder'' whose base is $A$ and height is $\ell$, and contains a ``pyramid'' with base of $A$ and height of $\ell$.

The unit-sphere $\mathbb{S}^{d-1}$ is the set of vectors in $\R^d$ of length $1$. The \emph{diameter} of the convex body $P$ is defined as $\diam(P) = \max_{p,q\in P}\|p-q\|$, and it is simple to see that $\diam(P) = \max_{u\in\mathbb{S}^{d-1}}\max_{p,q\in P}\langle p-q,u\rangle$. The \emph{width} of a convex body $P$ is analogously defined as $\width(P)=\min_{u\in\mathbb{S}^{d-1}}\max_{p,q\in P}\langle p-q,u\rangle$. A \emph{$\zeta$-angle cover} of the unit sphere is a set of vectors $V_\zeta$ such that for any $v\in \mathbb{S}^{d-1}$ there exist $u$ such that $\angle(u,v)\leq\zeta$. It is known that each vector in the sphere can be characterized by $d-1$ angles $\varphi_1, \varphi_2,..., \varphi_{d-1}$ where $\varphi_i\in [0,2\pi]$ and for any other $j$, $\varphi_j\in [0,\pi]$. Therefore by discretizing the interval $[0,\pi]$ we can create a $\zeta$-angle cover of size $2\lceil\pi/\zeta \rceil^{d-1}$. 
\begin{proposition}
	\label{pro:nearest_vector_in_angle_cover}
	Let $\zeta<1/2$. Let $V_\zeta$ be a $\zeta$-angle cover of $\mathbb{S}^{d-1}$. Then for any $u\in \mathbb{S}^{d-1}$ the closest $v\in V_\zeta$ satisfies that $\|u-v\|\leq\sqrt 2\zeta$.
\end{proposition}
\begin{proof}
	For a given $u$ and its closest $v$ we have that $\|u-v\|^2 = \langle u-v,u-v\rangle = 2\cdot 1^2 - 2\cdot 1^2\cos(\angle(u,v)) \leq 2-2\cos(\zeta) \stackrel{\rm Taylor}\leq  2-2(1-\zeta^2)\leq 2\zeta^2$. The fact that  $\zeta<\frac 1 2$ allows us to lower bound the Taylor expansion of $\cos(\zeta)$.
\end{proof}

\paragraph{Tukey Depth.} Given a finite set of points $P\subset \R^d$, the Tukey depth~\cite{Tukey75} of a point $x\in \R^d$ w.r.t $P$ is defined $\TD(x,P) = \min_{u\in \mathbb{S}^{d-1}} |\{p\in P:~ \langle p,u\rangle \leq \langle x,u\rangle \}|$. It is also the min-size of a set $S\subset P$ such that there exists a hyperplane separating $x$ from $P\setminus S$. Given $P$ and a depth parameter $\kappa\geq 0$ we denote the \emph{$\kappa$-Tukey region} as $D_P(\kappa)  = \{x\in \R^d: \TD(x,P)\geq \kappa\}$. Note that $D_P(0)=\R^d$ and $D_P(1)=\CH(P)$. It is known that for any set of points $P$ it holds that $\kappa^* = \max_{x}\TD(x,P)$ is upper bounded by $|P|/2$ and lower bounded by $|P|/(d+1)$ (see~\cite{Edelsbrunner87}). It is also known that for any $\kappa$, the set $D_P(\kappa)$ (assuming its non-empty) is a convex polytope which is the intersection of all closed halfspaces that contain at least $n-\kappa+1$ points out of $P$ \cite{RousseeuwR98}, this yields a simple algorithm to compute the $\kappa$-Tukey region in time $O(n^{(d-1)\lfloor\frac{d}{2} \rfloor})$. There is a faster algorithm to compute the $\kappa$-Tukey region in time $O(n^{d}\log n)$~\cite{LiuMM19}, and so to compute \emph{all} of the non-empty Tukey-regions in time $O(n^{d+1}\log n)$. There is also an efficient algorithm~\cite{Liu17} for computing the Tukey-depth of a given point in time $O(n^{d-1}\log n)$.

%DP (SVT, composition basic and advanced)
\subsection{Differential Privacy}
\label{subsec:prelim_DP}

Differential privacy is a mathematically rigorous notion of preserving privacy in data analysis. Formally, two datasets $P$ and $P'$ are called neighboring if they differ on a single datum, and in thus work we assume that this means that $|P\triangle P'|=1$.
\begin{definition}
	\label{def:DP}
	A randomized algorithm ${\cal A}$ is said to be $(\epsilon,\delta)$-differentially private (DP) if for any two neighboring datasets $P$ and $P'$ and for any set of possible outputs $S$ it holds that $\Pr[{\cal A}(P)\in S]\leq e^{\epsilon}\Pr[{\cal A}(P')\in S]+\delta$. When $\delta=0$ we say ${\cal A}$ is $\epsilon$-DP or $\epsilon$-pure DP.
\end{definition} 

Differential privacy composes. Namely, if ${\cal A}$ is $(\epsilon,\delta)$-DP and ${\cal B}$ is $(\epsilon',\delta')$-DP, then applying ${\cal A}$ and then applying ${\cal B}$ sequentially on $P$ is a $(\epsilon+\epsilon',\delta+\delta')$-DP algorithm (even when ${\cal B}$ is chosen adaptively, based of ${\cal A}$'s output). Moreover, post-processing the output of a $(\epsilon,\delta)$-DP algorithm cannot increase either $\epsilon$ or $\delta$. It is also worth noting the advanced-composition theorem of~\cite{DworkRV10}, where the sequential application of $k$ $(\epsilon,\delta)$-DP algorithms yields in total an algorithm which is  $(\epsilon\sqrt{k\ln(1/k\delta)}, 2k\delta)$-DP (provided $\epsilon<1$). Since we deal with a constant dimension $d$, then whenever we compose $\poly(d)$-many mechanisms, we rely on the basic composition; and whenever we compose $\exp(d)$-many mechanisms, we rely on the advanced composition.

Perhaps most common out of all DP-algorithms is the Laplace additive noise. Given a function $f$ that maps inputs to real numbers, we define its global sensitivity $GS(f) = \max_{P,P' {\rm neighbors}} |f(P)-f(P')|$. It is known that outputting $f(P)+\Lap{GS(f)/\epsilon}$ is $\epsilon$-DP. Unfortunately, many of the functions we discuss in this work exhibit high global sensitivity, rendering this mechanism useless. It is also worth noting the \emph{Sparse Vector Technique} which is a $\epsilon$-DP algorithm that allows us to assess $t$ queries $q_1, q_2,.., q_t$, each with $GS(q_i)=1$, and halt on the very first query that exceeds a certain (noisy) threshold. Our algorithms repeatedly rely on the SVT.

For more details on differential privacy the reader is deferred to the monographs~\cite{DworkR14, Vadhan17}.

\newcommand{\alphaqc}{\alpha^{\rm qc}}
\subsubsection{Private Approximation of Quasi-Concave Functions}
\label{subsec:prelim_quasi_concave}

In our work we use as ``building blocks'' several known techniques in differential privacy regarding private approximations of quasi-concave functions. We say a function $q:\R \to \R$ is a quasi-concave function if for any $x\leq y \leq z$ it holds that
\[ q(y) \geq \min\{ q(x),q(z)\}  \]
Quasi-concave functions that obtain a maximum (namely, there exists some $x\in \R$ such that $q(x) \geq q(y)$ for any other $y\in \R$) have the property that the maximum is obtained on a single closed interval $I=[x,y]$ (we also consider the possibility of  $I$ containing only a single point, i.e., $y=x$). Moreover, it follows that on the interval $(-\infty, x)$ the function $q$ is a monotone non-decreasing function and on the interval $(y,\infty)$ the function $q$ is monotone non-increasing.

Perhaps more than any other application, approximating quasi-concave functions privately has been successfully applied in private learning of thresholds or quantiles, in a fairly large body of works~\cite{KasiviswanathanLNRS08, BeimelNS13a, BeimelNS13b, FeldmanX14, BunNSV15, AlonLMM19, KaplanLMNS20}. Afterall, the function \[q(t) = -\left| \frac n 2 - |\{x\in P:~x\leq t\}| \right|\] is a quasi-concave function that allows one to approximate the median of a given dataset $P$. We thus take the liberty to convert previous papers discussing private quantile-approximation to works that approximate privately any quasi concave function. We thus summarize the results in the following theorem.

\begin{theorem}
	\label{thm:prelim_quasi_concave}
	Let $q$ be any function $q:\R\to\R$ satisfying (i) $q$ is quasi-convace, (ii) $q$ has global-sensitivity $1$ and (iii) for every closed interval $I$ one can efficiently compute $\max_{x\in I}q(x)$. Let $\G\subset\R$ be a grid of granularity $\Upsilon = 2^{-\upsilon}$, and denote $q^* = \max_{x\in \G}q(x)$. Then, for any $0<\beta <1/2$ there exist differentially private algorithms that w.p. $\geq 1-\beta$ return some $x\in \G$ such that $q(x)\geq q^*-\alphaqc$ where
	\[ \alphaqc(\epsilon,\delta,\beta) = \begin{cases}
	O(\frac{\upsilon+\log(1/\beta)}{\epsilon}), &~\textrm{using $\epsilon$-DP binary-search}\\
	\tilde O(\frac{\log(\upsilon/\beta\epsilon\delta) }{\epsilon}), &~\textrm{using the ``Between Thresholds'' Algorithm~\cite{BunSU17}}\\
	 O\left(\frac{8^{\log^{*}(\upsilon)}\log^{*}(\upsilon)}{\epsilon}\cdot \log(\frac{\log^{*}(\upsilon)}{\beta\delta})\right) &~\textrm{using the ``RecConvace'' algorithm~\cite{BeimelNS13b}}\\
	\end{cases}  \] 
\end{theorem}
The first bound is given by standard $\epsilon$-DP binary search algorithm (folklore). The second bound is given by the rather intuitive ``Between Threshold'' algorithm of Bun et al~\cite{BunSU17} where instead of the standard counting function $f(z)= |\{x\in P: x\leq z \}|$ we use the function $f(z) = \max\limits_{x\in (-\infty,z]}q(x) - \max\limits_{x\in [z,\infty)}q(x)$ and set thresholds close to $0$ (indicating a maximization point of $q$). The third is from the {\tt RecConcave} algorithm~\cite{BeimelNS13b} that deals with approximating quasi-concave function, a rather intricate algorithm. We comment that it is unknown\footnote{Uri Stemmer, private correspondence.} whether the more recent work~\cite{KaplanLMNS20} that improves upon the bounds of {\tt RecConcave} is applicable to general quasi-concave functions.

\section{Tools, Part 1: The Tukey-Depth Completion Function}
\label{sec:TDC}

In this section we discuss the implementation of the following \emph{Tukey Depth Completion} function. This function takes as a parameter a $i$-long tuple of coordinates $\bar y$, where $0\leq i < d$, and scores each $x\in \R$ with a value $\kappa$ if the $i+1$ prefix $\bar y \circ x$ can be completed to a point with Tukey-depth of $\kappa$. Formally, we present the following definition(s).

\begin{definition}
	\label{def:TDC}
	Fix $d\in \mathbb{N}$ and let $P$ be a collection of points in $\R^d$. For any $i$-tuple of coordinates $\bar y = (y_1, y_2, ..., y_i)$ where $0\leq i \leq d-1$ we define the function $\TDC_{\bar y}:\mathbb{R}\to \mathbb{R}$ by
	\begin{equation}
	\label{eq:def_TDC}
	\TDC^P_{\bar y}(x) = \max_{(z_1, z_2, ..., z_{d-1-i}) \in \R^{d-i-1}} \TD\big( (y_1,.., y_i, x, z_1,..,. z_{d-1-i}), P\big)
	\end{equation}
	For any closed interval $I = [a,b]\subset \mathbb{R}$ we overload the definition of $\TDC$ to denote
	\begin{equation}
	\label{eq:def_TDC_interval}
	\TDC^P_{\bar y}(I) = \max_{ x\in [a,b]} \TDC^P_{\bar y}(x)
	\end{equation}
	note that for $x\in \R$ it holds that $\TDC^P_{\bar y}(x) = \TDC^P_{\bar y}(I)$ for the closed (degenerate) interval $I = [x,x]$.
	
	Lastly, we introduce an additional variation that we will apply in our work. For any such $\bar y$ and any $\ell\in \R$ we denote the following function.
	\begin{equation}
	\label{eq:def_TDC_l}
	\lTDC_{\bar y}^{P}(x) = \min \{ \TDC^P_{\bar y}(x), \TDC^P_{\bar y}(x+\ell)\}
	\end{equation}
	and similarly, use $\lTDC_{\bar y}^{P}(I) = \max\limits_{x\in I} \lTDC_{\bar y}^{P}(x)$. We omit the superscript $P$ whenever the dataset is clear.
\end{definition}
\noindent It is worth noting that the definition of the $\TDC$-function is w.r.t.~the real Euclidean space and not just points on a grid. Later we discuss the refinement of the grid required for finding a point whose $\TDC$ w.r.t.~the grid $\G^d$ is equal.

We begin with a simple property of the $\TDC$ function, we show that the $\TDC$-function, as well as the $\lTDC$-function, are both quasi-concave functions. (The first part of the claim was proven in~\cite{BeimelMNS19}.)
\begin{proposition}
	\label{pro:TDC_is_quasy_concave}
	Fix $i$ and fix an $i$-tuple $\bar y$. For any $a\leq b\leq c$ on the grid we have that 
	\[  \TDC_{\bar y}(b) \geq \min\left\{ \TDC_{\bar y}(a), \TDC_{\bar y}(c)  \right\}  \]
	In addition for any $\ell >0$ we also have that
	\[  \lTDC_{\bar y}(b) \geq \min\left\{ \lTDC_{\bar y}(a), \lTDC_{\bar y}(c)  \right\}  \]
\end{proposition}
\begin{proof}
	If $b=a$ or $b=c$ then the claim trivially holds. Assuming $a<b<c$, denote $b = \lambda a + (1-\lambda) c$ for the suitable scalar $\lambda\in (0,1)$. Denote $\kappa = \min\{ \TDC_{\bar y}(a), \TDC_{\bar y}(c) \}$. Let $\bar z_a$ and $\bar z_c$ be the two completions such that $\TD\Big( \bar y \circ a \circ \bar z_a \Big), \TD\Big( \bar y \circ c \circ \bar z_c \Big) \geq \kappa$. It follows that the two points $p_a = \Big( \bar y \circ a \circ \bar z_a \Big)$ and $p_c=\Big( \bar y \circ c \circ \bar z_c \Big)$ are in $D_P(\kappa)$. Due to the convexity of $D_P(\kappa)$ it holds that the point $p_b = \lambda p_a + (1-\lambda)p_c$ also belongs to $D_P(\kappa)$. Since the $i+1$ coordinate of $p_b$ is $b$ it follows that $\TDC_{\bar y}(b) \geq \kappa$. 
	
	As for the function $\lTDC_{\bar y}$, note that by definition
	\[ \min\left\{ \lTDC_{\bar y}(a), \lTDC_{\bar y}(c)  \right\} = \min\{ \TDC_{\bar y}(a), \TDC_{\bar y}(c), \TDC_{\bar y}(a+\ell), \TDC_{\bar y}(c+\ell) \} \stackrel{\rm def} = \rho  \]
	And since we have shown the quasi-concavity of $\TDC$ then we have that
	\begin{align*}
	\TDC_{\bar y}(b) &\geq \min\{\TDC_{\bar y}(a), \TDC_{\bar y}(c)\}  \geq \rho
	\intertext{and similarly that}
	\TDC_{\bar y}(b+\ell) &\geq \min\{\TDC_{\bar y}(a+\ell), \TDC_{\bar y}(c+\ell)\}  \geq \rho
	\intertext{so it holds that}
	\lTDC_{\bar y}(b) &= \min\{ \TDC_{\bar y}(b),\TDC_{\bar y}(b+\ell)  \}  \geq \rho = \min\left\{ \lTDC_{\bar y}(a), \lTDC_{\bar y}(c)  \right\}
	\end{align*}
\end{proof}

Having established the quasi-concavity of $\TDC_{\bar y}$ it follows that on the real line the values of the function ascend from $0$ to the max-value, then descend back to $0$. In particular, for any (integer) $\kappa$ from $0$ to the max-value of the $\TDC_{\bar y}$-function, there exists an interval $[a_\kappa,b_\kappa]$ such that for any $x\in [a_\kappa,b_\kappa]$ it holds that $\TDC_{\bar y}(x)\geq \kappa$. And so, we give an algorithm that finds these sets of nested intervals $\{ [a_\kappa,b_\kappa] \}_{\kappa > 0}$, and then finds the maximum $\kappa$ whose interval intersect the given point $x$ or interval $I$.

\begin{algorithm}[ht]
	\caption{\label{alg:compute_TDC}Pre-processing for Computing the $\TDC$-function}
	{\bf Input}: $P\subset\mathbb{R}^{d}$; an $i$-tuple $\bar y = (y_1, y_2,.., y_{i})$.\\
	{\bf Output}: A collection of nested intervals.
	\begin{algorithmic}[1]
		\FOR {\textbf{each $\kappa$ from $1$ to $\frac n 2$}}
		\STATE Compute the $\kappa$-Tukey region $D_P(\kappa)$.
		\STATE Compute the intersection $D_P(\kappa) \cap S_{\bar y}$ where $S_{\bar y}$ is the affine subspace $S_{\bar y} = \{\bar y \circ \bar x:~ \bar x\in \R^{d-i}\}$.
		\STATE \textbf{if} ($D_P(\kappa) \cap S_{\bar y}=\emptyset$) \textbf{then break for-loop}
		%\\		(For each $H$ in the collection ${\cal H}_{\kappa}$ that makes the convex polytope $D_P(\kappa)$)
		\STATE Compute $a_\kappa$ and $b_{\kappa}$ where
		
		\begin{minipage}[h]{0.35\textwidth}
			\begin{align*}
			a_\kappa &= \min \langle \bar e_{i+1}, \bar y\circ \bar x\rangle \\
			\textrm{ s.t. }& \bar x \in D_P(\kappa) \cap S_{\bar y}
			\end{align*}
		\end{minipage}~and~
		\begin{minipage}[h]{0.35\textwidth}
		\begin{align*}
		b_\kappa &= \max \langle \bar e_{i+1}, \bar y\circ \bar x\rangle\\
		\textrm{ s.t. }& \bar x\in D_P(\kappa) \cap S_{\bar y}
		\end{align*}
		\end{minipage}
		\ENDFOR
		\RETURN the collection $\{[a_\kappa,b_\kappa]\}$
	\end{algorithmic}
\end{algorithm}
\begin{theorem}
	\label{thm:alg_TDC_works}
	For any $0\leq i <d$ and any $i$-tuple $\bar y$, Algorithm~\ref{alg:compute_TDC} runs in time polynomial in $n$ and $\upsilon$ and exponential in $d$ and returns a collection of nested intervals such that for any interval $I\subset\R$ it holds that $\TDC_{\bar y}(I) = \max \{\kappa: I\cap [a_\kappa,b_\kappa]\neq \emptyset\}$.
\end{theorem}
\begin{proof}
	Since there are at most $n/2$ possible intervals returned by the algorithm and each is formed by solving a LP in $n^{{\rm poly}(d)}$ constraints, the claim regarding the runtime of the algorithm holds. In fact, it is known that the exact computation of the Tukey region takes $T(n,d)=O(n^{(d-1)\lfloor d/2 \rfloor })$ time, a computation that returns a set of $\leq T(n,d)$ vetrices that form the boundary of the $D_P(\kappa)$-polytope. Thus, solving each LP involves $d-i$ variables with $T(n,d)$ many constraints, so it takes at most $O((d-i)!T(n,d))$-time~\cite{Seidel91}.	
		
	Secondly, each $D_P(\kappa)$ is a convex polytope so its intersection with the (convex) affine subspace $S_{\bar y}$ is also a convex polytope. Now, since for any $\kappa<\kappa'$ we have that $D_P(\kappa')\subset D_P(\kappa)$ then~--- denoting $\bar x_{\kappa'}$ as the vector that obtains $a_{\kappa'}$~--- we have that $\bar x_{\kappa'}$ is a valid solution for the minimization LP for $D_P(\kappa)$ hence $a_{\kappa} \leq \langle \bar e_{i+1}, \bar y\circ \bar x_{\kappa'}\rangle = a_{\kappa'}$, and similarly, $b_{\kappa}\geq b_{\kappa'}$; thus $[a_{\kappa'},b_{\kappa'}]\subset [a_{\kappa},b_{\kappa}]$. This implies the set of intervals are nested in one another.
	
	Next fix any interval $I\subset \R$. We argue that for any $\kappa$ such that there exists some $x\in I$ that also falls inside the interval $[a_{\kappa},b_{\kappa}]$, we have that $\TDC_{\bar y}(x)\geq \kappa$, making $\TDC_{\bar y}(I)\geq \kappa$. The reason is the following: denote the vectors $\bar u = \big(\bar y \circ a_{\kappa} \circ \bar z \big)$ and $\bar v = \big(\bar y \circ b_{\kappa} \circ \bar z' \big)$ as two vectors in $D_P(\kappa)\cap S_{\bar y}$ that obtain $a_{\kappa}$ and $b_{\kappa}$. These vectors give that $\TDC_{\bar y}(a_\kappa)\geq \kappa$ and $\TDC_{\bar y}(b_\kappa)\geq \kappa$ so by quasi-concaveness we have that $\TDC_{\bar y}(x)\geq k$. This shows that $\TDC_{\bar y}(I)\geq \max \{\kappa: I\cap [a_\kappa,b_\kappa]\neq \emptyset\}$. Conversely, denote $\kappa^*= \TDC_{\bar y}(I)$, then for some $x\in I$ we have that $\TDC_{\bar y}(x) = \kappa^*$ which means that for some completion $\bar u = \big(\bar y \circ x \circ \bar z\big)\in D_P(\kappa^*) \cap S_{\bar y}$, and so we have that $a_{\kappa^*} \leq x \leq b_{\kappa^*}$ as $a_{\kappa^*}$ (resp. $b_{\kappa^*}$) is a solution for a minimization (resp. maximization) problem over a domain containing $\bar u$. Thus $I \cap [a_{\kappa^*},b_{\kappa^*}]\neq\emptyset$ making $\kappa^* \leq \max \{\kappa: I\cap [a_\kappa,b_\kappa]\neq \emptyset\}$. 
\end{proof}

Now, given the collection of intervals returned by Algorithm~\ref{alg:compute_TDC}, which is in essence a set of $\leq 2\kappa^*\leq n$ points 
\begin{equation}
\label{eq:nested_intervals}
a_1 \leq a_2 \leq ... \leq a_{\kappa^*}\leq b_{\kappa^*}\leq b_{\kappa^*-1} \leq ... \leq b_2\leq b_1
\end{equation} 
on the real line where the $\TDC_{\bar y}$-function changes its value, we argue that for any interval $I$ computing $\TDC_{\bar y}(I)$ is simple and can be done in $O(\log(\kappa^*)) = O(\log(n))$-time by the following scheme. Denoting $I = [p,q]$ we have:
\begin{itemize}
	\item If $q < a_{\kappa^*}$ then $I$ is contained in the part of the real line where $\TDC_{\bar y}$ is monotone non-decreasing, thus $\TDC_{\bar y}(I) = \TDC_{\bar y}(q)$, so using binary search we find $\kappa$ such that $a_{\kappa-1}<q\leq a_{\kappa}$ and return it.
	\item Symmetrically, if $p > b_{\kappa^*}$ then $I$ is contained in the part of the real line where $\TDC_{\bar y}$ is monotone non-increasing, thus $\TDC_{\bar y}(I) = \TDC_{\bar y}(p)$, so using binary search we find $\kappa$ such that $b_{\kappa}\geq p>b_{\kappa-1}$ and return it.
	\item Otherwise, $q \geq a_{\kappa^*}$ and $p\leq b_{\kappa^*}$ which means $I\cap [a_{\kappa^*},b_{\kappa^*}]\neq \emptyset$ and so we return $\kappa^*$.
\end{itemize}
Next, in order to compute $\lTDC_{\bar y}(x)$ we just compute $\TDC_{\bar y}(x), \TDC_{\bar y}(x+\ell)$ and take the min of the two values, so this takes $O(\log(n))$ time as well. 

Lastly, in order to compute $\lTDC_{\bar y}(I)$, we append the collection of $2\kappa^*$ change-points with the set $\{a_{\kappa}+\ell:~ 1\leq \kappa\leq k^*\}\cup \{b_{\kappa}-\ell:~ 1\leq \kappa\leq k^*\}$ and sort the $4\kappa^*$ points. This is a superset of all the change points of $\lTDC_{\bar y}$: it is clear that between any pair of consecutive points $(a',b')$ the function $\lTDC_{\bar y}$ takes the same value (because both $(a',b')$ and $(a'+\ell,b'+\ell)$ are contained in a pair consecutive points among the original points in Eq~\eqref{eq:nested_intervals}). We then compute the value of $\lTDC_{\bar y}$ on each interval using a representative $x\in (a',b')$ and omit from the $4\kappa^*$ any point in which the value of $\lTDC_{\bar y}$ doesn't change. As $\kappa^*\leq n/2$ then this takes $O(n\log(n))$. Once we have a sorted list of points on the real line where the value of $\lTDC_{\bar y}$-function change, we can now compute the value of the quasi-concave function $\lTDC_{\bar y}(I)$ in a similar fashion to computing $\TDC_{\bar y}$ in $O(\log(n))$-time.

\paragraph{Extension.} We comment that the above-algorithm works for any set of convex polytopes $C_1 \supset C_2 \supset C_3 \supset..$ with at most $n^{\poly(d)}$-vertices each. We will rely on the this fact later, when we work with \emph{projections} of the various Tukey-regions. However, one of the key uses to the $\TDC$-function we rely on is when we rotate directions so that the first axis aligns with a given direction $v$. In such a case, this is equivalent to rotating the set $P$, so we use the notation $\TDC_{\bar y}^{R_v(P)}$ and on occasion just $\TDC_{\bar y}^{R_v}$.

\paragraph{Grid Refinement.} This establishes that for any $0\leq i\leq d-1$ and any $\bar y$ there exists an efficient (with pre-processing time of $O(d!n^{d(d-1)/2})$ and query time of $O(\log(n))$) algorithm that computes $\TDC_{\bar y}(x)$ and $\lTDC_{\bar y}(x)$. But as Beimel et al~\cite{BeimelMNS19} noted, it is not a-priori clear that the coordinates of the completion lie on the same grid $\G^d$ we start with. 

To this we provide two answers. The first, which we prefer by far, is that we can keep using the same grid $\G$, and each time we find a point $p$ we instead of formally stating ``we find a point $p$ inside the convex body'' we use ``we find a point $p$ within distance $\sqrt d \Upsilon$ from a point inside the convex body.'' After all, our work already deals with approximations, so under the (rather benign) premise that the diameter of the convex body is sufficiently larger than $\Upsilon$, this little additive factor changes very little in the overall scheme.

The second answer is to use a refinement of the grid $\G$ into some $\G'$. This approach is described here, in order for our results to be comparable with the results of~\cite{BeimelMNS19,KaplanSS20} regarding finding a point inside the convex-hull. However, past this section we assume this refinement has already happened as a pre-processing step for our analysis and so we set $\G\gets \G'$ and continue with the remainder of the algorithms as is.

In order to construct the grid $\G'$, we begin with the observation of Kaplan et al~\cite{KaplanSS20} that for any $\kappa$, the vertices of $D_P(\kappa)$ lie on a grid $\tilde \G^d$ with granularity of $\tilde\Upsilon = d^{\frac{-d(d+1)}{2}}\Upsilon^{d^2}$. We also use the notation $\xi = 1/\tilde\Upsilon$. We argue inductively that when applying $\TDC$-sequentially to reveal the coordinates of a point inside the convex hull, we obtain a point $p$ whose $i$th coordinate lies on a grid of granularity lower bounded by $\prod_{j=1}^{i}(\tilde\Upsilon/\sqrt{j})^j $. Note, this makes our grid (much like the grid in~\cite{BeimelMNS19}) highly unbalanced: on the first axis it suffices to use a discretization of $\tilde \Upsilon$, but on the $d$-th axis we require a discretization of $(\tilde \Upsilon/\sqrt d)^{O(d^2)}$. 

The claim is proven inductively. Let $\{a_\kappa^i,b_\kappa^i\}_{\kappa=1}^{\kappa^*}$ be the collection of coordinates returned by Algorithm~\ref{alg:compute_TDC} in the process of computing $\TDC_{\bar y}(x)$ where for each $i$ the prefix $\bar y$ is precisely the first $i$ coordinates of $p$. Now, for $i=1$ it is clear that each $a_\kappa^1$ or $b_\kappa^1$ is a coordinate of some vertex of $D_P(\kappa)$ and so it has granularity $\tilde\Upsilon$. Thus it suffices to place the grid $\tilde \G$ on the $[0,1]$-interval find use a DP-algorithm that returns a point on this grid, and so the first coordinate $p_1$ has granularity $\tilde\Upsilon$. Now, for each $i+1$, the coordinates $a_\kappa^{i+1},b_\kappa^{i+1}$ are the $i+1$-coordinates of two vertices of $D_P(\kappa)\cap S_{(p_1,..,p_i)}$. For brevity we denote $\bar p = (p_1,..,p_i)$. Such a vertex is found when we take a $(i+1)$-facet of $D_P(\kappa)$, whose vertices we denote as $v^1, v^2,..., v^{i+1}$, and find the intersection of $S_{\bar p}$ with this facet. So we check if there exists a point in this facet whose first $i$ coordinates are $\bar p$ and if so, retrieve its $i+1$-coordinate. Namely, we see if $\bar p$ is a convex combination of the $i+1$ vectors which are the $i$-prefixes of the facet's vertices, denoted as $\bar v^1,..., \bar v^{i+1}$; if indeed for some convex combination $\bar p= \sum_{j=1}^{i+1}\lambda_j \bar v^{j}$ then the $i+1$-coordinate is $\sum_{j=1}^{i+1}\lambda_j v^{j}_{i+1}$. Finding this convex combination requires that we solve a system of $i+1$ linear constraints $M\bar\lambda = (\bar p,1)$ where each column of $M$ is composed of the $i$-dimensional vector $\bar v^j$ concatenated with $1$, and the RHS is composed of $\bar p$ concatenated with $1$. Thus $\bar \lambda = M^{-1}(\bar p,1)$, and the $i+1$ coordinate we are after is a dot-product of $ M^{-1}(\bar p,1)$ with the vector $\bar u=(v^1_{i+1}, v^2_{i+2},..., v^{i+1}_{i+1})$ whose coordinates lies on $\tilde\G$.

Note that $M$ is composed of coordinates that lie on $\tilde \G$ as it is composed of prefixes of vertices of $D_P(\kappa)$ and ones. Thus each coordinate of $M$ has granularity $\tilde \Upsilon$, and so $\xi M$ is an integer matrix with entries in $[0,\xi]$. By Hadamard's inequality, $\det(\xi M)\leq (\sqrt{i+1}\xi)^{i+1}$, and so, writing $(\xi M)^{-1}$ using the adjugate formula, each entry of $(\xi M)^{-1}$ can be written as a fraction with a denominator of not larger than $(\sqrt{i+1} \xi)^{i+1}$.
%; and so the same holds for each entry of $M^{-1}=\xi\cdot (\xi M)^{-1}$. 
By our induction hypothesis, each coordinate of $(p_1,..,p_i)$ can be written as a rational fraction with the same denominator, and the denominator doesn't exceed $\prod_{j=1}^{i}(\sqrt{j}\xi)^j$. Lastly, by definition each coordinate of $\bar u$ can be written as a fraction with denominator $\xi$, so $\xi\bar u$ is a vector of integers. This means that $\langle M^{-1}(\bar p,1),\bar u\rangle = \langle \xi\cdot(\xi M)^{-1}(p,1), \bar u\rangle = \langle (\xi M)^{-1}(\bar p,1),\xi \bar u\rangle$ can be written as a rational fraction where its denominator doesn't exceed  $\prod_{j=1}^{i+1} (\sqrt j\xi)^j$.

This proves that the level of discretization we require for any axis is bounded below by \linebreak $d^{-\frac{(d+1)(d+2)}{4}}\cdot\tilde \Upsilon^{-\frac{(d+1)(d+2)}{2}} = d^{-O(d^4)}\Upsilon^{\frac{d^4+3d^3+2d^2}{2}}\geq \Upsilon^{4d^4}$ (assuming $d^{-1}>\Upsilon$).

\paragraph{Summary.} Now that we refined the grid from $\G$ to $\G'$ with granularity $\Upsilon^{4d^4} = 2^{-\upsilon(4d^4)}$, we can apply any DP-algorithm that w.p.$\geq1-\beta$ returns a point on $\G'$ with roughly the same value of the maximal value. This gives a DP-algorithm that returns w.p.$\geq 1-\beta$ a point $x\in \G'$ with either $\TDC_{\bar y}$-value or $\lTDC_{\bar y}$-value which is $\alphaqc(\cdot,\cdot,\cdot)$-close to the max-possible value on the grid. Altogether, we have the following corollary.\footnote{Note that we have not bothered applying the advanced composition theorem~\cite{DworkRV10} since we assume $d$ is a small constant.}

\begin{corollary}
	\label{cor:DP_approx_TD}
	Fix $\epsilon>0$, $\delta\geq 0$, $\beta\in (0,1/2)$. There exists an efficient $(\epsilon,\delta)$-DP-algorithm, denoted {\tt DPPointInTukeyRegion}, that takes as input a dataset $P$ and a parameter $\kappa$ where $D_P(\kappa)\neq\emptyset$ and w.p. $\geq1-\beta$ returns a point $\bar x\in (\G')^d$ whose Tukey-depth is at least
	\begin{equation}
	\kappa-d\alphaqc(\frac \epsilon d, \frac \delta d, \frac \beta d)\geq \frac{n}{d+1} - d\alphaqc(\frac \epsilon d, \frac \delta d, \frac \beta d)
	\end{equation}
	In particular, for any $\kappa\geq0$ we return a point of Tukey-depth $\geq \kappa$ provided $n = \Omega(d\kappa+d^2\alphaqc(\frac \epsilon d, \frac \delta d, \frac \beta d) )$
	\begin{equation}
	 = \begin{cases}
	\Omega(d\kappa + d^3\frac{d^4\upsilon+\log(d/\beta)}{\epsilon}), &~\textrm{Using the $\epsilon$-DP binary-search}\\
	\tilde \Omega(d\kappa + d^3\frac{\log(d\upsilon/\beta\epsilon\delta) }{\epsilon}), &~\textrm{Using the ``Between Thresholds'' algorithm}\\
	\Omega(d\kappa + d^3\frac{8^{\log^*(d\upsilon)}\log^*(d\upsilon)\cdot \log(d\log^*(\upsilon)/\delta\beta) }{\epsilon}), &~\textrm{Using the ``RecConcave'' algorithm}\\
	\end{cases}
	\end{equation} 
\end{corollary}
Again, we comment that quantitatively, the results are just as those obtained by~\cite{BeimelMNS19}. The key improvement of our work is the runtime which decreases from $\poly(1/\Upsilon)$ to $\poly(\upsilon)$. 

Similarly, we also obtain the following corollary.
\begin{corollary}
	\label{cor:DP_approx_TD_with_ell}
	Fix $\epsilon>0$, $\delta\geq 0$, $\beta\in (0,1/2)$ and also $\ell>0$. Denote
	\[ \kappa^*  = \max\{1\leq \kappa \leq \frac n 2:~ \exists p^1, p^2\in D_P(\kappa) \textrm{ where their first coordinates satisfy } p^2_1-p^1_1 \geq \ell \} \]
	Given $\ell$, there exists a $(\epsilon,\delta)$-DP-algorithm that w.p. $\geq1-\beta$ returns a pair of points $\bar x,\bar y\in (\G')^d$ s.t. $y_1-x_1 = \ell$ and where the Tukey-depth of both $\bar x$ and $\bar y$ is at least $\kappa^* -d\alphaqc(\frac \epsilon {2d-1}, \frac \delta {2d-1}, \frac \beta {2d-1})$ 
	\[ = 
	\kappa^* - 	\begin{cases}
	 O(d^2\frac{d^4\upsilon+\log(d/\beta)}{\epsilon}), &~\textrm{Using the $\epsilon$-DP binary-search}\\
	\tilde O(d^2\frac{\log(d\upsilon/\beta\epsilon\delta) }{\epsilon}), &~\textrm{Using the ``Between Thresholds'' algorithm}\\
	O(d^2\frac{8^{\log^*(\upsilon)}\log^*(\upsilon)\cdot \log(d\log^*(\upsilon)/\delta\beta) }{\epsilon}), &~\textrm{Using the ``RecConcave'' algorithm}\\
	\end{cases}\]
\end{corollary}
The idea behind Corollary~\ref{cor:DP_approx_TD_with_ell} is that we first find $x_1$ using our DP-algorithm for approximating $\lTDC$, and set the first coordinate of $\bar x$ to be $x_1$ whereas the first coordinate of $\bar y$ is set as $x_1+\ell$. We then continue and find the rest of the coordinates of $\bar x$ one by one, and the same for $\bar y$. Since we run the algorithm for $\lTDC$ once, and run the $\TDC$ algorithm twice for each of the $(d-1)$ remaining coordinates of $\bar x$ and $\bar y$, we divide the privacy budget by $2d-1$ per each execution of the algorithm.

\paragraph{Comment.} Note that, as mentioned above, in the reminder of the paper we either avoid refining the grid any further and rely on an additive $\Upsilon$-approximation, or alternatively refine the grid and apply the rest of the algorithms in this work after setting $\upsilon \gets (d^4+1)\upsilon$, namely setting $\upsilon$ as the $\log$-of the new grid size.

\section{Tools, Part 2: Approximating the Diameter \& Width of a Tukey-Region}
\label{sec:approx_diam_and_width}

\subsection{The Diameter}
\label{subsec:approx_diam}

Recall, given the dataset $P$, we denote the $\kappa$-Tukey region as $D_P(\kappa)$ and in this section our goal is to approximate the diameter of $D_P(\kappa)$, defined as 
$\diam_\kappa = \max_{a, b\in D_P(\kappa)}\|b-a\|$.
Yet, it is clear that the diameter, as well as other properties (such as the volume, width, etc.) of the $\kappa$-Tukey region are highly sensitive to the presence or absence of a single datum. Thus, our work returns an approximation $\ell$ which is a $(\alpha,\Delta)$-approximation, in the sense that
\begin{equation}
\label{eq:alpha_delta_approx_of_diam}
(1-\alpha) \diam_{\kappa} \leq \ell \leq \diam_{\kappa-\Delta}
\end{equation}
Clearly, since $D_P(\kappa) \subset D_P(\kappa-\Delta)$ then $\diam_{\kappa}\leq \diam_{\kappa-\Delta}$; yet the question whether $\diam_{\kappa-\Delta}$ is comparable to $\diam_{\kappa}$ or not is data-dependent. (Obviously, we comment that $\ell/(1-\alpha)$ is an upper bound on $\diam_{\kappa}$, a fact we occasionally require.)

In order to find such a diameter-approximation, we leverage on the idea of discretizing all possible directions, which is feasible in constant-dimension Euclidean space. We rely on a $\zeta$-angle cover of the unit-sphere, $V_\zeta$, for a suitably chosen $\zeta$.  Specifically, we use the following property.
\begin{proposition}
	\label{pro:discretization}
	For any $\zeta < 1/2$ and for any set $P\subset \R^d$ we have that
	\[ (1-\zeta^2)\diam(P)\leq  \max_{v\in V_\zeta} \max_{a,b\in P} \langle b-a,v\rangle \leq \diam(P)  \]
\end{proposition}
\begin{proof}
	On the one hand, for any $v\in V_\zeta$ and any $a,b\in P$ we have that
	$ \langle b-a,v\rangle \leq \|b-a\| \leq \diam(P)$
	so clearly the maximum over all $V_\zeta$ and all pairs of points in $P$ doesn't exceed this upper bound. On the other hand, denoting $a$ and $b$ as the two points in $P$ that obtain its diameter, and denoting
	$u_{ab}$ as the direction of the straight-line going from $a$ to $b$, we know that there exists a direction $v\in V_\zeta$ whose angle with $u_{ab}$ is at most $\zeta$, thus  
	\begin{equation*}  \langle  b-a,v \rangle  =  \langle  \big(a+ \|b-a\|\cdot u_{a,b}\big)-a,v \rangle  = \|b-a\| \cdot \langle u_{ab}, v\rangle \geq \diam(P)\cos(\zeta) \stackrel{\rm Taylor}\geq (1-\zeta^2)\diam(P)
	\label{eq:diam_projection}
	\end{equation*}
	hence the maximum is at least this lower bound.
\end{proof}

Based on the discretization $V_{\zeta}$, our approximation is fairly straight-forward, as it uses the Sparse-Vector Technique (SVT). For each $\ell$ we pose the query
\begin{equation}
q_P(\ell) = \max_{v\in V_{\zeta}} \max_{x\in \R} \lTDC^{R_v(P)}(x)
\end{equation}
where $R_v$ is a rotation that sets $v$ as the first vector basis, namely $v\stackrel{R_v}\mapsto e_1$, and $R_v(P) = \{R_v(p): ~ p \in P\}$. In other words, we rotate the standard basis so that the projection onto $v$ becomes the first coordinate, then run the query $\lTDC\big( (-\infty, \infty) \big)$.

\begin{algorithm}[ht]\caption{\label{alg:ApproximateDiameter}{\tt DPTukeyDiam} Approximate Tukey-Region Diameter}
	{\bf Input}: $P\subset \G^{d}\subset[0,1]^d$ of a given size $n$; privacy loss $\epsilon>0$ approximation parameters $\alpha,\beta >0$; Tukey depth parameter $\kappa \geq 0$.
	%\\{\bf Output}:
	\begin{algorithmic}[1]
		\STATE Set $\zeta = \sqrt{{\alpha/2}}$ and $T = \lceil \frac {2\upsilon+\ln(d)} {\alpha}\rceil$.
		\STATE Sample $X\sim \Lap{\tfrac 3 \epsilon}$.
		\STATE Set the seqeunce of $T+1$ lengths $\ell_i = \sqrt{d}(1-\alpha/2)^i$ for $i=0,1,2..., T$.
		\STATE Iterate on $i$ from $0$ to $T$. For each $i$ sample $Y_i \sim \Lap{\tfrac 3 \epsilon}$. Halt on the first $i$ satisfying \[Y_i +q_P(\ell_i) = Y_i + \max_{v\in V_{\zeta}} \max_{x\in \R} \ell_i\mhyphen\TDC^{R_v(P)}(x)  \geq \kappa - \tfrac{6}{\epsilon}\log((T+2)/\beta) + X\]
		\RETURN $\ell_i$ if halted on some $i$ and $0$ otherwise.
	\end{algorithmic}
\end{algorithm}

\begin{theorem}
	\label{thm:DP_alg_approx_diam_Tukey_region}
	Algorithm~\ref{alg:ApproximateDiameter} is a $\epsilon$-DP algorithm that w.p. $\geq 1-\beta$ returns a value $\ell$ which is $(\alpha,\Delta)$-approximation of $\diam_{\kappa}$ for $\Delta^{\rm diam}(\epsilon,\beta) = \frac{12\log((T+2)/\beta)}{\epsilon} = O(\frac{\log((\upsilon+\log(d))/\alpha\beta)}{\epsilon})$.
\end{theorem}
\begin{proof}
	First, Algorithm~\ref{alg:ApproximateDiameter} is $\epsilon$-DP since it applies the SVT over $T+1$ queries of global sensitivity $1$. (Since for any $x$ we have that $\lTDC(x)$ has global sensitivity of $1$, then a maximum of such queries over a fixed set also has global-senstivity of $1$, see~\cite{BlumLR08}.) Second, note that w.p. $\geq 1-\beta$ it holds that all the random variables in the SVT never exceed $\frac{3\log((T+2)/\beta)}{\epsilon}=\frac{\Delta} 4$ in magnitude. Under this event, since our halting condition is that the noisy answer of the query $\geq \kappa-\frac{\Delta}2$ it follows that upon reaching a query where $q_P(\ell)\geq \kappa$ we must halt, and for the query we halt on it must be that $q_P \geq \kappa-\Delta$.
	
	Now, consider any $\ell$ such that $\ell \leq (1-\alpha/2) \diam_\kappa$, and note that for the two points $a,b\in D_P(\kappa)$ obtaining $\diam_{\kappa}$ and for some $v\in V_\zeta$ we have $\langle b - a,v\rangle \geq (1-\zeta^2)\diam_{\kappa}\geq (1-\alpha/2)\diam_{\kappa}\geq \ell$. Thus, it must hold that $q_P(\ell) \geq \kappa$. Thus, if we denote $i_0 = \min\{i\in \N:~\ell_i \leq (1-\alpha/2) \diam_\kappa\}$ then $q_P(\ell_{i_0} )\geq \kappa$ and so  we halt at some $i\leq i_0$. By the minimality of $i_0$ we have that $\sqrt d(1- \alpha/2)^{i_0}\leq (1-\alpha/2)\diam_{\kappa} < \sqrt d(1-\alpha/2)^{i_0-1}$ and so we return $\ell_i \geq \ell_{i_0} = \sqrt d(1- \alpha/2)^{i_0} > (1-\alpha/2)^2\diam_\kappa>(1-\alpha)\diam_{\kappa}$.
	
	Conversely, for the $i$ on which we do halt we have that  $q_P(\ell_i) \geq \kappa - \Delta$. It follows that there exists two points $a',b'$, both of Tukey-depth at least $\kappa-\Delta$ whose projection over some $v\in V_\zeta$ is $\geq \ell_i$. But since $\langle b'-a', v\rangle \leq \|b'-a'\|\leq \diam_{\kappa-\Delta}$ then we have that $\diam_{\kappa-\Delta} \geq  \ell_i$
\end{proof}

\subsection{The Width}
\label{subsec:dp_approx_width}

We now turn our attention to the width estimation of the Tukey region $D_P(\kappa)$. Informally, the width of a set is the smallest ``sandwich'' of parallel hyperplanes that can hold the entire set; namely~--- of all pairs of parallel hyperplanes that bound the given set we pick the closest two, and the gap between them is the set's width. Formally, $\width_\kappa = \min_{v:~\|v\|=1} \max_{a,b\in D_P(\kappa)} \left| \langle b,v\rangle - \langle a,v\rangle   \right|$. Much like the in the case of the diameter, $\width_\kappa$ can also be drastically effected by the presence or absence of a single datum. Thus, our private approximation gives a $(\alpha,\Delta)$-approximation of the width, where we return a value $w$ such that
\[ (1-\alpha)\width_{\kappa} \leq w \leq (1+\alpha)\width_{\kappa-\Delta}  \]

Non-private width estimation is tougher problem than diameter estimation, and involves solving multiple LPs~\cite{DuncanGR97}. It is tempting to think that, much like the approach taken in Section~\ref{subsec:approx_diam}, a similar discertization/cover of all directions ought to produce a $(1+\alpha)$-approximation of the width. Alas, this approach fails when the width is very small, and in fact smaller or proportional to the discretization level. Somewhat surprisingly, the contra-positive is also true~--- when the discretization is up-to-scale, then we can easily argue the correctness of the discretization approach.

\begin{proposition}
	\label{pro:approx_width_via_discretization}
	Fix any $\alpha >0$. Given a set $P\subset \R^d$ with diameter $D$ and width $w$, if we set $\zeta \leq  \min\{\frac{\alpha w}{\sqrt 2 D},~1/2\}$ and take $V_\zeta$ as a $\zeta$-angle cover of the unit-sphere, then we have that
	\[ w  \leq \min_{v\in V_\zeta}\max_{a,b\in P} \langle b-a,v\rangle \leq (1+\alpha) w   \] 
\end{proposition}
\begin{proof}
	Since $V_\zeta \subset \mathbb{S}^{d-1}$ then obviously
	\[ w = \min_{v\in \mathbb{S}^{d-1}}\max_{a,b\in P} \langle b-a,v\rangle \leq \min_{v\in V_\zeta}\max_{a,b\in P} \langle b-a,v\rangle  \]
	Now, let $v$ be the direction on which the width of $P$ is obtained, i.e. $w = \max_{a,b\in P}\langle b-a,v\rangle$. Denote $u\in V_z$ as a vector whose angle with $v$ is at most $\zeta$, which by Proposition~\ref{pro:nearest_vector_in_angle_cover} is of distance $\leq \sqrt\zeta$ to $v$. This implies that for any $a,b\in P$ it holds that
	\[ \langle b-a,u\rangle =  \langle b-a,v\rangle - \langle b-a,v-u\rangle \leq w + \|b-a\|\cdot \|v-u\| \leq w + D \cdot \sqrt 2 \zeta \leq w+\alpha w  \]
	Thus $\max_{b,a\in P}\langle b-a,u\rangle \leq (1+\alpha)w$ implying that $\min_{v\in V_\zeta}\max_{a,b\in P} \langle b-a,v\rangle\leq (1+\alpha)w$.
\end{proof}

Following Proposition~\ref{pro:approx_width_via_discretization} we present our private approximation of $\width_\kappa$. This approximation too leverages on the query $\lTDC$ for a decreasing sequence of lengths $\ell_1 > \ell_2 > ...$, however, as opposed to Algorithm~\ref{alg:ApproximateDiameter}, with each smaller $\ell$ we also use a different discretization of the unit sphere. Details appear in Algorithm~\ref{alg:DP_approx_width_tukey_region}.

\begin{algorithm}[hbt]\caption{\label{alg:DP_approx_width_tukey_region}{\tt DPTukeyWidth} Approximate Tukey-Region Width}
	{\bf Input}: $P\subset \G^{d}\subset[0,1]^d$ of a given size $n$; privacy loss $\epsilon>0$ approximation parameters $\alpha,\beta >0$; Tukey depth parameter $\kappa \geq 0$; an upper-bound $D$ on the diameter of $D_P(\kappa)$ and a lower bound $B$ on the width of $D_P(\kappa)$.
	%\\{\bf Output}:
	\begin{algorithmic}[1]
		\STATE Set $T = \lceil \frac {2\ln(D/B)} {\alpha}\rceil$.
		\STATE Set the sequence of $T+1$ lengths $\ell_i = D(1-\alpha/2)^i$ for $i=0,1,2..., T$.
		\STATE Sample $X\sim \Lap{\tfrac 3 \epsilon}$.
		\FOR {$i$ {\bf from} $0$ {\bf to} $T$} 
		\STATE Set $\zeta =\min\{\frac{\alpha \ell_i}{4 D} , 1/2\} $ and $V_\zeta$ as the $\zeta$-angle cover of the unit-sphere.
		\STATE Denote $q_P(\ell_i) = \min\limits_{v\in V_\zeta}\max\limits_{x\in \R} \ell_i\mhyphen\TDC^{R_v(P)}(x)$ where $R_v$ is a rotation that sets $v$ as the first vector basis, namely $v\stackrel{R_v}\mapsto e_1$.
		\STATE Sample $Y_i \sim \Lap{\tfrac 3 \epsilon}$ and  {\bf break loop if}  
		$Y_i + q_P(\ell_i)  \geq \kappa - \frac{6}{\epsilon}\log((T+2)/\beta) + X$
		\ENDFOR
		\RETURN$\ell_i$ if halted on some $i$ and $0$ otherwise.
	\end{algorithmic}
\end{algorithm}
\begin{theorem}
	\label{thm:DP_alg_approx_width_Tukey_region}
	Algorithm~\ref{alg:DP_approx_width_tukey_region} is a $\epsilon$-DP algorithm that w.p. $\geq 1-\beta$ returns a value $\ell$ which is $(\alpha,\Delta)$-approximation of $\width_{\kappa}$ for $\Delta^{\rm width}(\epsilon,\beta) = \frac{12\log((T+2)/\beta)}{\epsilon} = O(\frac{\log((\upsilon+\log(d))/\alpha\beta)}{\epsilon})$.
\end{theorem}
In the statement of Theorem~\ref{thm:DP_alg_approx_width_Tukey_region} we use the na\"ive upper bound of $\diam_\kappa \leq \sqrt d$ and na\"ive lower bound of $B\geq \Upsilon$. 
Prior to proving the theorem, we need to establish two properties of the query $q_P(\ell_i)$ used by Algorithm~\ref{alg:DP_approx_width_tukey_region}.
\begin{claim}
	\label{clm:approx_width_query}
	Fix any $P\subset[0,1]^d$, any $\kappa\geq 0$, any $\ell>0$, any $D$ where $D\geq \diam(D_P(\kappa))$, any $\zeta < \frac 1 2$ and any $V_\zeta$ which is a $\zeta$-angle cover of the unit sphere. Then for the query $q_P(\ell_i) = \min\limits_{v\in V_\zeta}\max\limits_{x\in \G} \ell_i\mhyphen\TDC^{R_v(P)}(x)$ we have that (i) if $\width(D_P(\kappa)) \geq \ell$ then $q_P(\ell) \geq \kappa$; and (ii) if $\width(D_P(\kappa)) \leq (1-\alpha)\ell$ and $\zeta\leq\frac{\alpha\ell}{2 D}$ then $q_P(\ell) < \kappa$.
\end{claim}
\begin{proof}
	Clearly, if $\width(D_P(\kappa)) \geq \ell$ then due to the convexity of $D_P(\kappa)$, in any direction $v$ on can find two $a,b\in D_P(\kappa)$ where $\langle b-a,v\rangle=\ell\leq \width_{\kappa}$. Setting $x = \langle a,v\rangle$, we have that $\lTDC^{R_v(P)}(x) \geq \kappa$ and so $q_P(\ell)\geq \kappa$.
	
	Conversely, suppose $\width(D_P(\kappa)) \leq (1-\alpha)\ell$ and denote $u$ as the direction on which the width is obtained. Let $v\in V_\zeta$ be a vector whose angle with $u$ is at most $\zeta$. We argue that $\max_{a,b\in D_P(\kappa)} \lTDC^{R_v(P)}< \kappa$ which shows $q_P(\ell)<\kappa$. ASOC that there does exist a pair of points $a,b\in D_P(\kappa)$ such that $\langle b-a,v\rangle =\ell$. By the same argument as in the proof of Proposition~\ref{pro:approx_width_via_discretization}, we have that
	\[  \width_\kappa\geq \langle b-a,u\rangle = \langle b-a,v\rangle+\langle b-a,u-v\rangle \geq  \ell  - \|b-a\|\|u-v\| \geq \ell - D\cdot \sqrt 2\zeta \geq \ell - \frac{\alpha \ell}{\sqrt 2} > (1-\alpha)\ell \]
	which contradicts the assumption that $\width(D_P(\kappa))\leq (1-\alpha)\ell$.
\end{proof}
\begin{proof}[Proof of Theorem~\ref{thm:DP_alg_approx_width_Tukey_region}.]
	Much like in the proof of Theorem~\ref{thm:DP_alg_approx_diam_Tukey_region}, it is evident that our algorithm is $\epsilon$-DP since it applies the SVT over $T+1$ queries of global sensitivity $1$. Also, note that w.p. $\geq 1-\beta$ it holds that all the random variables in the SVT never exceed $\frac{3\log((T+2)/\beta)}{\epsilon}=\frac{\Delta} 4$ in magnitude. We continue our proof under the assumption this event hold, and since our algorithm adds noise to threshold of $\geq \kappa-\frac{\Delta}2$ it follows that upon reaching a query where $q_P(\ell)\geq \kappa$ we must halt, and for the query we halt on it must be that $q_P \geq \kappa-\Delta$.
	
	Based on Claim~\ref{clm:approx_width_query}, we have that by iteration $i_0 = \min\{i:~ \ell_i \leq \width(D_P(\kappa))\}$ we must halt, and so we return $\ell_i \geq \ell_{i_0} \geq (1-\alpha/2)\width(D_P(\kappa))$. Similarly, denoting the query on which we halt as $i$, then we have that if it were the case that $(1-\alpha/2)\ell_i \geq \width(D_P(\kappa-\Delta))$ then the value of the query is $<\kappa-\Delta$ and we would continue. Thus $(1-\alpha/2)\ell_i < \width(D_P(\kappa-\Delta))$ implying $\ell_i <\width(D_P(\kappa-\Delta))/(1-\alpha/2) < (1+\alpha)\width_{\kappa-\Delta}$ Thus, under our event (of bounded random noise) we return a $(\alpha,\Delta)$-approximation of the width of $D_P(\kappa)$.
\end{proof}

\paragraph{On the runtime of our algorithms.} Denoting $R_n$ as the runtime of executing the $\max_x\lTDC(x)$ query and using $T$ to denote the number of queries used in the SVT, it is fairly straight-forward that Algorithm~\ref{alg:ApproximateDiameter}  can be implemented in time $O(T|V_\zeta|)\cdot R_n = \tilde O(\frac{\upsilon}{\alpha}\cdot \alpha^{-\frac{d-1} 2} \cdot R_n)$. (Also, the algorithms in the following subsection are even easier to implement than the diameter-approximation algorithm and so they are also efficient.) Algorithm~\ref{alg:DP_approx_width_tukey_region} however requires we refine the discretization $V_\zeta$ with each iteration. In the extreme case where we only rely on the na\"ive lower bound of $B\geq \Upsilon$ and we indeed reach the last iteration $T$, the refinement we use is smaller than $\Upsilon$ making the runtime of the algorithm $\poly(1/\Upsilon)$ rather than $\poly(\upsilon)$. That is why in our work we rely on having a particular lower bound, of either $1/2d\cdot 5^d~(d!)$ or $1/(4d^{5/2}\cdot 5^d \cdot (d!))$. The reason for these particular bounds will become clear in later sections (specifically, Section~\ref{subsec:bounding_box_to_fatness}).
%However, the same analysis of the algorithm follows through if instead of using $q_P(\ell)$ we rely on the query $\tilde q_P(\ell) = \max\{ q_P(\ell), \kappa-\Delta-1\}$, which also has sensitivity of $1$ (thus the privacy preserving property isn't affected by replacing $q$ with $\tilde q$). Now, as Claim~\ref{clm:approx_width_query} suggests, \os{COMPLETE RUNTIME DISCUSSION!}

\subsection{Additional Tools: Max Projection and Large-Depth Direction}
\label{subsec:dp_directional_width}

Next, we give two similar algorithms for two particular tasks we will require later. The two algorithms may be of independent interest, although they do not provide an approximation of a well-studied quantity such as the diameter or width of a convex-set. These two algorithms are based on the SVT and they are both even simpler than the algorithm for approximating the diameter. We thus omit their proofs and merely describe them and state their correctness. 

\paragraph{Max-Projection.} First we deal the problem of approximating the max projection along any fixed direction $v$ of $D_P(\kappa)$. The algorithm for approximating the max-projection on a given direction is remarkably similar to Algorithm~\ref{alg:ApproximateDiameter} and in fact, is even simpler. Its guarantee is also similar: it returns a $(\alpha,\Delta)$-approximation of the length of the longest projection from a given point $p$ in direction $v$. Namely, it returns a number $\ell$ satisfying
$(1-\alpha)\max_{x\in D_P(\kappa)} \langle x-p,v\rangle \leq \ell\leq \max_{x\in D_{p}(\kappa-\Delta)}\langle x-p,v\rangle $.
The algorithm and its correctness are stated below. Note that the algorithm requires some a-priori knowledge about $D_P(\kappa)$~---not only does it need to be provided a point $p$ inside $D_P(\kappa)$, it also requires an upper-bound $D$ on $\diam_{\kappa}$. 

\begin{algorithm}[hbt]\caption{\label{alg:dp_approx_max_proj}{\tt DPMaxProjection:} Approximate Tukey-Region Max-Projection along a Given Direction}
	{\bf Input}: $P\subset \G^{d}\subset[0,1]^d$ of a given size $n$; privacy loss $\epsilon>0$, approximation parameters $\alpha,\beta >0$; Tukey-depth parameter $\kappa \geq 0$; a given direction (unit-length vector)~$v$; a point $p\in D_P(\kappa)$; an upper bound $D$ on the diameter of $D_P(\kappa)$.
	%\\{\bf Output}:
	\begin{algorithmic}[1]
		\STATE $T = \lceil \frac {2\upsilon+2\ln(D)} {\alpha}\rceil$.
		\STATE Sample $X\sim \Lap{\tfrac 3 \epsilon}$.
		\STATE Set the sequence of $T+1$ lengths $\ell_i = D(1-\alpha/2)^i$ for $i=0,1,2..., T$.
		\STATE Compute $x\gets \langle p,v\rangle$
		\STATE Iterate on $i$ from $0$ to $T$. For each $i$ sample $Y_i \sim \Lap{\tfrac 3 \epsilon}$. Halt on the first $i$ satisfying \[ Y_i +  \TDC^{R_v(P)}(x+\ell_i)  \geq \kappa - \tfrac{6}{\epsilon}\log((T+2)/\beta) + X\]
		where $R_v$ is a rotation that sets $v$ as the first vector basis, namely $v\stackrel{R_v}\mapsto e_1$.
		\RETURN $\ell_i$ if halted on some $i$ and $x$ otherwise.
	\end{algorithmic}
\end{algorithm}
\begin{theorem}
	\label{thm:dp_approx_max_proj}
	Algorithm~\ref{alg:dp_approx_max_proj} is a $\epsilon$-DP algorithm that w.p. $\geq 1-\beta$ returns a value $\ell$ which satisfied \[(1-\alpha)\max_{x\in D_P(\kappa)} \langle x-p,v\rangle \leq \ell\leq \max_{x\in D_{p}(\kappa-\Delta^{\rm max-proj})}\langle x-p,v\rangle\] where $\Delta^{\rm max-proj}(\epsilon,\beta) = \frac{12\log((T+2)/\beta)}{\epsilon} = O(\frac{\log((\upsilon+\log(d))/\alpha\beta)}{\epsilon})$.
\end{theorem}
Note that in the bound of Theorem~\ref{thm:dp_approx_max_proj} we relied on the na\"ive upper bound of $\diam_{\kappa}\leq\sqrt{d}$. Clearly, if $D\ll\sqrt d$ then we get a tighter bound on $\Delta^{\rm max-proj}$.

\paragraph{Large-TDC Direction.} Second, we deal with a problem of finding a good direction $v$ where there a point $q$, where $\langle q,v\rangle$ takes a particular value and $q$ has large Tukey-depth. Formally, our algorithm takes as input a particular point $p$ and a scalar $\lambda$, a candidate set of possible directions $V$, and a Tukey-depth parameter $\kappa$. It returns (w.h.p.) a directions $v$ where there exists a point $q$ of large Tukey-depth and where $\langle q,v\rangle = \langle p, v\rangle+\lambda$ (if such a direction exists).

\begin{algorithm}[hbt]\caption{\label{alg:dp_large_TDC_direction}{\tt DPLargeTDCDirection:} Finds $v$ where some point $q$ or large Tukey-depth exists such that $\langle q,v\rangle$ is given}
	{\bf Input}: $P\subset \G^{d}\subset[0,1]^d$; privacy loss $\epsilon>0$, approximation parameters $\beta >0$; Tukey-depth parameter $\kappa \geq 0$; a given set of directions (unit-length vectors)~$V$; a point $p$; a scalar~$\lambda$.
	%\\{\bf Output}:
	\begin{algorithmic}[1]
		\STATE $T = |V|$.
		\FOR {\textbf{each} $v\in V$}
		\STATE Compute $x\gets \langle p,v\rangle+\lambda$
		\STATE Sample $Y_v \sim \Lap{\tfrac 3 \epsilon}$.
		\STATE \textbf{if} $\TDC^{R_v}(x) + Y_v \geq \kappa - \frac{6\ln(T+1)/\beta)}{\epsilon} + X$,
		where $R_v$ is a rotation that sets $v$ as the first vector basis (namely $v\stackrel{R_v}\mapsto e_1$), \textbf{then return} $v$ and \textbf{halt}.
		\ENDFOR
		\RETURN $\bot$.
	\end{algorithmic}
\end{algorithm}

\begin{theorem}
	\label{thm:dp_find_direction_with_large_TDC}
	Algorithm~\ref{alg:dp_large_TDC_direction} is a $\epsilon$-DP algorithm that, given a point $p$, a scalar $\lambda$ and a set of possible directions $V$, returns w.p. $\geq 1-\beta$ a direction $v\in V$ such that there exists a point $q$ with Tukey-depth $\kappa-\Delta^{\rm LargeTDCDir}(\epsilon,\beta)$ where $\langle q,v\rangle = \langle p,v\rangle + \lambda$, with $\Delta^{\rm LargeTDCDir}(\epsilon,\beta) = \frac{12\ln((|V|+1)/\beta)}{\epsilon} = O(\frac{\log(|V|/\beta)}{\epsilon})$ (if such a direction exists).
\end{theorem}

\section{Private Approximation of a Kernel~--- For a ``Fat'' Tukey Region}
\label{sec:DP_kernel_for_fat_Tukey_region}

\subsection{Different Notions of Kernels and Various Definitions of Fatness}
\label{subsec:kernel_definitions}

Before we give our algorithm(s) for finding a kernel of a Tukey-region, we first discuss \emph{our goal}~--- what it is we wish to output, and \emph{our premise}~--- the kinds of datasets on which we are guaranteed to release such outputs. Recall, our goal is to give a  differentially private algorithm that outputs a collection of points $\calS$ which is a $(\alpha,\Delta)$-kernel of $D_P(\kappa)$. Namely, this $\calS$ satisfies that
\begin{equation}
\label{eq:kernel_as_CH}
(1-\alpha) D_P(\kappa)\subset \CH(\calS) \subset (1+\alpha)D_P(\kappa-\Delta) 
\end{equation}
Clearly, if two convex bodies ${\cal A}\subset {\cal B}$ then for any projection $\Pi$ we have that $\Pi({\cal A})\subset \Pi({\cal B})$. (In fact, this holds for any affine transformation, not just projections.) In particular, if $\calS$ is a $(\alpha,\Delta)$-kernel, then: 
\begin{equation}
\label{eq:directional_approximation_kernel}
\forall\textrm{ direction $u$,}~~~~~(1-\alpha) \max_{p\in D_{P}(\kappa)} \langle p,u\rangle \leq \max_{p\in \CH(\calS)}\langle p,u\rangle \leq (1+\alpha) \max_{p\in D_{P}(\kappa-\Delta)}\langle p,u\rangle
\end{equation}
It is actually easy to see that the two are equivalent conditions.
\begin{proposition}
	\label{pro:approx_every_direction_is_approximate_kernel}
	Assume that the origin $\bar 0$ is a point in $D_P(\kappa)$. Let $\calS$ be a set that satisfy that for every direction $u$ it holds that $(1-\alpha) \max_{p\in D_{P}(\kappa)} \langle p,u\rangle \leq \max_{p\in \CH(\calS)}\langle p,u\rangle \leq (1+\alpha) \max_{p\in D_{P}(\kappa-\Delta)}\langle p,u\rangle$; then $\calS$ is a $(\alpha,\Delta)$-kernel.
\end{proposition}
\begin{proof}
	$\CH(\calS)$ is the intersection of a finite number, $k$, of closed half-spaces. Thus there exist $k$ vectors $v_1,...v_k$ and $k$ scalars $\lambda_1,...\lambda_k$ such that $\CH(\calS) = \{x\in \R^d: \forall v_i, ~~ \langle x,v_i\rangle\leq  \lambda_i\}$. For any $v_i,\lambda_i$ we have that for any $p\in D_P(\kappa)$ where $\langle p,v_i\rangle \geq 0$ it holds that
	\[ \langle(1-\alpha) p, v_i\rangle \leq (1-\alpha)\max_{p\in D_P(\kappa)}\langle p,v_i\rangle \leq \max_{p\in \CH(\calS)}\langle p,v_i\rangle \leq \lambda_i \]
	In particular, for the origin, $\bar 0\in D_P(\kappa)$ this shows that $(1-\alpha)\langle \bar 0,v_i\rangle = 0\leq \lambda_i$ proving that $\lambda_i$ is non-negative. Thus, for any $p$ where $\langle p,v_i\rangle < 0$ we obviously have
	$\langle(1-\alpha) p, v_i\rangle < 0 \leq \lambda_i$. And so $(1-\alpha)D_P(\kappa)\subset \CH(\calS)$. The proof that $\CH(\calS)\subset(1+\alpha)D_P(\kappa-\Delta)$ is symmetric, since we now know $\bar  0\in (1-\alpha)D_P(\kappa)\subset \CH(\calS)$.
\end{proof}
It is worth noting that in addition to the property in~\eqref{eq:directional_approximation_kernel}, if $\calS$ is a $(\alpha,\Delta)$-kernel of $D_P(\kappa)$ then it also holds that
\begin{equation}
\label{eq:directional_kernel_width}
\forall\textrm{ direction $u$,}~~~~~(1-\alpha) \max_{p,q\in D_{P}(\kappa)} \langle p-q,u\rangle \leq \max_{p,q\in \CH(\calS)}\langle p-q,u\rangle \leq (1+\alpha) \max_{p,q\in D_{P}(\kappa-\Delta)}\langle p-q,u\rangle
\end{equation}
In the standard, non-private, setting, the definition of a directional-kernel~\cite{AgarwalHV04} is a set $\calS$ that is required to satisfy both the property in~\eqref{eq:directional_kernel_width} (with $\Delta=0$) and the property that $\calS\subset D_P(\kappa)$. These two properties together yield the desired property of a kernel given in~\eqref{eq:kernel_as_CH}. It turns out that in our setting, with $\Delta>0$, since it doesn't necessarily hold that $\calS\subset D_P(\kappa)$, then property~\eqref{eq:directional_kernel_width} does \emph{not} guarantee that we output $(\alpha,\Delta)$-kernel. Figure~\ref{fig:approx_width_not_a_kernel} illustrates such a setting. 

\begin{figure}[t]
	\begin{minipage}[c]{0.575\textwidth}
	\caption{\label{fig:approx_width_not_a_kernel}
		An example showing that the property of Equation~\eqref{eq:directional_kernel_width} doesn't imply that $(1-\alpha)D_P(\kappa)\subset \CH(\calS)$. Suppose $D_P(\kappa)$ is an equilateral triangle of edge-length $2r$ and $\calS$ happens to be a ball of diameter $2\cdot0.99\cdot r$ (and $D_P(\kappa-\Delta)$ is a much larger region). Note that $\calS$ does satisfy Equation~\eqref{eq:directional_kernel_width} for $\alpha=0.01$ yet $0.99 D_P(\kappa) \nsubset \CH(\calS)$. (E.g, the ball containing $D_P(\kappa)$ must have radius $\geq \frac{2}{\sqrt 3} r$ whereas $\CH(\calS)$ is contained inside a ball of radius $0.99r$.)
	}
	\end{minipage}
\hfill
\begin{minipage}[c]{0.4\textwidth}
	\centering\includegraphics[scale=0.3]{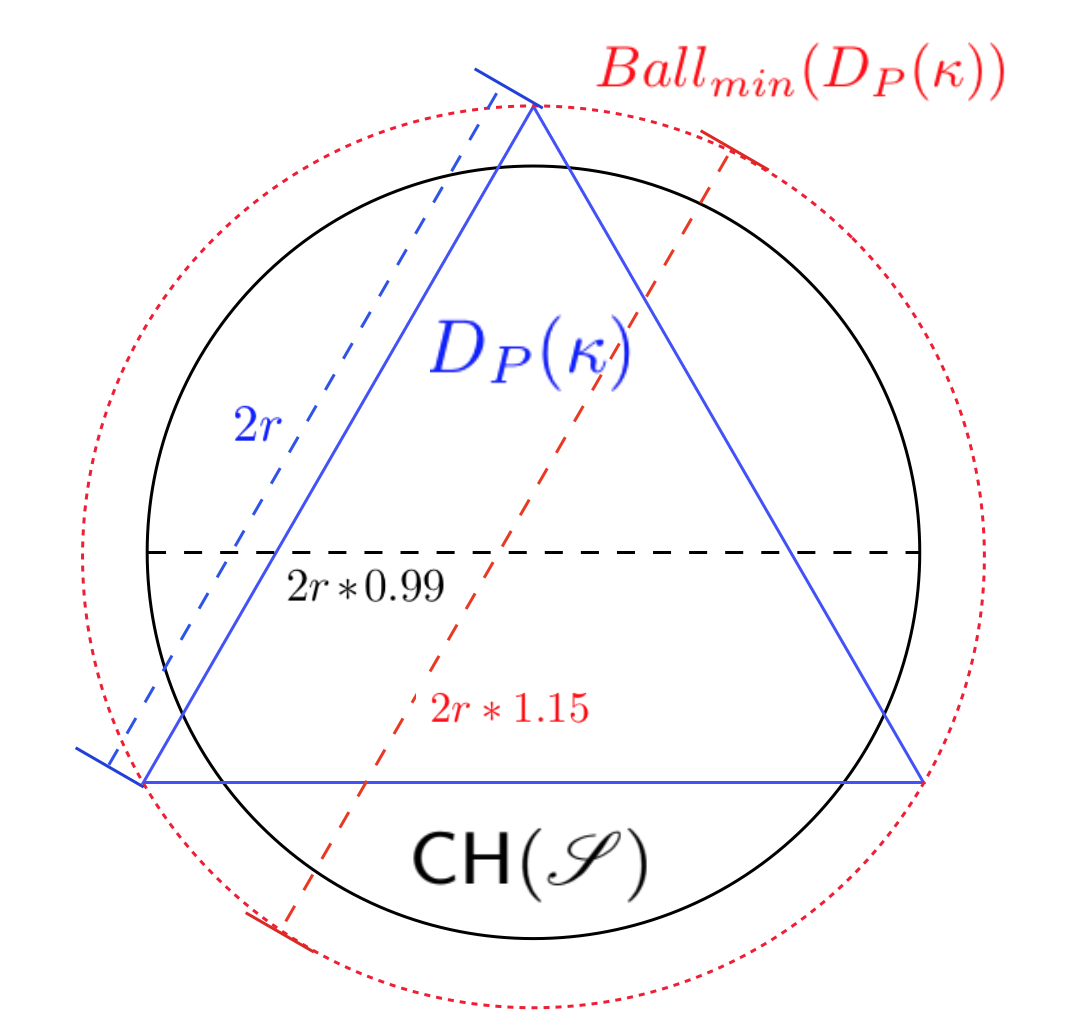}
\end{minipage}
\end{figure}

In our work, we give algorithms that satisfy variations of property~\eqref{eq:directional_approximation_kernel}. We give now the respective claims showing that each variation indeed yields a $(\alpha',\Delta)$-kernel.

\begin{claim}
	\label{clm:additive_width_approx_yields_kernel}
	Let $\calS$ be a set that satisfies the following property in regards to $D_P(\kappa)$ and $D_P(\kappa-\Delta)$:
	\begin{equation}
	\label{eq:additive_width_kernel}
	\forall\textrm{ direction $u$,}~~~~~\max_{p\in D_{P}(\kappa)} \langle p,u\rangle - \alpha\cdot\width_{\kappa} \leq \max_{p\in \CH(\calS)}\langle p,u\rangle \leq  \max_{p\in D_{P}(\kappa-\Delta)}\langle p,u\rangle + \alpha\cdot \width_{\kappa-\Delta}
	\end{equation}
	then, setting $\alpha' = 2\alpha\sqrt{d+\frac 1 2}$, there exists two vectors $p_1$ and $p_2$ s.t. we can shift $D_P(\kappa)$ and $D_P(\kappa-\Delta)$ and have that  $(1-\alpha')(D_{P}(\kappa)-p_1) \subset \CH(\calS)-p_1$ and $\CH({\calS})-p_2\subset (1+\alpha')(D_P(\kappa-\Delta)-p_2)$.
\end{claim}
\begin{proof}
	We first argue about the relation between $\CH({\calS})$ and $D_P(\kappa-\Delta)$. Denote the convex polytope $D_P(\kappa-\Delta)$ as the intersection of a finite number, $k$, of closed half-spaces: $\{x\in \R^d: \langle x,v_i\rangle \leq\lambda_i\}$. We continue and leverage on the fact (see~\cite{GritzmannK92}) that any convex body with width $w$ must contain a ball of radius at least $\frac{w\sqrt{d+2}}{2(d+1)}$. Let $p_2$ be the center of this ball, and so $D_P(\kappa-\Delta)-p_2$ is a shift of $D_P(\kappa-\Delta)$ where this ball is centered at the origin. 
	%Now, we denote the bounding hyperplanes of the shifted $D_P(\kappa-\Delta)-p_2 = \{x:~ \langle x,v_i\rangle \leq \lambda_i\}$. 
	Note that the origin is not only a point inside this shifted convex polytope, it is also a point of distance at least $\frac{\width_{\kappa-\Delta}\sqrt{d+2}}{2d+2}$ of all hyperplanes bounding it. Thus, based on for this particular shift, we can redefine the closed halfspaces and have that  $D_P(\kappa-\Delta)=\bigcap_i\{ x\in \R^d: \langle x-p_2,v_i\rangle \leq \lambda_i \}$ where we also have that each $\lambda_i \geq \frac{\width_{\kappa-\Delta}\sqrt{d+2}}{2d+2}$.
	
	For any closed halfspaced parameterized by $v_i,\lambda_i$ we have that for any $x\in \CH(\calS)$ it holds that
	\begin{align*} 
	\langle x-p_2,v_i\rangle &\leq \max_{x\in \CH(\calS)}\langle x,v_i\rangle - \langle p_2,v_i\rangle \leq \max_{x\in D_{P}(\kappa-\Delta)}\langle x,v_i\rangle +\alpha\cdot width_{\kappa-\Delta} - \langle  p_2,v_i\rangle 
	\cr & = \max_{x\in D_P(\kappa-\Delta)} \langle x-p_2,v_i\rangle + \alpha\cdot \width_{\kappa-\Delta} \leq \lambda_i + \frac{2d+2}{\sqrt{d+2}}\alpha \lambda_i \leq (1+2\alpha\sqrt{d+\tfrac 1 2})\lambda_i  
	\end{align*}
	This proves that $\CH(\calS)-p_2\subset (1+\alpha')(D_P(\kappa-\Delta)-p_2)$.
	
	Next, we show that there exists $p_1$ such that $(1-\alpha)(D_P(\kappa)-p_1) \subset \CH(\calS)-p_1$. 
	We start by comparing $\width(\CH(\calS))$ to $\width_\kappa$.
	Let $u$ be the direction on which the width of $\CH(\calS)$ is obtained. We thus have that 
	\begin{align*}
	\width(\CH(\calS)) &= \max_{a,b\in \CH(\calS)}\langle a-b,u\rangle = \max_{a\in \CH(\calS)}\langle a,u\rangle- \min_{b\in \CH(\calS)}\langle b,u\rangle 
	\cr &=  \max_{a\in \CH(\calS)}\langle a,u\rangle + \max_{b\in \CH(\calS)}\langle b,-u\rangle 
	\cr &\geq \max_{a\in D_P(\kappa)}\langle a,u\rangle + \max_{b\in D_P(\kappa)}\langle b,-u\rangle - 2\alpha\cdot  \width_{\kappa} 
	\cr &= \max_{a,b\in D_P(\kappa)}\langle a-b,u\rangle - 2\alpha\cdot  \width_{\kappa} \geq \width_{\kappa}(1-2\alpha)
	\end{align*}
	This implies that in any direction $u$ we have that $\max_{p\in D_P(\kappa)}\langle p,u\rangle \leq \max_{p\in \CH(\calS)} \langle p,u\rangle + \alpha\cdot \width_{\kappa} \leq  \max_{p\in \CH(\calS)} \langle p,u\rangle + \frac{\alpha}{1-2\alpha}\cdot \width(\CH(\calS))$. We can now apply the above argument and have that \linebreak $D_P(\kappa)~-~p_1~\subset (1+\frac{2\alpha}{1-2\alpha}\sqrt{d+\tfrac 1 2})(\CH(\calS)-p_1)$ for some $p_1$. Thus, $(1-{2\alpha}\sqrt{d+\tfrac 1 2})(D_P(\kappa)-p_1)\subset \CH(\calS)-p_1$.
\end{proof}

As discussed, our first algorithm yields a set $\calS$ that (w.h.p.) satisfies the premise of Claim~\ref{clm:additive_width_approx_yields_kernel} and therefore it is a  $(\alpha,\Delta)$-kernel. Similarly, the second algorithm we provide (under slightly different conditions) yields the premise of the following claim.

\begin{claim}
	\label{clm:(1-alpha)_width_approx_yields_kernel}
	Fix $\alpha<1/6$ and let $\calS\subset D_P(\kappa-\Delta)$ be a set that satisfy the following property in regards to $D_P(\kappa)$. There exists a point $c\in D_P(\kappa)\cap \calS$ such that:
	\begin{equation}
	\label{eq:(1-alpha)_width_approximation_yields_kernel}
	\forall\textrm{ direction $u$,}~~~~~(1-\alpha)\max_{p\in D_{P}(\kappa)} \langle p-c,u\rangle \leq \max_{p\in \CH(\calS)}\langle p-c,u\rangle + \alpha\cdot \width_{\kappa}
	\end{equation}
	then there exists a vector $b$ such that we can shift $D_P(\kappa)$ and $\CH(\calS)$ by $b$ and have that  \[D_{P}(\kappa)-b \subset (1+\alpha')\left(\CH(\calS)-b\right)\] for $\alpha' = \frac{\alpha}{1-\alpha}(1+4\sqrt{d+\frac 1 2})$; and we also have that $\CH({\calS})\subset D_P(\kappa-\Delta)$. Thus, obviously, $\calS$ is a $\left(\frac{\alpha'}{1+\alpha'},\Delta\right)$-kernel of $D_P(\kappa)$.
\end{claim}
\begin{proof}
	First, since $\calS\subset D_P(\kappa-\Delta)$ then it is obvious that $\CH(\calS)\subset D_P(\kappa-\Delta)$. The difficulty lies in showing the first part.
	
	We start by a similar argument to the one in the proof of Claim~\ref{clm:additive_width_approx_yields_kernel}, comparing $\width(\CH(\calS))$ to $\width_\kappa$.
	Let $u$ be the direction on which the width of $\CH(\calS)$ is obtained. We thus have that 
	\begin{align*}
	\width(\CH(\calS)) &= \max_{a,b\in \CH(\calS)}\langle a-b,u\rangle = \max_{a\in \CH(\calS)}\langle a-c,u\rangle- \min_{b\in \CH(\calS)}\langle b-c,u\rangle 
	\cr &=  \max_{a\in \CH(\calS)}\langle a-c,u\rangle + \max_{b\in \CH(\calS)}\langle b-c,-u\rangle 
	\cr &\geq (1-\alpha)\max_{a\in D_P(\kappa)}\langle a-c,u\rangle + (1-\alpha)\max_{b\in D_P(\kappa)}\langle b-c,-u\rangle - 2\alpha\cdot  \width_{\kappa} 
	\cr &= (1-\alpha)\max_{a,b\in D_P(\kappa)}\langle a-b,u\rangle - 2\alpha\cdot  \width_{\kappa} \geq \width_{\kappa}(1-3\alpha)
	\end{align*}
	This implies that in any direction $u$ we have that $\max_{p\in D_P(\kappa)}\langle p-c,u\rangle \leq \frac 1 {1-\alpha}\max_{p\in \CH(\calS)} \langle p-c,u\rangle + \frac{\alpha}{1-\alpha}\cdot \width_{\kappa} \leq  \frac 1 {1-\alpha}\max_{p\in \CH(\calS)} \langle p-c,u\rangle + \frac{\alpha}{(1-\alpha)(1-3\alpha)}\cdot \width(\CH(\calS))$.
	
	Now, based on a claim from~\cite{GritzmannK92}, we know that $\CH(\calS)$ contains a ball, centered at some point $z$, such that its radius is at least $\width(\CH(\calS)) \frac{\sqrt{d+2}}{(2d+2)}$. Thus we denote the convex polytope $\CH(\calS)$ as $\CH(\calS) = \bigcap_i \{x\in\R^d:~ \langle x-z,v_i\rangle \leq \lambda_i\}$ where for every $i$ it holds that $\lambda_i \geq \width(\CH(\calS))\frac{\sqrt{d+2}}{2d+2} \geq \width_{\kappa}(1-3\alpha)\frac{\sqrt{d+2}}{2d+2}$.
	
	Set $\beta = \frac{1}{1+4\sqrt{d+\frac 1 2}}\in[0,1]$, and denote $b = (1-\beta)z+\beta c$. Note that $b\in \CH(\calS)$ due to convexity. Moreover, since for every $v_i$ it holds that $\forall x, \langle x-z,v_i\rangle \leq \lambda_i$ iff $\langle x-b,v_i\rangle \leq \lambda_i + \beta\langle z-c,v_i\rangle$, then we can rewrite $\CH(\calS)$ as $\CH(\calS) = \cap_i\{x\in \R^d:~ \langle x-b, v_i\rangle \leq \lambda_i +\beta\langle z-c,v_i\rangle\}$. So, as our goal is to show that $D_P(\kappa)-b \subset (1+\alpha')(\CH(\calS)-b)$. Namely, we show that all $x\in D_P(\kappa)$ satisfy $\langle x-b, v_i\rangle \leq (1+\alpha')\left(\lambda_i +\beta\langle z-c,v_i\rangle\right)$ for any $v_i,\lambda_i$.
	
	Fix any closed halfspace parameterized by $v_i,\lambda_i$ and any $x\in D_P(\kappa)$. We have that
	\begin{align*}
	\langle x-b,v_i\rangle &= \langle x-c,v_i\rangle + \langle c-b,v_i\rangle
	\cr &\leq \frac 1 {1-\alpha}\max_{p\in \CH(\calS)} \langle p-c,v_i\rangle + \frac{\alpha}{1-\alpha}\cdot \width_{\kappa}  + \langle c-b,v_i\rangle
	\cr & \leq \frac 1 {1-\alpha}\max_{p\in \CH(\calS)} \langle p-b,v_i\rangle + \frac{1}{1-\alpha}\langle b-c,v_i\rangle + \lambda_i\frac{\alpha(2d+2)}{(1-\alpha)(1-3\alpha)\sqrt{d+2}}  - \langle b-c,v_i\rangle
	\cr& \leq \frac 1 {1-\alpha}\big(\lambda_i + \beta\langle z-c,v_i\rangle\big) + \lambda_i\frac{\alpha(2d+2)}{(1-\alpha)(1-3\alpha)\sqrt{d+2}} + \frac{\alpha}{1-\alpha}\langle (1-\beta)(z-c) ,v_i\rangle
	\cr & = \lambda_i + \lambda_i\frac{\alpha}{1-\alpha}+\lambda_i\frac{\alpha(2d+2)}{(1-\alpha)(1-3\alpha)\sqrt{d+2}} + \frac{\beta + \alpha-\alpha\beta}{1-\alpha}\langle z-c,v_i\rangle
	\cr& = \lambda_i  \left[1+\frac{\alpha}{1-\alpha}\left(1+\frac{2d+2}{(1-3\alpha)\sqrt{d+2}}\right)\right] + \left(\beta + \frac{\alpha}{1-\alpha}\right) \langle z-c,v_i\rangle
	\end{align*}
	Note that
	$
	(1+\alpha')\beta= \left(1+\frac{\alpha}{1-\alpha} \cdot \frac 1 \beta\right)\beta = \beta + \frac{\alpha}{1-\alpha}
	$. The key point here is that we have equality, not inequality, and this allows us to ignore partitioning into cases and see whether $\langle z-c,v_i\rangle$ is positive or not. Plugging in this equality into the above bound we get
	\begin{align*}
	\langle x-b,v_i\rangle &\leq  \left[1+\tfrac{\alpha}{1-\alpha}\left(1+\tfrac{2d+2}{(1-3\alpha)\sqrt{d+2}}\right)\right]\lambda_i + (1+\alpha')\beta\langle z-c,v_i\rangle \leq (1+\alpha')\lambda_i + (1+\alpha')\beta\langle z-c,v_i\rangle
	\end{align*}
	where the last inequality holds because $\lambda_i>0$ and because \[\tfrac{\alpha}{1-\alpha}\left(1+\tfrac{2d+2}{(1-3\alpha)\sqrt{d+2}}\right)\leq \alpha' = \tfrac{\alpha}{1-\alpha}(1+4\sqrt{d+\tfrac 1 2})\] (Since $\alpha\leq 1/6$ then $\frac{2d+2}{(1-3\alpha)\sqrt{d+2}}\leq 4\frac{d+1}{\sqrt {d+2}}\leq 4\sqrt{d+\frac 1 2}$.)
\end{proof}

\medskip

\paragraph{Definition of Fantess.} In the following subsection we detail our algorithms whose respective outputs satisfy the premise of Claims~\ref{clm:additive_width_approx_yields_kernel} and~\ref{clm:(1-alpha)_width_approx_yields_kernel}.
Unfortunately, we were unable to find an algorithm that returns a kernel for \emph{any} $D_P(\kappa)$. Much like in the non-private setting~\cite{AgarwalHV04}, in order to give an algorithm that outputs a kernel of $D_P(\kappa)$ we must require $D_P(\kappa)$ satisfies a certain ``fatness'' property. In the standard, non-private setting, a convex polytope $D_P(\kappa)$ is called $c_d$-fat if there exists a constant $c_d\geq 1$ (depends solely on the dimension $d$) where $\diam(D_P(\kappa))\leq c_d \width(D_P(\kappa))$ (see~\cite{AgarwalHV04}). Alas, our differentially private algorithm requires something stronger. Formally, we define the follow various notions of fatness. 
\begin{definition}
	\label{def:fat_Tukey_region}
	Given a dataset $P$, we say that its $\kappa$-Tukey region is
	\begin{itemize}
		\item  \emph{$(c_d,\Delta)$-fat} if it holds that $\width_\kappa \geq \tfrac 1 {c_d}\diam_{\kappa-\Delta}$.
		\item \emph{$(c_d,\Delta_+,\Delta_-)$-fat} if it holds that
		$\width_{\kappa+\Delta_+} \geq \tfrac 1 {c_d}\diam_{\kappa-\Delta_-}$.
		\item \emph{$c_d$-absolutely fat} if
		$\width_{\kappa} \geq \tfrac {1}{c_d}$.
	\end{itemize}
\end{definition}
\noindent The following properties are immediate from the various definitions.
\begin{itemize}
	\item If $D_P(\kappa)$ is $(c_d,\Delta)$-fat then for any $0\leq \Delta'\leq \Delta$ it is also $(c_d, \Delta')$-fat. In particular it is also $(c_d,0)$-fat which is the standard, non-private, definition of fatness.
	\item If $D_P(\kappa)$ is $(c_d,\Delta)$-fat then for any $c_d'\geq c_d$ it is also $(c_d', \Delta)$-fat.
	\item $D_P(\kappa)$ is $(c_d,\Delta_+,\Delta_-)$-fat iff $D_P(\kappa+\Delta_+)$ is $(c_d, (\Delta_++\Delta_-)~)$-fat. 
	\item If $D_P(\kappa)$ is $(c_d,\Delta_+,\Delta_-)$-fat then for any $\tilde{\Delta}_+\leq\Delta_+$ and $\tilde{\Delta}_-\leq\Delta_-$ it holds that $D_P(\kappa)$ is also\linebreak $(c_d, \tilde{\Delta}_+, \tilde{\Delta}_-)$-fat. 
	\item Since $P\subset[0,1]^d$, then for any $\kappa'$ we have that $\diam_{\kappa'}\leq\diam(P)\leq \sqrt d$. It follows that if $D_P(\kappa)$ is $c_d$-absolutely fat, then it is also $(c_d/\sqrt d, \Delta)$-fat for any $\Delta$.
\end{itemize}

\paragraph{Discussion.} It is clear that the fatness properties (i.e., non-private $c_d$-fat, $(c_d,\Delta)$-fat, $c_d$-absolutely fat) can be violated by the addition or removal of a single datapoint to/from $P$. Therefore, no differentially private algorithm can always assert w.h.p. whether $D_P(\kappa)$ is fat or not, nor estimate its fatness parameter $c_d$. We therefore proceed as follows. In the next few subsections we give our differentially private algorithms for fat Tukey-regions. That is, in Subsection~\ref{subsec:private_kernel_absolute_fatness} we assume that $D_P(\kappa)$ is $c_d$-absolutely fat and return a set $\calS$ that satisfies the premise of Claim~\ref{clm:additive_width_approx_yields_kernel}; and in Subsection~\ref{subsec:private_kernel_fatness} we assume $D_P(\kappa)$ is $(c_d,\Delta)$-fat and return a set $\calS$ which satisfies the premise of Claim~\ref{clm:(1-alpha)_width_approx_yields_kernel}. Moreover, in Subsection~\ref{subsec:private_selection} we propose a heuristic that, assuming $D_P(\kappa)$ is $(c_d,\kappa_1,\kappa_2)$-fat for some particular values of $\kappa_1,\kappa_2$, returns an estimation of $c_d$. But more importantly, in Sections~\ref{sec:approx_bounding_box} and~\ref{sec:finding_good_kappa} we show how to privately find a transformation $T$ that turns $T(D_P(\kappa))$ into a fat dataset. This transformation relies on the promise that $\vol(D_P(\kappa))\geq\frac 12 \vol(D_P(\kappa-\Delta))$, a promise which $D_P(\kappa)$ may not satisfy. However, in Section~\ref{sec:finding_good_kappa} we show how to find a value $\kappa^*$ where $D_P(\kappa^*)$ does satisfy this promise, allowing us to convert $D_P(\kappa^*)$ into a fat Tukey-region and then produce a kernel for $D_P(\kappa^*)$.

\subsection{Private Kernel Approximation Under ``Absolute Fatness''}
\label{subsec:private_kernel_absolute_fatness}

In this section, we work under the premise that $D_P(\kappa)$ is $c_d$-absolutely fat, that is, that $\width_\kappa \geq 1/c_d$. For some instances, we are able to privately check whether $D_P(\kappa)$ is $c_d$-absolutely fat~--- if it happens to be the case that $D_P(\kappa+\Delta^{\rm width})$ is $c_d$-absolutely fat, we can apply Algorithm~\ref{alg:DP_approx_width_tukey_region} and verify it is indeed the case.
Moreover, even when $D_P(\kappa)$ isn't absolutely-fat, in Sections~\ref{sec:approx_bounding_box} and~\ref{sec:finding_good_kappa} we discuss at length how to privately find a parameter $\kappa$ and a mapping $T$ that transforms $D_P(\kappa)$ into a absolutely-fat Tukey-region.

For absolutely fat Tukey-regions, we are able to give a pretty simple algorithm: we traverse a fine enough grid and add a point to $\calS$ if it is in a vicinity of a point in $D_P(\kappa)$. Details appear below.

\begin{algorithm}[ht]
	\caption{\label{alg:private-kernel-abs-fatness} Private Kernel Approximation of an Absolutely-Fat Tukey Region}
	{\bf Input}: Dataset $P\subset \mathbb{R}^{d}$; Approximation parameter $0< \alpha,\beta <1/2$; privacy parameters $\epsilon,\delta> 0$; Tukey depth parameter $\kappa$ and fatness parameter $c_d$.
	\begin{algorithmic}[1]
		\STATE Set $\zeta \gets \frac{\alpha}{c_d\sqrt d}$. Let $G_\zeta$ be the partitioning of the unit-cube $[0,1]^d$ to subcubes of edge-length $\zeta$. 
		\STATE Set $k \gets |G_\zeta|$ and $\epsilon_0\gets \frac{\epsilon}{2\sqrt{k\log(1/\delta)}}$, $\beta_0 \gets \frac{\beta}{k}$.
		\STATE Init $\calS \gets \emptyset$.
		\FOR{\textbf{each} $C\in G_\zeta$}
		\IF { ($\max_{x\in C}\TD(x,P)+Y_C\geq \kappa-\frac{\ln(1/\beta_0)}{\epsilon_0}$ where $Y_C\sim\Lap{\frac 1 {\epsilon_0}}$)}
		\STATE Add the center of $C$ to $\calS$.
		\ENDIF 
		\ENDFOR
		\RETURN $\calS$
	\end{algorithmic}
\end{algorithm}
\begin{theorem}
	\label{thm:private_kernel_abs_fatness}
	Algorithm~\ref{alg:private-kernel-abs-fatness} is an efficient, $(\epsilon,\delta)$-DP algorithm that returns w.p. $\geq 1-\beta$ a set $\calS$  that satisfies for every direction $u$ that $\max_{p\in D_P(\kappa)}\langle p,u\rangle - \alpha\cdot \width_{\kappa} \leq \max_{p\in\calS}\langle p,u\rangle \leq \max_{p\in D_P(\kappa-\Delta^{\rm kernel})}\langle p,u\rangle - \alpha\cdot \width_{\kappa-\Delta^{\rm kernel}}$,
	where $\Delta^{\rm kernel}= O(\frac{d(\frac {c_d\sqrt d}{\alpha})^{d/2}\sqrt{\log(1/\delta)}\log(\frac{c_d d}{\alpha\beta})}{\epsilon})$.
\end{theorem}
\begin{proof}
	First, to see that Algorithm~\ref{alg:private-kernel-abs-fatness} is efficient, note that $k = O((c_d\sqrt d/\alpha)^d)$. For each of the $k$ cubes in $G_\zeta$ we find the largest $\kappa'$ such that $D_P(\kappa')\cap C\neq\emptyset$ using a LP, thus we are able to answer of the $\max_{x\in C}\TD(x,P)$ queries in $\poly$-time. (In fact, it is enough to check for each cube that some vertex of $D_P(\kappa')$ exists on each of the $2d$-sides of the cube.)
	
	Second, Algorithm~\ref{alg:private-kernel-abs-fatness} is clearly $(\epsilon,\delta)$-DP since it relies on $k$ queries, each with sensitivity of $1$. Thus, ``budgeting'' the additive Laplace mechanism with privacy-loss parameter of $\epsilon_0$ turns the entire algorithm to $(\epsilon,\delta)$-DP algorithm based on the advanced composition theorem~\cite{DworkRV10}.
	
	We continue under the event that for each $C\in G_\zeta$ we picked a random variable $Y_C$ such that $|Y_C|\leq \frac{\ln(1/\beta_0)}{\epsilon_0}$, an event we know to hold with probability $\geq 1-\beta$. Under this event, two things must occur: (i) for any $C$ where some $x_C\in C$ has Tukey-depth $\geq \kappa$ we place a point $y_C\in \calS$, and (ii) for any $C$ where we place its center point $y_C\in \calS$ there exists some $z_C\in C$ with Tukey-depth of at least $\kappa-\frac{2\ln(1/\beta_0)}{\epsilon_0} = \kappa - \Delta^{\rm kernel}$. Note that $\|x_C-y_C\|,\|y_C-z_C\| \leq \sqrt{d(\frac{\zeta}{2})^2} = \frac{\alpha}{2c_d} \leq \frac{\alpha \width_{\kappa}}{2}$. 
	
	%We rely on these two implications to show that indeed we output a $(\alpha,\Delta^{\rm kernel})$-kernel of $D_P(\kappa)$. Fix any direction $u$. Let $a,b\in D_P(\kappa)$ be the two points that obtain the directional-width of $D_P(\kappa)$, and let $y_a$ and $y_b$ be the two points in $\calS$ that are of distance $\leq \frac{\alpha \cdot \width_{\kappa}}{2}$ to $a$ and $b$ resp. We thus have that
	%\begin{align*} 
	%w_\kappa(u) &= \langle a-b,u\rangle = \langle y_a-y_b,u\rangle +\langle a-y_a,u\rangle -\langle b-y_b,u\rangle 
	%\cr &\leq \langle y_a-y_b,u\rangle + 2\frac{\alpha \cdot \width_{\kappa}}{2}\leq \max_{p,q\in \calS}\langle p-q,u\rangle + \alpha \cdot w_\kappa(u)
	%\end{align*}
	%implying that $(1-\alpha)w_\kappa(u)\leq \max_{p,q\in \calS}\langle p-q,u\rangle$.
	
	%Similarly, denote $p$ and $q$ as the two points in $\calS$ that obtain the directional width along $u$, and let $z_p$ and $z_q$ as the two points of depth $\geq \kappa-\Delta^{\rm kernel}$ closest to $p$ and $q$ resp. Again, we have that
	%\begin{align*} 
	%\max_{a,b\in \calS}\langle a-b,u\rangle &= \langle p-q,u\rangle \leq \langle z_p-z_q,u\rangle + 2\frac{\alpha \cdot \width_{\kappa}}{2} 
	%\cr &\leq \max_{z_1, z_2\in D_P(\kappa-\Delta^{\rm kernel})}\langle z_1-z_2,u\rangle + \alpha \cdot \width_{\kappa-\Delta^{\rm kernel}}(u) \leq (1+\alpha) w_{\kappa-\Delta^{\rm kernel}}(u)  
	%\end{align*}

	We rely on these two implications to show that indeed we output a $(\alpha,\Delta^{\rm kernel})$-kernel of $D_P(\kappa)$. Fix any direction $u$. Let $a\in D_P(\kappa)$ be the point that obtain the directional-max of $D_P(\kappa)$, and let $y_a$ be a point in $\calS$ that is of distance $\leq \frac{\alpha \cdot \width_{\kappa}}{2}$ from $a$. We thus have that
	\begin{align*}
	\max_{p\in D_P({\kappa})}\langle p,u\rangle &= \langle a,u\rangle = \langle y_a,u\rangle +\langle a-y_a,u\rangle 
	\cr &\leq \langle y_a,u\rangle + \frac{\alpha \cdot \width_{\kappa}}{2}\leq \max_{p\in \calS}\langle p,u\rangle + \alpha \cdot \width_\kappa
	\end{align*}
	
	Similarly, denote $p$ as the points in $\calS$ that obtains the directional width along $u$, and let $z_p$ be a point of depth $\geq \kappa-\Delta^{\rm kernel}$ closest to $p$. Again, we have that
	\begin{align*} 
	\max_{z\in \calS}\langle a,u\rangle &= \langle p,u\rangle = \langle z_p,u\rangle +\langle p-z_p,u\rangle \leq \langle z_p,u\rangle + \frac{\alpha \cdot \width_{\kappa}}{2} 
	\cr &\leq \max_{p\in D_P(\kappa-\Delta^{\rm kernel})}\langle p,u\rangle + \alpha \cdot \width_{\kappa-\Delta^{\rm kernel}}(u) \leq (1+\alpha) w_{\kappa-\Delta^{\rm kernel}}(u) \qedhere \end{align*}
	
	%Moreover, let $r$ be the radius of the smallest ball contained in $D_P(\kappa)$, which, as was shown in~\cite{GritzmannK92}, satisfies $r \geq \frac{\width_\kappa}{2\sqrt{d+1}}$. It follows that there exists some point $c\in D_P(\kappa)$ which is of distance \emph{at least} $\frac{\width_\kappa}{2\sqrt{d+1}}$ from any vertex of the polytope $D_P(\kappa)$. 
\end{proof}

\subsection{Private Kernel Approximation Under a $(c_d,\Delta)$-Fatness Assumption}
\label{subsec:private_kernel_fatness}

In this section we give a different algorithm for finding a $(\alpha,\Delta)$-kernel of $D_P(\kappa)$ under the notion of $(c_d,\Delta)$-fatness (recall Definition~\ref{def:fat_Tukey_region}), namely that $\width_{\kappa} \geq \frac{\diam_{\kappa-\Delta}}{c_d}$. Why do we present this algorithm in addition to the previous one? After all, if we know that the given dataset is $(c_d,\Delta)$-fat we can run Algorithm~\ref{alg:ApproximateDiameter} to find an approximation of $\diam_{\kappa-\Delta}$, use the algorithm in Corollary~\ref{cor:DP_approx_TD} to find a point $p\in D_P(\kappa)$, and then inflate the ball around $p$ and have that the resulting dataset is $c_d$-absolutely fat. 

The answer is composed of several facts. First, the above-mentioned suggestion for turning $D_P(\kappa)$ into a $c_d$-absolutely fat might fail, since there are datasets where it may return a number much greater than $\diam_{\kappa-\Delta}$. (Recall, the guarantee of Algorithm~\ref{alg:ApproximateDiameter} is based on $\diam_{\kappa-\Delta}$ and $\diam_{\kappa-\Delta-\Delta^{\diam}}$, and the latter could potentially be much larger than the former.) But even if this was not the case, there are additional reasons for presenting a dedicated algorithm for $(c_d,\Delta)$-fat Tukey regions. First and far most, there are datasets which are $c_d$-absolutely fat yet are $(c_d',\Delta)$-fat for a significantly smaller $c_d'< c_d$. Secondly, under a certain regime of parameters it may yield a smaller $\Delta$ then Algorithm~\ref{alg:private-kernel-abs-fatness}~--- it is scaled down by a factor of $d^{d/4}\sqrt{c_d/\alpha}$ (which significant for the large $c_d$ values we introduce) at the expense of added $\poly(d)$-factors. In addition, the guarantee of the returned kernel $\calS$ is slightly better: it actually satisfies that $\calS \subset D_P(\kappa-\Delta)$, which makes it so that $\CH(\calS)\subset D_P(\kappa-\Delta)$ (without the rescaling by a factor of $(1+\alpha)$ and without some unknown shift).

So throughout this section, we assume we know that for the given $\kappa$, our input dataset is $(c_d,\Delta)$-fat and so we produce a $(\alpha,\Delta)$-kernel for it (for the same value of $\Delta$). In fact, in order to avoid confusion with the definition of $\Delta^{\rm kernel}$ from Theorem~\ref{thm:private_kernel_abs_fatness}, we denote the change to the Tukey depth by $\Gamma^{\rm kernel}$ in this subsection and the following. 

The algorithm we discuss here mimics its non-private kernel analogue, where one uses a $\zeta$-angle cover $V_\zeta$ of the unit-sphere, with $\zeta \approx {\frac {\alpha} {c_d}}$. We start by finding some $c\in D_P(\kappa)$ using the algorithm from Corollary~\ref{cor:DP_approx_TD}. Then, iterating through all directions in $V_\zeta$, we find a point in $D_P(\kappa-\Gamma^{\rm kernel})$ which approximately maximizes the projection along the given direction. Our algorithm is thus provided below.

\begin{algorithm}[ht]
	\caption{\label{alg:private-kernel} Private Kernel Approximation of a $(c_d,\Delta)$-Fat Tukey Region}
	{\bf Input}: Dataset $P\subset \mathbb{R}^{d}$; Approximation parameter $0< \alpha,\beta <1/2$; privacy parameters $\epsilon>0,\delta\geq 0$; Tukey depth parameter $\kappa$ and fatness parameter $c_d$.
	\begin{algorithmic}[1]
		\STATE Set $\zeta \gets \min\{{\frac \alpha {2\sqrt 2c_d}}, 1/2\}$. Let $V_\zeta$ be a $\zeta$-angle cover of the unit sphere of size $2(\pi/\zeta)^{d-1}$.
		\STATE Set $k \gets d(|V_\zeta|+1)$ and $\epsilon_0\gets \frac{\epsilon}{2\sqrt{k\log(2/\delta)}}$, $\delta_0 \gets \frac{\delta}{2k}$, $\beta_0 \gets \frac{\beta}{k}$.
		\STATE Find $c\in D_{P}(\kappa)$ by setting $\kappa^* = \kappa + d\alphaqc(\epsilon_0,\delta_0,\beta_0)$ and applying the algorithm from Corollary~\ref{cor:DP_approx_TD} with $\kappa^*$.
		\STATE Init $\calS \gets \{c\}$.
		\FOR{\textbf{each} $v\in V_\zeta$}
		\STATE Compute $\ell_v\gets {\tt DPMaxProjection}(\epsilon_0,\beta_0,c,v,\kappa)$ approximated to a factor of $(1-\alpha)$.
		\STATE Using the rotation $R_v$ that maps $v$ as the first standard axis, complete the first coordinate $x=\langle c,v\rangle+\ell_v$ to a point $q_v$ by applying $\TDC(\epsilon_0,\delta_0,\beta_0)$ for $d-1$ times.\\
		Add $q_v$ to $\calS$.
		\ENDFOR
		\RETURN the pair $(c,\calS)$
	\end{algorithmic}
\end{algorithm}

\begin{theorem}
	\label{thm:dp_approx_kernel}
	Given $\epsilon>0,\delta>0, \alpha>0,\beta>0$ and $c_d\geq 1$, set 
	\begin{equation}
	\Gamma^{\rm kernel}_{c_d} = \begin{cases}
	O(\frac{d^{\frac 3 2}(\frac{c_d}{\alpha})^{\frac{d-1}{2}}\sqrt{\log(\frac 1 \delta)}\big(\upsilon + d\log(dc_d/\alpha\beta)\big)}{\epsilon}), &~\textrm{using $\epsilon$-DP binary-search}\\
	\tilde O(\frac{d^{\frac 5 2}(\frac{c_d}{\alpha})^{\frac{d-1}{2}}\sqrt{\log(\frac 1 \delta)} \log(\upsilon d c_d/\epsilon\delta\alpha\beta)}{\epsilon}), &~\textrm{using ``Between Thresholds''}\\
	O\left(\frac{d^{\frac 5 2}(\frac{c_d}{\alpha})^{\frac{d-1}{2}}\sqrt{\log(\frac 1 \delta)}\log(\frac{dc_d\upsilon}{\alpha\beta\delta})8^{\log^{*}(\upsilon)}\log^{*}(\upsilon)}{\epsilon}\right), &~\textrm{using ``RecConcave''}
	\end{cases}
	\end{equation} 
	Let $P\subset \G^d$ be a dataset where (i) for $\kappa^* = \kappa + d\alphaqc(\epsilon_0,\delta_0,\beta_0)$, its $\kappa^*$-Tukey region is non empty, and (ii) its $\kappa$-Tukey region is $(c_d,\Gamma_{c_d}^{\rm kernel})$-fat.
	Then Algorithm~\ref{alg:private-kernel} is an efficient $(\epsilon,\delta)$-differentially private algorithm that when applied to $P$ returns w.p. $\geq1-\beta$ a set $\calS$ and a point $c\in \calS$ which satisfies (i) $\calS~\subset~D_P(\kappa-\Gamma^{\rm kernel})$ and (ii) for every direction $u$ it holds that $(1-\alpha)\max_{p\in D_P(\kappa)}\langle p-c,u\rangle \leq \max_{p\in\calS}\langle p-c,u\rangle + \alpha\cdot \width_{\kappa}$.
\end{theorem}
\begin{proof}
	First, we argue this algorithm is efficient. This is clear, since $V_\zeta$ is of size $\propto (c_d/\alpha)^{d-1}$ and (initially and) for each direction we run an efficient $\TDC$-function (as discussed in Section~\ref{sec:TDC}).
	
	Second, this algorithm is $(\epsilon,\delta)$-DP due to the advanced composition theorem of~\cite{DworkRV10}, and the fact that overall we apply $k$ $(\epsilon_0,\delta_0)$-differentially private subprocedures.
	Furthermore, since each subprocedure has a probability of $\beta_0$ of failure, we continue this proof under the assumption that no subroutine has failed and all guarantees are satisfied, which happens with probability $\geq 1 - k\beta_0 = 1-\beta$.
	
	We thus have that $c$ has Tukey depth of $\geq  \kappa^*-d\alphaqc(\epsilon_0,\delta_0,\beta_0)=\kappa$ thus $c\in D_P(\kappa)$. Moreover, for each direction $v\in V_\zeta$ we have $\ell_v$ satisfies
	$(1- \alpha)\max_{p\in D_P(\kappa)} \langle p-c,v\rangle \leq \ell_v \leq \max_{p\in D_P(\kappa-\Delta^{\rm max-proj}(\epsilon_0,\beta_0))}$ where $\Delta^{\rm max-proj}(\epsilon_0,\beta_0) = O(\frac{\log(\frac{\upsilon + \log(d)}{\alpha\beta_0} )}{\epsilon_0}) = O(\frac{\sqrt{k\log(1/\delta)}\log(\frac{k(\upsilon + \log(d))}{\alpha\beta} )}{\epsilon})$ by Theorem~\ref{thm:dp_approx_max_proj}. As $k= O(d(\frac{c_d}{\alpha})^{({d-1})})$ we get that $\Delta^{\rm max-proj} =
	O( \frac{\sqrt{d^3}(\frac{c_d}{\alpha})^{\frac{d-1}2}\log( \upsilon dc_d/\alpha\beta )\sqrt{\log(1/\delta)}}{\epsilon} ) $. Then, by Corollary~\ref{cor:DP_approx_TD} we have that the point $q_v$ retrieved for direction $v$ have Tukey-depth of $\kappa-\Delta^{\rm max-proj}- (d-1)\alphaqc(\epsilon_0,\delta_0,\beta_0) = \kappa-\Gamma^{\rm kernel}$ where
	\begin{align*}
	\Gamma^{\rm kernel} & = O( \frac{d^{\frac 3 2} (\frac{c_d}{\alpha})^{\frac{d-1}2}\log( \upsilon dc_d/\alpha\beta )\sqrt{\log(\frac 1 \delta)}}{\epsilon} ) + O(d\alphaqc(\frac{\epsilon}{\sqrt{d\log(\frac 1 \delta)}(\frac {c_d} {\alpha})^{\frac{d-1}2}}, \frac{\delta}{d(\frac {c_d} {\alpha})^{{d-1}}}, \frac{\beta}{d(\frac {c_d} {\alpha})^{{d-1}}} ))
	\cr & =  O( \frac{d^{\frac 3 2} (\frac{c_d}{\alpha})^{\frac{d-1}2}\log( \upsilon dc_d/\alpha\beta )\sqrt{\log(1/\delta)}}{\epsilon} ) + 
	\begin{cases}
	O(\frac{d\sqrt{d\log(1/\delta)}(\frac {c_d} {\alpha})^{\frac{d-1}2}(\upsilon+\log(d(\frac {c_d} {\alpha})^{{d-1}}/\beta))}{\epsilon})\\
	\tilde O(\frac{d\sqrt{d\log(1/\delta)}(\frac {c_d} {\alpha})^{\frac{d-1}2}\log(d^2(\frac {c_d} {\alpha})^{2(d-1)}\upsilon/\beta\epsilon\delta) }{\epsilon})\\
	O\left(\frac{d\sqrt{d\log(1/\delta)}(\frac {c_d} {\alpha})^{\frac{d-1}2}8^{\log^{*}(\upsilon)}\log^{*}(\upsilon)\cdot \log(\frac{d^2(\frac {c_d} {\alpha})^{2(d-1)}\log^{*}(\upsilon)}{\beta\delta})}{\epsilon}\right)
	\end{cases}
	\cr & = \begin{cases}
	O(\frac{d^{\frac 3 2}(\frac{c_d}{\alpha})^{\frac{d-1}{2}}\sqrt{\log(\frac 1 \delta)}\big(\upsilon + d\log(dc_d/\alpha\beta)\big)}{\epsilon}), &~\textrm{using $\epsilon$-DP binary-search}\\
	\tilde O(\frac{d^{\frac 5 2}(\frac{c_d}{\alpha})^{\frac{d-1}{2}}\sqrt{\log(\frac 1 \delta)} \log(\upsilon d c_d/\epsilon\delta\alpha\beta)}{\epsilon}), &~\textrm{using ``Between Thresholds''}\\
	O\left(\frac{d^{\frac 5 2}(\frac{c_d}{\alpha})^{\frac{d-1}{2}}\sqrt{\log(\frac 1 \delta)}\log(\frac{dc_d\upsilon}{\alpha\beta\delta})8^{\log^{*}(\upsilon)}\log^{*}(\upsilon)}{\epsilon}\right), &~\textrm{using ``RecConcave''}
	\end{cases}
	\end{align*}
	And so, it follows that $\calS\subset D_P(\kappa-\Gamma^{\rm kernel})$.
	
	Next, we argue that for any $u\in\mathbb{S}^{d-1}$ it holds that $(1-\alpha)\max_{p\in D_P(\kappa)}\langle p-c,u\rangle \leq \max_{p\in\calS}\langle p-c,u\rangle + \alpha\cdot \width_{\kappa}$ for the point $c$ chosen by the algorithm (which we know to be in $D_P(\kappa)$). Fix any direction $u$. Let $p^*\in D_P(\kappa)$ be the point obtaining $\max_{p\in D_P(\kappa)}\langle p-c,u\rangle$ and denote $m_u = \langle p^*-c,u\rangle$. Denote $v\in V_\zeta$ as the nearest direction (of angle at most $\zeta$) to $u$, and recall that $\|u-v\|\leq \sqrt 2 \zeta\leq \frac{\alpha}{2c_d}$. It follows that
	\[ \langle p^*-c,v\rangle = \langle p^*-c,u\rangle + \langle p^*-c,v-u\rangle \geq m_u - \|p^*-c\|\|u-v\| \geq m_u - \frac{\alpha\diam_{\kappa}}{2c_d} \]
	As a result, it holds that $\max_{p\in D_P(\kappa)}\langle p-c,v\rangle \geq m_u - \frac{\alpha\diam_{\kappa}}{2c_d}$ and thus $\ell_v \geq (1- \alpha)(m_u - \frac{\alpha\diam_{\kappa}}{2c_d})$. Let $q_v$ be the point in $\calS\subset D_P(\kappa-\Gamma^{\rm kernel})$ which we picked for direction $v$ and whose projection onto $v$ is precisely $\ell_v+\langle c,v\rangle$. We therefore have that	
	Thus,
	\begin{align*}
	\max\limits_{q\in \calS}\langle q-c,u\rangle &\geq \langle q_v-c,u\rangle = \langle q_v-c,v\rangle + \langle q_v-c,u-v\rangle  \geq \ell_v - \diam_{\kappa-\Gamma^{\rm kernel}}\|u-v\|
	\cr & \geq (1- \alpha )(m_u-\frac{\alpha\diam_{\kappa}}{2c_d})- \diam_{\kappa-\Gamma^{\rm kernel}}\frac{\alpha}{2c_d} 
	\cr & \geq (1-\alpha)m_u - \frac{\alpha(\diam_\kappa+\diam_{\kappa-\Gamma^{\rm kernel}})}{2c_d}\geq  (1-\alpha)m_u - \frac{2\alpha\cdot\diam_{\kappa-\Gamma^{\rm kernel}}}{2c_d}
	\cr & \stackrel{\rm fatness}\geq (1-\alpha)m_u -\alpha\cdot \width_{\kappa}
	\end{align*}
\end{proof}

\subsection{Coping with the Relation between $c_d$ and $\Gamma_{c_d}^{\rm kernel}$}
\label{subsec:private_selection}

Theorem~\ref{thm:dp_approx_kernel} implies that we are able to privately output a $(\alpha,\Gamma^{\rm kernel})$-kernel for datasets which are a-priori guaranteed to be $(c_d,\Gamma^{\rm kernel})$-fat, where $\Gamma^{\rm kernel}_{c_d}$ is a function of multiple parameters, including $c_d$. Yet, when $c_d$ is not a-priori given, it is unclear how to verify the ``right'' $c_d$, seeing as it is unclear the diameter of which Tukey-region $\kappa'=\kappa-\Gamma(c_d)$ we should compare to the width of $D_P(\kappa)$. The non-private approach would be to try multiple values of $c_d$ but this leads to multiple (sensitive) queries about the input. 

We offer two solutions to this problem. One is a heuristic, and so~--- while we believe it does work for many datasets~--- it is not guaranteed to always work. This is the solution we discuss in this section. The other, which is guaranteed to work but involves choosing a particular $\kappa$, is discussed at length in the following sections.

The heuristic we pose here is based on the work of Liu and Talwar~\cite{LiuT19}. Basically, we traverse each option of $c_d$, here in powers of $2$, up to $c_{\max} = 4d^{5/2}\cdot 5^d\cdot(d!)$ (in Section~\ref{sec:approx_bounding_box} the choice for this particular parameter is explained), and for each value check whether the diameter of the suitable $D_P(\kappa-\Gamma^{\rm kernel})$ is upper bounded by $c_d\cdot \width_\kappa$ or not.

Formally, we set $t=\lceil\log_2(c_{\max})\rceil =\lceil\log_2(4d^{5/2}\cdot 5^d\cdot(d!))\rceil$ as the number of levels we test. For each $1\leq i \leq t$ let $M_i$ be the $2\epsilon$-DP mechanism that works as follows:
\begin{enumerate}
	\item Set $c_i = 2^{i}$, set $\Gamma_i = \Gamma^{\rm kernel}_{c_i}$ according to Theorem~\ref{thm:dp_approx_kernel}. 
	\item Run Algorithm~\ref{alg:ApproximateDiameter} with Tukey-depth parameter of $\kappa_i^1=\kappa-\Gamma_i$ and failure probability of $\frac{\beta}{12t\ln(2/\beta)}$. Denote the result $D$, and by Theorem~\ref{thm:DP_alg_approx_diam_Tukey_region} we know that w.p. $\geq 1-\frac{\beta}{12t\ln(2/\beta)}$ it holds that $D\geq (1-\alpha)\diam_{\kappa-\Gamma_i}$.
	\item Run Algorithm~\ref{alg:DP_approx_width_tukey_region} with a lower bound of $B = D/c_i$, with Tukey-depth parameter of $\kappa_2=\kappa + \Delta^{\width}(\epsilon,\frac{\beta}{12t\ln(\frac2 \beta)}) = \kappa + O(\frac{\log(\frac{t(\upsilon+\log(d))}{\alpha\beta})}{\epsilon})$ as defined in Theorem~\ref{thm:DP_alg_approx_width_Tukey_region}, and with failure probability of $\frac{\beta}{12t\ln(2/\beta)}$. Denote the result $w$, and by Theorem~\ref{thm:DP_alg_approx_width_Tukey_region} we know that w.p. $\geq 1-\frac{\beta}{12t\ln(2/\beta)}$ it holds that $w\leq (1+\alpha)\width_\kappa$ if $w\neq 0$.
	\item return the tuple $(i, \tau)$ when the score $\tau$ is set be $\tau=2^{-i}$ if $D \leq c_i \frac{1-\alpha}{1+\alpha}w$, or set as $\tau=0$ if any of the returned values is $0$ or if $D > c_i \frac{1-\alpha}{1+\alpha}w$.
\end{enumerate}

As in the work of Liu and Talwar~\cite{LiuT19} we define the $2\epsilon$-DP algorithm $Q$ which picks $i\in [t]$ u.a.r and runs $M_i$. Setting $\gamma = 1/3t$, we define $\tilde Q$ as the mechanism that works as follows:
\begin{itemize}
	\item Repeat:
	\begin{enumerate}
		\item Run $Q$, namely pick $i\in [t]$ u.a.r and add its output $M_i(P)$ to a (multi-)set $S$.
		\item Toss a biased coin: w.p. $\gamma$ output an element in $S$ with maximal $\tau$ and halt.
	\end{enumerate}
\end{itemize}
Applying Theorem~3.2 from~\cite{LiuT19}, we infer that $\tilde Q$ is $6\epsilon$-DP. Moreover, we can argue the following about its utility.
\begin{claim}
	\label{clm:private_selection_of_a_good_c}
	W.p. $\geq 1-\beta$, if $\tilde Q$ returns an index $i$ with score $\tau>0$ then $D_P(\kappa)$ is $(c_i, \Gamma_i)$-fat. 
	Furthermore, denote 
	\[Good =\{1\leq i \leq t: D_P(\kappa) \textrm{ is }\big(\tfrac{1-\alpha}{1+\alpha}c_i, \Delta^{\width}(\epsilon,\tfrac{\beta}{12t\ln(\frac2 \beta)}),\Gamma_i+\Delta^{\diam}(\epsilon,\tfrac{\beta}{12t\ln(\frac2 \beta)})\big)\textrm{-fat}\}\] If $Good\neq\emptyset$ then, assuming $\beta<1/6$, w.p. $\geq 1/2$ the mechanism $\tilde Q$ returns $i\leq \min\{i\in Good\}$.
\end{claim}
%Before proving the claim, we make two comments. First, assume $D_P(\kappa)$ is $(c,(\kappa_2-\kappa), \Gamma(c))$-fat for some $c>0$, namely that $D_P(\kappa_2)$ is $\Delta^{\rm width}(\epsilon,\frac{\beta}{12t\ln(\frac2 \beta)})+\Gamma(c)$-fat. Then for $i$ where $2^i< c\leq 2^{i+1}$ we must also have that $D_P(\kappa_2)$ is $(2^{i+1}, \Delta^{\rm width}(\epsilon,\frac{\beta}{12t\ln(\frac2 \beta)})+\Gamma(2^i))$-fat. Thus, the set $Good$ is non-empty and we return $i$ which is a factor of $2$ approximation of the smallest $c$ satisfying $(c,(\kappa_2-\kappa), \Gamma(c))$-fatness.
%We comment that standard amplification (running $\tilde Q$ for $O(\log(1/\beta))$ times, rescaling the privacy budget accordingly) makes $\tilde Q$ have success probability of $1-\beta$.
\begin{proof}
	First, algorithm $\tilde Q$ halts the very first time its biased coin comes up heads. Clearly, the probability it iterates for $\geq \ln(2/\beta)/\gamma =3t\ln(2/\beta)$ is at most $(1-\gamma)^{\ln(2/\beta)/\gamma}\leq e^{-\ln(2/\beta)}\leq \beta/2$. Second, since we set the failure probability of each $M_i$ to be upper bounded by $2\cdot \beta/(12t\ln(2/\beta))$ then the probability that in any of the $\leq 3t(\ln(2/\beta))$ iterations one of the executed $M_i$ fails is at most $\beta/2$. It follows that w.p. $\geq1-\beta$ the algorithm $\tilde Q$ returns some score of some successfully executed $M_i$.
	
	We continue under the assumption that this event indeed holds. Now, if $\tilde Q$ returns $i$ with $\tau>0$ it follows that for this $i$ it indeed holds that $(1-\alpha)\diam_{\kappa-\Gamma_i}\leq D \leq c_i \frac{1-\alpha}{1+\alpha}w \leq (1-\alpha)c_i {\width_\kappa}$, and so $\diam_{\kappa-\Gamma_i}\leq c_i \width_{\kappa}$, hence $D_P(\kappa)$ if $(c_i,\Gamma_i)$-fat.
	
	Next, assume $Good\neq\emptyset$ and consider $i^* = \min\{i\in Good\}$. If indeed $M_{i^*}$ is executed and no failure occurs then it must be that $D\leq \diam_{\kappa-\Gamma_i - \Delta^{\diam}(\epsilon,\frac{\beta}{12t\ln(\frac2 \beta)})}$; similarly, it must be that $w\geq \width_{\kappa+\Delta^{\width}(\epsilon,\frac{\beta}{12t\ln(\frac2 \beta)})}$. Since $i^*\in Good$ then this means that $D\leq w\frac{1-\alpha}{1+\alpha}c_i$, and so we place $i^*$ with score of $2^{-i^*}$ in $S$. This means that $\tilde Q$ must return some $i$ with a higher score, namely $2^{-i}\geq 2^{-i^*}$, so $i\leq i^*$.
	
	Well, what is the probability that $\tilde Q$ does \emph{not} execute $i^*$?
	\begin{align*}
	\Pr[\tilde Q \textrm{ never picks }i^*] &= \sum_{j=1}^{\ln(2/\beta)/\gamma} \Pr[\tilde Q \textrm{ never picks }i^* \textrm{ and iterates $j$ times}]
	= \sum_{j=1}^{\ln(2/\beta)/\gamma} (1-\frac 1 t)^j(1-\gamma)^{j-1}\gamma 
	\cr &= \gamma(\frac{t-1}t)\sum_{j=1}^{\ln(2/\beta)/\gamma} [(1-\frac 1 t)(1-\gamma)]^{j-1} \leq \gamma(\frac{t-1}t)\sum_{j\geq 0} [(1-(\frac 1 t+\gamma-\frac{\gamma}{t})]^{j-1}
	\cr &= \frac{ \frac{\gamma(t-1)}{t} }{\frac{1}{t}+\gamma - \frac{\gamma}t} = \frac{\gamma(t-1)}{1+\gamma(t-1)} \leq \gamma(t-1) \leq \frac{t-1}{3t}\leq \frac 1 3
	\end{align*}
	Altogether, the probability of $\tilde Q$ to never run $M_{i^*}$ is upper bounded by $\frac 1 3 + \beta \leq 1/2$. Thus, w.p. $\geq 1/2$ we return $i\leq i^*$.
\end{proof}

\section{Private Approximation of the Bounding Box of $D_P(\kappa)$}
\label{sec:approx_bounding_box}

In this section we give a differentially private algorithm that returns a transformation that turns $D_P(\kappa)$ into a fat Tukey-region, if it is the case that the volume of $D_P(\kappa)$ and the volume of some shallower $D_P(\kappa')$ are comparable. The transformation itself is based on (privately) finding an approximated bounding-box for $D_P(\kappa)$, which is of an independent interest. Once such a box is found, then the transformation $T$ is merely a linear transformation, composed of rotation and axes scaling, that maps the returned box $\B$ to the hypercube $[0,1]^d$. We thus focus in this section on a private algorithm that gives a good approximation of the bounding box of $D_P(\kappa)$.

\subsection{A Non-Private Bounding-Box Approximation algorithm}
\label{subsec:non-private_bounding_box}

Before giving our differentially-private algorithm for the bounding-box approximation, we present its standard, non-private, version. For brevity, we discuss an algorithm that returns a box $\B$ that bounds the convex-hull of the given set of points $P$, which is precisely $D_P(1)$. We present the algorithm from~\cite{Har-Peled11} which gives a $c_d$-approximation of the bounding box, with $c_d$ denoting some constant depending solely on $d$. Namely, denoting $\Bopt_1$ as the box of minimal volume out of all boxes that contain $D_P(1)$, this algorithm returns a box $\B$ which is a bounding box for $D_P(1)$ and satisfies
\[  \vol(\Bopt_1)\leq \vol(\B) \leq c_d\vol(\Bopt_1) \]
The algorithm is given below. (Note that its first step is described in a black-box fashion.)

\begin{algorithm}[H]
	\caption{\label{alg:non-private-bounding-box-approx} Non-Private Approximation of the Bounding Box}
	{\bf Input}: Dataset $P\subset \mathbb{R}^{d}$; Approximation parameter $\gamma>1$.
	\begin{algorithmic}[1]
		\STATE Find two points $s,t\in P$ satisfying $\|s-t\| \geq \diam(P)/\gamma$. Denote $u_{st} = \frac{t-s}{\|t-s\|}$.
		\STATE $I \gets [\min_{x\in P} \langle x,u_{st}\rangle, \max_{x\in P}\langle x,u_{st}\rangle]$.
		\IF {$P$ is one dimensional}
		\RETURN $I$.
		\ENDIF
		\STATE Compute $\Pi^{\perp u_{st}}$, the projection on the subspace orthogonal to $u_{st}$.
		\STATE Recurse on $\Pi^{\perp u_{st}}(P)$ and obtain its $(d-1)$-dimensional bounding box $\B'$.
		\RETURN $\B \gets \B' \times I$.
	\end{algorithmic}
\end{algorithm}

\begin{claim}[\label{clm:bounding_box_non_private} Lemma 18.3.1 \cite{Har-Peled11}, restated]
	Fix $\gamma>1$. Given $s,t\in D_P(1)$ such that the segment connecting the two points is of length $\geq \frac{\diam(D_P(1))}{\gamma}$, then Algorithm \ref{alg:non-private-bounding-box-approx} returns a box $\B$ bounding $D_P(1)$ s.t. $\vol(\B)\leq \gamma^{d-1}(d!)\cdot \vol(D_{1})$.
\end{claim}
\begin{proof}
	The proof works by induction on $d$, where for $d=1$ it is evident that $|I|$ is the minimal convex $1$-dimensional body that holds the data. Now fix any $d>1$. Recall that $D_P(1)$ is the convex hull of $P$ and let $st\in D_P(1)$ be the segment such that $||s-t||\geq \frac{diam}{\gamma}$. Wlog (we can apply rotation) $st$ lies on the $x_{d}$-axis (i.e., the line $\ell\equiv\cup_{x}(0,...,0,x)$). Thus, $\Pi$ is a projection onto the hyperplane  $h\equiv x_{d}=0$, and $I$ is projection of $P$ onto the $x_d$-axis.
	
	By the induction hypothesis, the returned box $\B'$ from the recursive call is (i) a bounding box for $\Pi(P)$ and (ii) has volume $\vol(B')\leq \gamma^{d-2}\cdot (d-1)!\cdot \vol(Q)$ where we use $Q$ to denote the convex hull of $\Pi(P)$ (which is contained in $h$). As a result of (i) we have that any $p\in P \subset \B'\times I$ and we returns a bounding box for $D_P(1)$, hence $\vol(D_P(1))\leq \vol(\B^*_1)\leq \vol(\B_1)$. We thus upper bound the volume of $\B$. Our proof requires Fact~\ref{fact:volume_bounds_by_projections}. (For any convex body in $\R^{d}$, let $x$ be the length the longest segment in direction $u$ and $Y$ be the volume of its projection onto the subspace orthogonal to $u$; then the volume of this body is $\geq x\cdot Y/d$.)
	
	Given a point $q \in Q\subset \R^{d-1}$,	
	let $\ell_{q}$ be the line parallel to $x_{d}$-axis passing through $q$. Let $L(q)$ be the minimum value of $x_{d}$ for the points of $\ell_{q}$ lying inside $D_P(1)$, and similarly, let $U(q)$ be the maximum value of $x_{d}$ for the points of $\ell_{q}$ lying inside $D_P(1)$; and let $f(q)$ be their difference. In other words, $D_P(1)\cap \ell_{q} = [L(q), U(q)]$. We thus have $\vol(D_P(1)) = \int_{q\in Q} U(q)dq - \int_{q\in Q} L(q)dq = \int_{q\in Q} f(q)dp$. So consider the body $C$ which is bounded by the hyperplance $h\equiv x_d=0$ on the one side and the curve $f(q)$ on the other side, whose volume is precisely the volume of $D_P(1)$. First, since $D_P(1)$ is convex then so is $C$. More importantly, since both $s$ and $t$ belong to $D_P(1)$ then the line (orthogonal to $h$) is inside this convex body and its length is at least $\|s-t\|$. Therefore, we apply Fact~\ref{fact:volume_bounds_by_projections} and have that
	\begin{align*}
	\vol(D_P(1)) = \vol(C) &\geq \vol( Q )\frac{\|s-t\|}d 
	%\cr & 
	\geq \vol( Q ) \cdot \frac{\diam(D_P(1))}{\gamma d}
	%\cr & 
	\stackrel{\rm induction}\geq \frac{\vol( \B' )}{\gamma^{d-2}(d-1)!} \cdot \frac{|I|}{\gamma d} = \frac{\vol(\B)}{\gamma^{d-1}\cdot d!}\qedhere
	\end{align*}
\end{proof}

\subsection{A Private Algorithm for a Bounding Box for $D_P(\kappa)$}
\label{subsec:private_bounding_box_approx}

Leveraging on the ideas from the non-private Algorithm~\ref{alg:non-private-bounding-box-approx}, in this section we give our differentially private algorithm that approximates a bounding box for $D_P(\kappa)$. As ever, our algorithm's guarantee relates to both the volume of $D_P(\kappa)$ and the volume of $D_P(\kappa-\Delta)$. Formally, we return (w.h.p) a box $\B$ which is a $(c_d,\Delta)$-approximation, defined as a bounding box that holds $D_P(\kappa)$ and where
\begin{equation}
\label{eq:dp_bounding_box_guarantee}
 \vol(D_P(\kappa))\leq \vol(\B^*_{\kappa})\leq \vol(\B) \leq c_d\cdot \vol(D_P(\kappa-\Delta)) 
\end{equation}
with $\B^*_\kappa$ denoting the bounding box of $D_P(\kappa)$ of minimal volume.

The algorithm we give mimics the idea of the recursive Algorithm~\ref{alg:non-private-bounding-box-approx}, where it is crucial that in each level of the recursion we find two points \emph{inside} the convex body $D_P(\kappa-\Delta)$. While it seems like a subtle point, if were to, say, use a more na\"ive approach of finding a direction where the projection onto it is proportional in length to the diameter of $D_P(\kappa)$ then we run the risk of finding a box whose volume is much bigger in comparison to the volume of $D_P{(\kappa-\Delta)}$. (Several other attempts of finding ``good'' directions based on the $\lTDC$-query failed as well due to pathological polytopes where the given direction isn't correlated with a required pair of points.)

In fact, it is best we summarize what is required on each level of our recursive algorithm. We find a segment $\bar{st}$ and an interval $I$ on the line extending this segment where the following three properties must hold: (i) both $s,t\in D_{P}(\kappa-\Delta)$, (ii) the length of $I$ is proportional to $\|s-t\|$ and (iii)  $\forall x\in D_P(\kappa)$, the projection of $x$ onto the line extending $\bar{st}$ lies inside $I$. Property (iii) asserts that $D_P(\kappa)$ is contained inside the box we return; property (i) combined with Fact~\ref{fact:volume_bounds_by_projections} allows us to infer that $\vol(D_P(\kappa-\Delta)) \geq \|s-t\| \cdot \vol( \Pi^{\perp st} D_{P}(\kappa-\Delta) )/d$ where $\Pi^{\perp st}$ is the projection onto the subspace orthogonal to the line extending $\bar{st}$; and property (ii) asserts $\|s-t\|\geq |I|/c$ so that recursively we get $c^d\cdot d!$ approximation of the volume of $D_P(\kappa-\Delta)$. Thus, asserting that these properties hold w.h.p.~becomes the goal of our algorithm. Details are given in the algorithm below and in the following theorem.

\begin{algorithm}[hbt]
	\caption{\label{alg:private-bounding-box-approx} Private Approximation of the Bounding Box of $D_P(\kappa)$}
	{\bf Input}: Dataset $P\subset \mathbb{R}^{d}$; privacy parameters $\epsilon>0, \delta\geq0$, failure probability $\beta>0$, desired Tukey-depth parameter $\kappa$.
	\begin{algorithmic}[1]
		\STATE Set $\zeta \gets \arccos(1/10)$, $V_\zeta$ as a $\zeta$-angle cover, and $k \gets 2+\sum_{j=2}^d(j+2+j-1) = d^2+2d-1 $.
		\STATE $\epsilon_0 \gets \frac{\epsilon}k, \delta_0 \gets \frac{\delta}k, \beta_0\gets\frac{\beta}k$.
		\STATE Construct the chain of convex polytopes $C_1 \supset C_2 \supset C_3 \supset ...$ where $C_i = D_P(i)$ for every $i$.
		\FOR {$j$ from $d$ downto $2$}
		\STATE Set $\kappa' = \kappa + j\alphaqc(\epsilon_0,\delta_0,\beta_0)$.\\ 
		Run $j$-times the $(\epsilon_0,\delta_0)$-differentially private algorithm to approximation $\TDC$ to obtain the $j$-coordinates of a point in $C_{\kappa' - j\alphaqc(\epsilon_0,\delta_0,\beta_0)} = C_{\kappa}$. 
		Denote the returned point as $s$.
		%\\$s\gets {\tt DPPointInTukeyRegion} (\kappa + j\alphaqc(\epsilon_0,\delta_0,\beta_0), j\epsilon_0, j\delta_0, j\beta_0)$ with respect to $C_1, C_2,...$
		\STATE Run ${\tt DPTukeyDiam}(\kappa,\epsilon_0,\beta_0)$ to obtain a parameter $\ell$ s.t. $0.9 \diam(C_\kappa)\leq \ell \leq \diam(C_{\kappa-\Delta^{\diam}(\epsilon_0,\beta_0)})$
		\STATE Run the $\epsilon_0$-differentially private algorithm ${\tt DPLargeTDCDirection}$ for the point $s$, the distance $\lambda = 0.45\ell$, the directions in $V_\zeta$ and depth $\kappa-\Delta^{\diam}(\epsilon_0,\beta_0)$.\\
		Denote the returned direction as $v$.
		\STATE Using $j-1$ calls to the $(\epsilon_0,\delta_0)$-DP algorithm that approximates $\TDC$, complete the coordinate $\left(\langle s,v\rangle + 0.45 \ell\right)$ on the $v$-axis to a point $t$.
		\STATE Set $u = \frac{t-s}{\|t-s\|}$ and $\Pi^{\perp u}$ as the projection orthogonal to $u$.
		\STATE Set $I_j \gets [\langle s,u\rangle - \frac {10} {9} \ell, \langle s,u\rangle +\frac {10}{9}\ell]$ in direction $u$.
		\STATE Project the polytopes: for all $j$ we set $C_j \gets \Pi^{\perp u}(C_j)$.
		\ENDFOR\\ 
		\%\% {and in dimension $1$ (on a line)}
		\STATE $s\gets$ {\tt DPPointInTukeyRegion} $(\kappa + \alphaqc(\epsilon_0,\delta_0,\beta_0))$ with respect to closed intervals $C_1, C_2,...$
		\STATE $\ell\gets$ {\tt DPTukeyDiam}$(\kappa,\epsilon_0,\beta_0)$ so that $0.9 \diam(C_\kappa)\leq \ell \leq \diam(C_{\kappa-\Delta^{\diam}(\epsilon_0,\beta_0)})$
		\STATE $I_1 \gets [s-\frac {10} 9 \ell, s+\frac {10}9\ell]$
		\RETURN The box $\B = \bigtimes\limits_{j=1}^d I_j$.
	\end{algorithmic}
\end{algorithm}

\begin{theorem}
	\label{thm:dp_bounding_box}
	Algorithm~\ref{alg:private-bounding-box-approx} is $(\epsilon,\delta)$-differentially private. Moreover, let $P\subset \G^d$ be a set of points whose Tukey-region $\kappa + d\alphaqc(\frac\epsilon{d^2+2d-1},\frac{\delta}{d^2+2d-1},\frac{\beta}{d^2+2d-1})$ is non-empty. Then w.p. $\geq 1-\beta$ Algorithm~\ref{alg:private-bounding-box-approx} returns a box $\B$ where $D_P(\kappa)\subset \B$ and $\vol(D_P(\kappa))\leq \vol(\B)\leq 5^d(d!)\vol(D_{p}(\kappa-\Delta^{\rm BB}))$ for \[\Delta^{\rm BB}(\epsilon,\delta,\beta) = \begin{cases}
	O(\frac{d^3(\upsilon + \log(d/\beta))}{\epsilon}), &\textrm{Using $\epsilon$-DP binary search}\\
	\tilde O(\frac{d^3{\log(d\upsilon/\epsilon\delta\beta)}}{\epsilon}), &\textrm{Using the ``Between Threshold'' Alg}\\
	O(\frac{d^3{\log(d\upsilon/\beta)}}{\epsilon}+ \frac{d^38^{\log^*(\upsilon)}\log^*(\upsilon)\log(d\log^*(\upsilon)/\delta\beta)}{\epsilon}), &\textrm{Using the ``RecConcave'' algorithm}
	\end{cases}\]
\end{theorem}
\begin{proof}
	First, this algorithm is clearly $\epsilon$-DP as in each of the $d-1$ iterations we run differentially private subroutines, and all that is required is that we apply basic composition on the various calls. For each $j$, (i) finding $s$ entails the call to $\TDC$-approximation $j$-times, (ii) finding $\ell$ requires another call, (iii) finding $u$ requires another call and (iv) finding $t$ requires another $j-1$ calls (to find its remaining $j-1$ coordinates); lastly (in dimension $1$) we invoke two more calls for finding an interior point $s$ and for finding $\ell$. All in all, we invoke $\left(\sum_{j=2}^d j+1+1+j-1 \right)+2=d^2+2d-1$ calls to differentially private procedures, each with privacy parameters  $(\frac{\epsilon}{d^2+2d-1}, \frac{\delta}{d^2+2d-1})$, so by basic composition, this algorithm is $(\epsilon,\delta)$-differentially private. Similarly, each subroutine has $\beta/(d^2+2d-1)$ probability of failure, so the rest of the proof continues under the assumption that the event in which no subroutine has failed holds, which happens w.p. $\geq1-\beta$.
	
	Under this event, we argue that for
	\begin{align*}
	\Delta^{\rm BB} &= \Delta^{\diam}(\frac{\epsilon}{d^2+2d-1},\frac{\beta}{d^2+2d-1}) + \Delta^{\rm LargeTDCDir}(\frac{\epsilon}{d^2+2d-1},\frac{\beta}{d^2+2d-1}, (1/\zeta)^{d-1}) 
	\cr &~~~+ d\alphaqc(\frac{\epsilon}{d^2+2d-1},\frac{\delta}{d^2+2d-1},\frac{\beta}{d^2+2d-1})
	\cr & = O(\frac{d^3{\log(d\upsilon/\beta)}}{\epsilon}) + O(d^3 \alphaqc(\epsilon,\frac{\delta}{3d^2},\frac{\beta}{3d^2}))
	\cr & = \begin{cases}
	O(\frac{d^3(\upsilon + \log(d/\beta))}{\epsilon}), &~\textrm{Using $\epsilon$-DP binary search}\\
	\tilde O(\frac{d^3{\log(d\upsilon/\epsilon\delta\beta)}}{\epsilon}), &~\textrm{Using the ``Between Threshold'' Algorithm}\\
	 O(\frac{d^3{\log(d\upsilon/\beta)}}{\epsilon}+ \frac{d^38^{\log^*(\upsilon)}\log^*(\upsilon)\log(d\log^*(\upsilon)/\delta\beta)}{\epsilon}), &~\textrm{Using the ``RecConcave'' algorithm}
	\end{cases}
	\end{align*}
	it holds that $D_P(\kappa) \subset \B$ and that $\vol(\B)\leq 5^d\cdot (d!)\cdot \vol(D_P(\kappa-\Delta^{\rm BB}))$.
	
	First, it is rather simple to see that $D_P(\kappa)\subset \B$. In each iteration, under the non-failure event, $s\in C_\kappa$ since its dimension is decremented with each iteration and since, by assumption, $C_{\kappa + d\alphaqc}$ isn't empty. Moreover, since we know $\ell \geq 0.9 \diam(C_\kappa)$ then for any $x\in C_\kappa$ and for any direction $u$ it holds that $|\langle x-s,u\rangle| \leq \frac {10} 9 \ell$. Thus, for the particular direction $u$ we chose it also holds that $\forall x\in C_\kappa, \langle x,u\rangle \in I_j$. Applying this to all dimensions, we get that $D_P(\kappa) \subset \bigtimes\limits_{j=1}^d I_j$.
	
	Now, consider $\ell$ which we know to be upper bounded by $\diam(C_{\kappa-\Delta^{\diam}})$. Denote $a,b\in C_{\kappa-\Delta^{\diam}}$ as the two points whose distance is the diameter of $C_{\kappa-\Delta^{\diam}}$. Since $s\in C_{\kappa}\subset C_{\kappa-\Delta^{\diam}}$ then, by the triangle inequality, the distance of $s$ to either $a$ or $b$ is at least $\ell/2$, so assume wlog that $\|s-a\|\geq \ell/2$. Let $v\in V_\zeta$ be the closest direction to the line connecting $s$ and $a$, and so the projection onto $v$ is at least $0.9 \cdot \frac \ell 2 = 0.45\ell$. It follows that some point $q\in C_{\kappa-\Delta^{\diam}}$ is such that its projection onto some $v\in V_\zeta$ is at least $0.45\ell$. Thus, {\tt DPLargeTDCDirection} returns some direction $v$ where the coordinate $\langle s,v\rangle +0.45\ell$ can be extended to a point of depth  $\geq \kappa - \Delta^{\diam} - \Delta^{\rm LargeTDCDir}$. Finally, the repeated invocation of $\TDC$ returns such a point $t$ with depth $\geq \kappa-\Delta^{\rm BB}$. We infer that both $s$ and $t$ belong to $C_{\kappa-\Delta^{\rm BB}}$ and that $\|s-t\| \geq |\langle s-t,v\rangle| \geq 0.45\ell$. 
	
	We can now apply Fact~\ref{fact:volume_bounds_by_projections} to infer that
	\begin{align*}
	\vol(C_{\kappa-\Delta^{\rm BB}}) &\geq \vol( \Pi^{\perp u}(C_{\kappa-\Delta^{\rm BB}}) )\cdot\frac{ \|s-t\| }d \geq \vol( \Pi^{\perp u}(C_{\kappa-\Delta^{\rm BB}}) )\cdot \frac {0.45\ell} {d} 
	\cr &\geq \vol( \Pi^{\perp u}(C_{\kappa-\Delta^{\rm BB}}) )\cdot \frac{|I_d|}{20/9}\cdot\frac{0.45}{d} > \vol( \Pi^{\perp u}(C_{\kappa-\Delta^{\rm BB}}) )\cdot \frac{|I_d|}{5d}
	\end{align*}
	and so, by induction we get that $\vol(D_P(\kappa-\Delta^{\rm BB})) \geq \frac{\bigtimes\limits_{j=1}^d |I_j|} {5^d\cdot d!} = \frac{\vol(\B)} {5^d\cdot d!}$
\end{proof}

\subsection{From a Bounding Box to a ``Fat'' Input}
\label{subsec:bounding_box_to_fatness}

In classic, non-private, computational geometry, the bounding-box approximation algorithm can be used to design an affine transformation that turns the input dataset into a fat. In more detail, in the non-private setting, one works with the convex-hull of the input $D_P(1)$ and finds a bounding box $\B$ such that $\vol(B)\leq 2^d\cdot d!\cdot \vol(D_P(1))$. Then, denoting $T$ as the linear transformation that maps $\B$ to the cube $[0,1]^d$, we can argue that the resulting dataset $T(P)$ is fat, namely that $\width(T(P))\geq \diam(T(P))/c_d$ for $c_d = 2^d\cdot (d!) \cdot d^{\frac 5 2}$. (Moreover, once we have a kernel for (the fat) $T(P)$ we can reshape it back into a kernel for $P$ using $T^{-1}$. Afterall, if two convex polytopes we have ${\cal A}\subset {\cal B}$ then any $x\in {\cal A}$ can be expressed as convex combination of the vertices of ${\cal B}$, thus, for any linear transformation we have that any $T(x)$ can be expressed as convex combination of the vertices of $T({\cal B})$, hence $T({\cal A})\subset T({\cal B})$.)

\begin{figure}[b]
	\begin{minipage}[c]{0.6\textwidth}
		\caption{\label{fig:different_volumes} 
			An example where $D_P(\kappa-\Delta^{\rm BB})$ has a much larger volume than $D_P(\kappa)$. This makes it so that mapping $\B$ to $[0,1]^d$ doesn't change the width of $D_P(\kappa)$ significantly.}
	\end{minipage}
\hfill\begin{minipage}[c]{0.4\textwidth}
	\centering
	\includegraphics[scale=0.3]{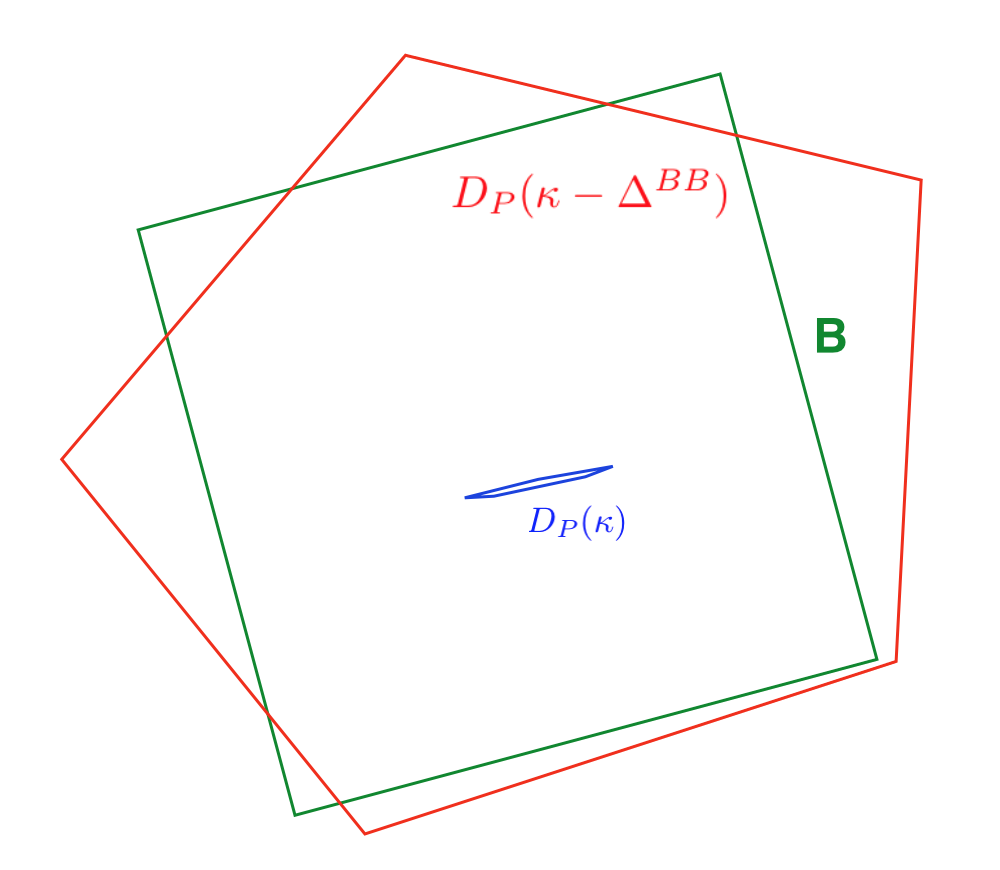}
\end{minipage}
\end{figure}
Unfortunately, we cannot make a similar claim in our setting. Granted, our bounding box is $(c_d,\Delta)$-approximation, but the resulting affine transformation does not, always, guarantee that applying it turns $D_{P}(\kappa)$ to be $(c_d',\Delta)$-fat or $c_d'$-absolutely fat. This should be obvious~--- as Figure~\ref{fig:different_volumes} shows, there exists settings where $D_P(\kappa-\Delta)$ is drastically bigger than $D_P(\kappa)$, resulting in a linear transformation that doesn't ``stretch'' $D_P(\kappa)$ enough to make it fat. Luckily, we show that non-comparable volumes is the only reason this transformation fails to produce a fat Tukey-region. As the following lemma shows, when $\vol(D_P(\kappa))$ and $\vol(D_P(\kappa-\Delta))$ are comparable, then the bounding box yields a linear transformation that does make the $\kappa$-Tukey region fat.

\begin{lemma}
	\label{lem:fat_Tukey_region_when_volumes_are_comparable}
	Fix $\epsilon>0$, $\delta\geq 0$ and $\beta>0$, and define $\Delta^{\rm BB}$ as in Theorem~\ref{thm:dp_bounding_box}.
	Suppose $P\subset \G^d$ is such that for two parameters $\kappa\geq \kappa'$ where $\kappa-\kappa'\geq \Delta^{\rm BB}$ are such that $\vol(D_P(\kappa))\geq \frac 1 2 \vol(D_P(\kappa'))$. Then there exists a $(\epsilon,\delta)$-differentially private algorithm that w.p. $\geq1-\beta$ computes (i) an affine transformation $M$ that turns $M(D_P(\kappa))$ into a convex polytope which is $(c_d,\kappa-\kappa')$-fat, for $c_d = 4 d^{\frac 5 2 }5^d\cdot (d!)$, and (ii) a transformation $\tilde M$ making $\tilde M(D_P(\kappa))$ $2d\cdot 5^d\cdot (d!)$-absolutely fat.
\end{lemma}
We refer to a pair of values $(\kappa,\kappa')$ for which $\vol(D_P(\kappa))\geq \frac 1 2 \vol(D_P(\kappa'))$ as a ``good'' pair. Next, in Section~\ref{sec:finding_good_kappa} we deal with finding such a good pair of $(\kappa,\kappa')$. But for now we assume that the given $\kappa,\kappa'$ are good and prove Lemma~\ref{lem:fat_Tukey_region_when_volumes_are_comparable} under this assumption. Note that the lemma is not vacuous~--- even if $\vol(D_P(\kappa))\geq \frac 1 2 \vol(D_P(\kappa'))$ it does not mean that $\width_\kappa$ is proportional to $\diam_{\kappa'}$ as both Tukey-regions might be ``slim.''
\begin{proof}
	By Theorem~\ref{thm:dp_bounding_box}, w.p. $\geq 1-\beta$, we can compute a box $\B$ s.t. $D_P(\kappa)\subset \B$ and $\vol(\B)\leq 5^d\cdot d!\cdot \vol(D_P(\kappa-\Delta^{\rm BB}))$. Let $T$ be the affine transformation mapping $\B$ into the unit hypercube $[0,1]^{d}$, and note that as a linear transformation $T$ maps the convex polytope $D_P(\kappa)$ into a (different) convex polytope. It is left to show that $T(D_P(\kappa))$ is fat. 
	
	Comparing the volumes of the different regions, we have that
	\[1= \vol(T(\B)) \leq 5^d\cdot d! \cdot \vol(T(D_P(\kappa-\Delta^{\rm BB}))) 
	\leq  5^d\cdot d! \cdot \vol(T(D_P(\kappa')))\leq  2\cdot 5^d\cdot (d!) \cdot \vol(T(D_P(\kappa)))\]
	where the last inequality follows from the fact that for any region $S$ we have that $\vol(T(S)) = \vol(S)\cdot |\det(T)|$.\footnote{Abusing notation, as $T$ is a combination of a linear transformation and a shift which doesn't change volumes.}
	To proceed, we require the following fact.
	\begin{fact}
		\label{fact:volume_of_hyperplane_intersecting_unit_cube}
		Let $h$ be any hyperplane, and let $h_C$ be the projection of the unit-cube $[0,1]^d$ onto $h$. Then $\vol(h_C)\leq d$.
	\end{fact}
	\noindent(Its proof relies on Fact~\ref{fact:volume_bounds_by_projections}, where for any direction $u$, the longest segment inside $[0,1]^d$ in direction $u$ must be at least $1$.)
	 
	Now, let $v$ be the unit direction vector on which $\width(T(D_P(\kappa)))$ is obtained. Denote $h^{*}$ be the hyperplane orthogonal to $v$, and denote $h_D$ as the projection of $T(D_P(\kappa))$ onto $h^*$. Note that $\vol(h_D) \leq \vol(h_C)\leq d$. It follows that  $\vol(T(\mathcal{D}(\kappa)))$ is upper bounded by the volume of the ``cylinder'' whose base is $h_D$ and height $\width(T(D_P(\kappa)))$. Namely, we have that $\vol(T(D_P(\kappa))) \leq d \cdot \width(T(D_P(\kappa)))$. All in all, we get that  $\width(T(D_P(\kappa)))\geq \frac{1}{2\cdot d \cdot 5^d\cdot d!}$.
	
	Next, consider $u$ as the direction on which $\rho=\diam(T(D_P(\kappa')))$ is obtained (connecting the two vertices whose distance is the diameter). Let $h$ be the hyperplane orthogonal to $u$ and let $Y = \Pi^h(T(D_P(\kappa')))$. Again, by Fact~\ref{fact:volume_bounds_by_projections} we have that $\vol(T(D_P(\kappa'))) \geq \vol(Y)\cdot \rho /d$. In contrast, the volume of $T(D_P(\kappa))$ is upper bounded by the volume of the ``cylinder'' whose base is $X = \Pi^h(T(D_P(\kappa)))$ and height is the projection of $D_P(\kappa)$ onto $u$, which is clearly upper bounded by $\diam(T(D_P(\kappa)))$. Since $T(D_P(\kappa))\subset T(D_P(\kappa'))$ then their respective projections, $X$ and $Y$, are contained in one another. Putting it all together we have
	\begin{align*}  
	\vol(Y)\cdot \rho /d &\leq \vol(T(D_P(\kappa'))) \leq 2\cdot \vol(T(D_P(\kappa))) \leq 2\cdot \vol(X)\cdot \diam(T(D_P(\kappa))) 
	\cr &\leq 2\cdot \vol(Y)\cdot \diam(T(D_P(\kappa)))  \leq 2\cdot \vol(Y)\cdot \diam([0,1]^d) = 2\sqrt{d} \cdot \vol(Y)  \end{align*}
	We infer that $\rho = \diam(T(D_P(\kappa'))) \leq 2 \cdot d^{3/2}$. Combining this with the lower bound on the width of $T(D_P(\kappa))$, we have that $\width(T(D_P(\kappa)))\geq \frac{1}{2\cdot d \cdot 5^d\cdot d!} \geq \frac{\diam(T(D_P(\kappa')))}{2^2\cdot d^{\frac 5 2} \cdot 5^d\cdot d!}$.
	
	Now, based on $T$ we define $M$~--- it is a composition of affine transformations. First we apply $T$, then we apply rotation and shrinking by a constant factor so that $T([0,1]^d)$ fits back inside the hypercube. Formally, $T$ is a shift, rotation and different rescaling of each direction, so that $B\stackrel{T}\mapsto [0,1]^d$. Now, the vertices of the hypercube $[0,1]^d$ are mapped to various points but they still make a box that contains all points mapped by $T$. So let $\varphi$ be the affine transformation that maps one of these vertices to the origin, rotates all vertices so that they align with the standard $d$-axes and scales everything by the same constant until this box fits into $[0,1]^d$. Then we define $M = \varphi\circ T$ as the composition of the two affine transformations. Note that shifts and rotations do not change lengths and that scaling doesn't change the ratio between lengths, so since $T(D_P(\kappa))$ is $(c_d,\Delta)$-fat we have that $\varphi(T(D_P(\kappa)))$ is also $(c_d,\Delta)$-fat.
	
	Now, denote $\tilde M$ as the transformation that applies $T$ and then removes any points that falls outside of the hypercube. Since $D_P(\kappa)\subset \B$ it follows that $T(D_P(\kappa))\subset[0,1]^d$ and so $\tilde M(D_P(\kappa))=T(D_P(\kappa))$. Thus, $\width(\tilde M(D_P(\kappa))) =\width(T(D_P(\kappa))) \geq 2d\cdot5^d\cdot(d!)$. Note that since $\tilde M$ caps the domain at the hypercube $[0,1]^d$, then for any $\kappa$ we have that $\tilde M(D_P(\kappa)) = T(D_P(\kappa))\cap[0,1]^d$. This means that the transition from $D_P(\kappa)$ to $M(D_P(\kappa))$ involves both computing the operation of $T$ on the vertices of $D_P(\kappa)$ and intersecting those with the $2d$ faces of the hypercube. However, by finding a $(\alpha,\Delta)$-kernel $\calS$ for $\tilde M(D_P(\kappa))$ we actually make it so that  $(1-\alpha)D_P(\kappa) \subset T^{-1}(\CH(\calS)) \subset T^{-1}(\tilde M(D_P(\kappa-\Delta)))$, which is a subset of $D_P(\kappa-\Delta)$ and might improve the overall performance of our algorithm.
\end{proof}

\section{Finding a ``Good'' $\kappa$ Privately}
\label{sec:finding_good_kappa}

Our discussion in Section~\ref{subsec:bounding_box_to_fatness} leaves us with the question of finding a good $\kappa$ and $\kappa'=\kappa-\Delta^{\rm kernel}$, namely a pair $(\kappa,\kappa')$ where $\vol(D_P(\kappa))\geq \vol(D_P(\kappa'))/2$. Once we have found that $\kappa,\kappa'$ are such that $\vol(D_P(\kappa)) \geq \vol(D_P(\kappa'))/2$, then we can find $T$ that makes $T(D_P(\kappa))$ to be $(c_d,\kappa-\Delta^{\rm kernel})$-fat, and apply Algorithm~\ref{alg:private-kernel} to find its $(\alpha,\Delta^{\rm kernel})$-kernel. Moreover, note that if we know that $(\kappa,\kappa')$ is a good pair, then for any $\kappa'\leq \tilde \kappa < \hat \kappa \leq \kappa$ it also holds that $\vol(D_P(\hat \kappa))\geq \vol(D_P(\tilde \kappa))/2$ so $(\hat \kappa,\tilde \kappa)$ is also a good pair. 

Yet, how do we know such a $(\kappa,\kappa')$-pair even exists? To establish this, we rely on the following result from~\cite{KaplanSS20}, regarding the volume of a non-zero Tukey region.
\begin{theorem}[Lemma~3.2 from~\cite{KaplanSS20}.]
	\label{thm:Kaplan_et_al_LB_volume_tukey_region}
	If the volume of a Tukey region is non-zero, then it is at least $(d/\Upsilon)^{-d^3}$.
\end{theorem}
\begin{corollary}
	Set $t = \lceil d^3\upsilon+d^3\log_2(d)\rceil$. Let $\kappa_1 < \kappa_2 < ... < \kappa_t$ be any monotonic series, and assume our input $P\subset [0,1]^d$ is such that $\vol(D_P(\kappa_t))>0$. Then there exists some $i\geq 2$ such that $(\kappa_i,\kappa_{i-1})$ is a good pair.
\end{corollary}
\begin{proof}
	ASOC that for all $2\leq i\leq t$ it holds that $(\kappa_i, \kappa_{i-1})$ isn't a good pair. This implies that for any such $i$ we have that $\vol(D_P(\kappa_i)) < \vol(D_P(\kappa_{i-1}))/2$. Through induction, we can show that $\vol(D_P(\kappa_t)) < \vol(D_P(\kappa_1))/2^t \leq 1/2^t <(d/\Upsilon)^{-d^3}$, contradicting Theorem~\ref{thm:Kaplan_et_al_LB_volume_tukey_region}.
\end{proof}
In our work, since we want a good pair of $\kappa$s which is at least of distance $\Delta^{\rm kernel}$ apart, we look at the series $\kappa_i = i\cdot (4\Delta^{\rm kernel})$. This means that among the $m=4t\Delta^{\rm kernel}$ indices we consider, there are multiple pairs of $(\kappa,\kappa')$ where $\kappa' \leq \kappa-\Delta^{\rm kernel}$ that are good. This motivates the query that we use throughout this section
\begin{equation}
\label{eq:query_good_kappa}
q_P(\kappa) \stackrel{\rm def} = \max\left\{0\leq i \leq \min\{\kappa-1, m-\kappa\}:~ \frac{\vol(D_P(\kappa+i))}{\vol(D_P(\kappa-i))}\geq \frac 1 2 \right\}
\end{equation}
It is obvious that $q_p(\kappa)\geq 0$ for any $\kappa$, and that $q_P(1)=q_P(m)=0$. But our goal is to retrieve a $\kappa$ where $q_P(\kappa) \geq \Delta^{\rm kernel}$ which allows us to establish that the pair $(\kappa,\kappa-\Delta^{\rm kernel})$ is a good pair. The question we discuss in this section is how to output such a $\kappa$ in a way which is $(\epsilon,\delta)$-differentially private. The answer we give is based on a rather uncommon composition of two mechanisms (the exponential mechanism and the additive discrete Laplace) and is discussed in detail in the following section. However, in order to give our differntially-private mechanism, we prove that $q$ is a query which exhibits sort-of low sensitivity. Details are in the following claim.

\begin{claim}
	\label{clm:good_kappa_query_low_sensitivity}
	Let $P\subset[0,1]^d$ and $P'\subset[0,1]^d$ be two neighboring datasets where for some $x\in [0,1]^d$ it holds that $P' = P\cup \{x\}$. Then for any $1\leq \kappa<m$ we have that
	\[  |q_P(\kappa)-q_{P'}(\kappa+1)|\leq 1 \]
\end{claim}
Note that the claim implies that when $P = P'\cup\{x\}$ for some $x\in \R^d$ then $|q_P(\kappa)-q_{P'}(\kappa-1)|\leq 1$ for any $1<\kappa\leq m$.
\begin{proof}
 	Fix such $P$ and $P'$. 
 	Note that since $P'$ has one additional point then $P$, then for any $x$ it holds that $\TD(x,P)\leq \TD(x,P') \leq \TD(x,P)+1$. Thus, for any $\kappa$ we have that $D_{P}(\kappa) = \{x:~ \TD(x,P)\geq \kappa\} \subset D_{P'}(\kappa)$ and similarly, $D_{P'}(\kappa)\subset D_P(\kappa-1)$. We thus have a chain:
 	\[ D_{P'}(1) \supset D_P(1)\supset D_{P'}(2)\supset D_P(2) \supset D_{P'}(3) \supset... \supset D_{P}(\kappa-1)\supset D_{P'}(\kappa)\supset D_P(\kappa)\supset D_{P'}(\kappa+1) \supset ...   \]
 	
 	Denote $q_P(\kappa)=i^*$. We show that $i^*-1\leq q_{P'}(\kappa+1) \leq i^*+1$ proving the required. First, since $q_P(\kappa)=i^*$ then we know that 
 	\[\frac 1 2 \leq \frac{\vol(D_P(\kappa+i^*))}{\vol(D_P(\kappa-i^*))}\leq \frac{\vol(D_{P'}(\kappa+i^*))}{\vol(D_{P'}(\kappa-i^*+1))} \leq\frac{\vol(D_{P'}(\kappa+i^*))}{\vol(D_{P'}(\kappa-i^*+2))} =\frac{\vol(D_{P'}(\kappa+1+(i^*-1)))}{\vol(D_{P'}(\kappa+1-(i^*-1)))}\]
 	proving that $q_{P'}(\kappa+1)\geq i^*-1$. Secondly, since $q_P(\kappa)<i^*+1$ then we know that
 	\[\frac 1 2 > \frac{\vol(D_P(\kappa+i^*+1))}{\vol(D_P(\kappa-i^*-1))} \geq \frac{\vol(D_{P'}(\kappa+i^*+2))}{\vol(D_{P'}(\kappa-i^*-1))} \geq \frac{\vol(D_{P'}(\kappa+i^*+3))}{\vol(D_{P'}(\kappa-i^*-1))} = \frac{\vol(D_{P'}(\kappa+1+(i^*+2)))}{\vol(D_{P'}(\kappa+1-(i^*+2)))}\]
 	proving that $q_{P'}(\kappa+1)<i^*+2$ thus $q_{P'}(\kappa+1)\leq i^*+1$.
\end{proof}

We are now ready to give our differentially private algorithm that returns a good $\kappa$. We refer to it as the ``Shifted Exponential Mechanism.''

\begin{algorithm}[hbt]
	\caption{\label{alg:private-good-kappa} The Shifted Exponential Mechanism}
	\textbf{Input:} Dataset $P\subset\G^{d}$; privacy parameter $\epsilon\in(0,1)$; Set of indices $[m]$ with $m\geq\frac{16}\epsilon$.
	
	\begin{algorithmic}[1]
		\STATE 	Compute the Tukey-regions $D_P(\kappa)$ for all $1\leq \kappa\leq m$ and their respective volume.
		\STATE Pick $\kappa\in[m]$ w.p. $\propto w_P(\kappa) = \exp(\frac \epsilon 8q_P(\kappa))$.\\
		Namely, find $w_P(\kappa)$ for every $1\leq \kappa\leq m$ and compute $W_P = \sum_{\kappa=1}^m w_P(\kappa)$, then $\Pr[\kappa] = w_P(\kappa)/W_P$.
		\STATE Pick $X\sim {\sf D}\Lap{\frac 8 \epsilon}$, the discrete Laplace distribution.
		\\ Namely, pick $X=i$ w.p. $\propto \omega(i)= e^{-\frac \epsilon 8 |i|}$.
		\RETURN $\kappa+X$.
	\end{algorithmic}
\end{algorithm}

\newcommand{\calM}{\mathcal{M}}
\newcommand{\Z}{\mathbb{Z}}
\begin{theorem}
	\label{thm:shifted_exp_mech_is_private}
	If $\epsilon<1$ and $m\geq \frac{16}\epsilon$ then the Shifted Exponential Mechanism is $\epsilon$-differentially private.
\end{theorem}
\begin{proof}
	Perhaps the key point in this theorem is that our algorithm is pure-DP and not approximated-DP. To establish this fact we require Claim~\ref{clm:good_kappa_query_low_sensitivity} as well as the facts that (i) $q_P(\kappa)\geq 0$ for any $P$ and any $\kappa$ and (ii) $q_P(1)=q_P(m)=0$ for any $P$.
	
	We prove the $\epsilon$-DP property, by breaking symmetry, and for now we fix $P$ and $P'$ as two neighboring datasets subject to $P' = P \cup \{x\}$ for some $x$. Denoting $\calM$ as the shifted exponential mechanism, our goal is to show that for any $j\in \Z$ we have that $\Pr[\calM(P)=j]/\Pr[\calM(P')=j]\leq\exp(\epsilon)$.
	
	First, we establish the following inequality.
	\begin{align*}
	W_P &= \sum_{\kappa=1}^m w_P(\kappa) = \sum_{\kappa=1}^{m-1}\exp(\frac \epsilon 8 q_P(\kappa)) + \exp(\frac \epsilon 8 q_P(m))
	 \stackrel{\rm Claim~\ref{clm:good_kappa_query_low_sensitivity}}\leq 
	\sum_{\kappa=1}^{m-1}\exp(\frac \epsilon 8 (q_{P'}(\kappa+1)+1)) + \exp(0)
	\cr &
	\leq e^{\frac\epsilon 8} \sum_{\kappa=2}^{m}\exp(\frac \epsilon 8 q_{P'}(\kappa)) + e^{\frac \epsilon 8}\cdot 1 = e^{\frac \epsilon 8} \sum_{\kappa=1}^{m}\exp(\frac \epsilon 8 q_{P'}(\kappa)) = e^{\frac \epsilon 8}W_{P'}
	\intertext{Similarly,}
	W_P &= \sum_{\kappa=1}^m w_P(\kappa) = \sum_{\kappa=1}^{m-1}\exp(\frac \epsilon 8 q_P(\kappa)) + \exp(0)
	\stackrel{\rm Claim~\ref{clm:good_kappa_query_low_sensitivity}}\geq 
	\sum_{\kappa=1}^{m-1}\exp(\frac \epsilon 8 (q_{P'}(\kappa+1)-1)) + e^{-\frac \epsilon 8}
	 = e^{-\frac \epsilon 8}W_{P'}
	\end{align*}
	
	Second, we turn our attention to the discrete Laplace distribution. It is evident that for any $i_1, i_2$ s.t. $|i_1-i_2|=1$ we have that $\Pr[X=i_1]/\Pr[X=i_2] = \omega(i_1)/\omega(i_2) \leq e^{\frac \epsilon 8}$. But perhaps more interesting is the following claim: for any index $i$ and any closed interval $I$ with one endpoint in $i$ and of length $8/\epsilon$ (i.e., either $I \supset [i, i+\frac 8 \epsilon]$ or $I \supset [i-\frac 8\epsilon, i]$) we have that
	\begin{align*}
	\frac{\Pr[X=i]}{\Pr[X\in I]} &\leq \frac{\exp(-\frac \epsilon 8 |i|)}{\exp(-\frac \epsilon 8|i|)\left( e^0 + e^{-\frac\epsilon 8}+...+ e^{-\frac{\epsilon}8\cdot \frac{8}\epsilon}   \right)}\leq \frac{1}{e^{-1}\cdot \frac{8}\epsilon} \leq \frac{3\epsilon}{8}
	\end{align*}
	
	With these two inequalities in our disposal, we can now argue that the mechanism is differentially private. Fix any $j\in \Z$.
	\begin{align*}
	\frac{\Pr[\calM(P)=j]}{\Pr[\calM(P')=j]} &= \frac{\sum_{i = 1}^m \Pr_P[\kappa=i]\Pr[X = j-i]}{\sum_{i = 1}^m \Pr_{P'}[\kappa=i]\Pr[X = j-i]} = \frac{W_{P'}\cdot \sum_{i = 1}^m w_P(i)\omega(j-i)}{W_P\cdot \sum_{i = 1}^m w_{P'}(i)\omega(j-i)}
	\cr &\leq e^{\frac \epsilon 8}\frac{\sum_{i = 1}^{m-1} w_P(i)\omega(j-i) + w_P(m)\omega(j-m)}{\sum_{i = 2}^m w_{P'}(i)\omega(j-i) + w_{P'}(1)\omega(j-1)} 
	\cr &\leq e^{\frac \epsilon 8}\frac{\sum_{i = 1}^{m-1} \exp(\frac \epsilon 8 q_P(i))\omega(j-i) + \exp(0)\omega(j-m)}{\sum_{i = 2}^m \exp(\frac \epsilon 8 q_{P'}(i))\omega(j-i)}
	\cr \text{(Claim~\ref{clm:good_kappa_query_low_sensitivity})}~~~~~~~~ &\leq e^{\frac \epsilon 8}\left(\frac{\sum_{i = 1}^{m-1} \exp(\frac \epsilon 8 (q_{P'}(i+1)+1))\omega(j-i)}{\sum_{i = 2}^m \exp(\frac \epsilon 8 q_{P'}(i))\omega(j-i)}+ \frac{\exp(0)\omega(j-m)}{\sum_{i = 2}^m \exp(\frac \epsilon 8 q_{P'}(i))\omega(j-i)}\right)
	\cr (\forall i, q_{P'}(i)\geq 0)~~~~ &\leq
	e^{\frac \epsilon 8}\left(e^{\frac \epsilon 8}\frac{\sum_{i = 2}^{m} \exp(\frac \epsilon 8 q_{P'}(i))\omega(j-i+1)}{\sum_{i = 2}^m \exp(\frac \epsilon 8 q_{P'}(i))\omega(j-i)}+ \frac{\omega(j-m)}{\sum_{i = 2}^m \omega(j-i)}\right)
	 \cr &\leq
	e^{\frac \epsilon 8}\left(e^{\frac \epsilon 8}\frac{e^{\frac \epsilon 8}\sum_{i = 2}^{m} \exp(\frac \epsilon 8 q_{P'}(i))\omega(j-i)}{\sum_{i = 2}^m \exp(\frac \epsilon 8 q_{P'}(i))\omega(j-i)}+ \frac{\omega(j-m)}{\sum_{t = 0}^{m/2} \omega(j-m+t)}\right)
	\cr (m\geq\frac{16}{\epsilon})~~~~ &\leq e^{\frac \epsilon 8}\left(e^{\frac\epsilon 4}+ \frac{\omega(j-m)}{\sum_{t = 0}^{8/\epsilon} \omega(j-m+t)} \right) \stackrel{e^{\frac \epsilon 4}\leq 1+\frac{\epsilon}2}\leq e^{\frac \epsilon 8}\left(1 + \frac \epsilon 2 + \frac{3\epsilon}8 \right) \leq e^{\frac \epsilon 8}e^{\frac{7\epsilon} 8} = e^{\epsilon}
	\intertext{And symmetrically, we have}
	\frac{\Pr[\calM(P')=j]}{\Pr[\calM(P)=j]} &= \frac{\sum_{i = 1}^m \Pr_{P'}[\kappa=i]\Pr[X = j-i]}{\sum_{i = 1}^m \Pr_{P}[\kappa=i]\Pr[X = j-i]} = \frac{W_{P}\cdot \sum_{i = 1}^m w_{P'}(i)\omega(j-i)}{W_{P'}\cdot \sum_{i = 1}^m w_{P}(i)\omega(j-i)}
	\cr &\leq e^{\frac \epsilon 8}\frac{\sum_{i = 2}^{m} \exp(\frac \epsilon 8 q_{P'}(i))\omega(j-i) + \exp(0)\omega(j-1)}{\sum_{i = 1}^{m-1} \exp(\frac \epsilon 8 q_{P}(i))\omega(j-i)}
	\cr \text{(Claim~\ref{clm:good_kappa_query_low_sensitivity})}~~~~~~~~ &\leq e^{\frac \epsilon 8}\left(\frac{\sum_{i = 2}^{m} \exp(\frac \epsilon 8 q_{P'}(i))\omega(j-i)}{\sum_{i = 1}^{m-1} \exp(\frac \epsilon 8 (q_{P'}(i+1)-1))\omega(j-i)}+ \frac{\exp(0)\omega(j-1)}{\sum_{i = 1}^{m-1} \exp(\frac \epsilon 8 q_{P}(i))\omega(j-i)}\right)
	\cr (\forall i, q_{P'}(i)\geq 0)~~~~ &\leq
	e^{\frac \epsilon 8}\left(e^{\frac \epsilon 8}\frac{\sum_{i = 2}^{m} \exp(\frac \epsilon 8 q_{P'}(i))\omega(j-i)}{\sum_{i = 2}^m \exp(\frac \epsilon 8 q_{P'}(i))\omega(j-i+1)}+ \frac{\omega(j-1)}{\sum_{i = 1}^{m-1} \omega(j-i)}\right)
	\cr &\leq
	e^{\frac \epsilon 8}\left(e^{\frac \epsilon 8}e^{\frac\epsilon 8}+ \frac{\omega(j-1)}{\sum_{t = 0}^{8/\epsilon} \omega(j-1-t)}\right)
	\stackrel{e^{\frac \epsilon 4}\leq 1+\frac{\epsilon}2}\leq e^{\frac \epsilon 8}\left(1 + \frac \epsilon 2 + \frac{3\epsilon}8 \right) \leq e^{\frac \epsilon 8}e^{\frac{7\epsilon} 8} = e^{\epsilon}
	\end{align*}
\end{proof}

The utility of the shifted exponential mechanism follows basically from the utility guarantees of both the exponential mechanism and additive Laplace noise mechanism.
\begin{theorem}
	\label{thm:shifted_exp_mechanism_utility}
	W.p. $\geq 1-\beta$ Algorithm~\ref{alg:private-good-kappa} returns an index $j$ such that $q_P(j)\geq 2\Delta^{\rm kernel} - \frac{17\ln(2m/\beta)}{\epsilon}$
\end{theorem}
\begin{proof}
	Denote $j^*$ as the index whose $q_P(j^*)$ is the largest among all $m$ indices. Denote $\tau = 8\log(2m/\beta)/\epsilon$. By standard argument, the probability that $\kappa$ is such that $q_P(\kappa)\leq q_P(j^*)-\tau$ can be upper bounded by:
	\[ \Pr[q_P(\kappa)\leq q_P(j^*)-\tau] \leq \frac{\Pr[q_P(\kappa)\leq q_P(j^*)-\tau]}{\Pr[\kappa = j^*]} \leq \frac{m\exp(\frac \epsilon 8(q_P(j^*)-\tau) )}{\exp(\frac \epsilon 8 q_P(j^*))} = m\exp(-\frac {\epsilon \tau}8) \leq \frac\beta 2   \] 
	Secondly, setting $\tau' = \frac{8\ln(2/\beta)}{\epsilon}$ we can straight-forwardly calculate \[\Pr[|X|> \tau' ] = \frac{2\sum_{i>\tau'}e^{-\frac{\epsilon i}8}}{1+2\sum_{i>0} e^{-\frac{\epsilon i}8}} \leq \frac{e^{-\epsilon\tau'/8}\sum_{i>0} e^{-\frac{\epsilon i}8}}{\sum_{i>0} e^{-\frac{\epsilon i}8}} \leq \frac{\beta}{2}\]
		
	Lastly, one can repeat the same argument from Claim~\ref{clm:good_kappa_query_low_sensitivity} to show that for any index $\kappa'$ we have that  $|q_P(\kappa')-q_P(\kappa'+1)|\leq 1$. (In fact, it is best to consider the dataset $P' = P\cup\{x\}$ where $x$ is a point of Tukey-depth $\TD(x,P)\geq m+1$; this implies that for any $\kappa'\leq m$ the Tukey-regions $D_P(\kappa')=D_{P'}(\kappa')$ as the borders of the Tukey-region are hyperplane that separate exactly $\kappa'$ points from the rest of $P$.) 
	
	Putting it all together, we have that w.p. $\geq 1-\beta$ Algorithm~\ref{alg:private-good-kappa} returns an index $j$ which is of distance $\leq \tau'$ from an index $\kappa$ which query-value of $q_P(\kappa)\geq q_P(j^*)-\tau$. This implies that 
	\begin{align*}
	q_P(j)&\geq q_P(\kappa) -\tau' \geq q_P(j^*)-(\tau+\tau') \geq 2\Delta^{\rm kernel}-1 - \frac{8\left( \ln(\frac{2m}{\beta})+\ln(\frac{2}{\beta}) \right)}\epsilon \geq 2\Delta^{\rm kernel} - \frac{17\ln(\frac{2m}{\beta})}\epsilon
	\end{align*}
\end{proof}
All that is left is to assert that $m = 4\lceil d^3\upsilon+d^3\log_2(d)\rceil\Delta^{\rm kernel}$ is (a) $\geq \frac{16}\epsilon$ and (b) is such that $\frac{17\ln(\frac{2m}{\beta})}{\epsilon}\leq  \Delta^{\rm kernel}$. Clearly, requirement (b) is the stricter and once $m$ satisfies it then (a) also holds. To assert (b) we rely on the following fact.
\begin{fact}
	\label{fact:x-aln(bx)_function}
	For any $a,b>0$ such that $ab>2$ we have that the function $f(x) = x-a\ln(bx) > 0$ for any $x>2a\ln(ab)$.
\end{fact}
\begin{proof}
	Clearly, $f'(x) = 1-\frac{a}{x} >0$ for any $x>a$  so on the interval $(2a\ln(ab),\infty)$ it is monotonic increasing. Thus, on this interval
	\begin{align*}
	f(x) &> f(2a\ln(ab)) = 2a\ln(ab) - a\ln(2ab\ln(ab)) = 2a\ln(ab)-\left(a\ln(ab) + a\ln(2\ln(ab)) \right)
	\cr & = a\ln(ab)-a\ln(2\ln(ab)) = a\ln(\frac{ab}{2\ln(ab)}) > 0
	\end{align*}
	since $ab>2\ln(ab)$ for any $ab>2$.
\end{proof}
Setting $a = \frac{17}\epsilon$ and $b=\frac{8\lceil d^3\upsilon+d^3\log_2(d)\rceil}{\beta}$ we have that (b) holds when $\Delta^{\rm kernel}\geq \frac{34}\epsilon\ln(\frac{8\cdot 17\lceil d^3\upsilon+d^3\log_2(d)\rceil}{\beta\epsilon})$. Under almost any setting of parameters (unless, among other things, $\epsilon\ll \beta$) we have that this bound is far smaller than the definitions of $\Delta^{\rm kernel}$ given in Theorem~\ref{thm:private_kernel_abs_fatness}. We thus can infer the following corollary.
\begin{corollary}
	\label{cor:good_kappa}
	Fix $\epsilon>0, \delta\geq0, \beta>0$ and set $\Delta^{\rm kernel}$ as in Theorem~\ref{thm:private_kernel_abs_fatness} and $m=4\lceil d^3\upsilon+d^3\log_2(d)\rceil\Delta^{\rm kernel}$. Let $P\subset \G^d$ be a set of points such that $D_P(m)$ is non-empty and non-degenerate. Then w.p. $\geq 1-\beta$, Algorithm~\ref{alg:private-good-kappa} return a value $j$ such that $\vol(D_{P}(j))/\vol(D_P(j-\Delta^{\rm kernel}))\geq 1/2$.
\end{corollary}

\paragraph{Comment.} There exists a variant of Algorithm~\ref{alg:private-good-kappa} where instead of picking $\kappa$ in Step 2 in a fashion similar to the exponential mechanism, we pick $\kappa$ by iterating on all indices from $1$ to $m$ and apply the SVT to halt on the \emph{first} index $\kappa$ where $q_P(\kappa)$ exceeds some noisy threshold. This variant has the benefit that it returns the \emph{first} index whose query-value exceed $\Delta^{\rm kernel}$. However, we were only able to show that this variant is $(\epsilon,\delta)$-differentially private since we need that when applying the SVT step to $P'$ we (a) don't halt on the very first query and (b) its noisy value doesn't exceed the value of other queries. Showing this variation of the mechanism is $\epsilon$-differentially private is a challenging open problem.

\section{Conclusions,  Applications and Open Problems.}
\label{sec:conclusion}

To conclude, we give an overview of the steps required for privately outputting a kernel for a Tukey-region $D(\kappa)$. In the following, we assume the privacy parameter $\epsilon,\delta$ are given, as well as $\alpha,\beta$. We set $c_d = 2d\cdot 5^d \cdot (d!)$, $\Delta^{\rm kernel}$ as in Theorem~\ref{thm:private_kernel_abs_fatness} and $m = 4\lceil d^3\upsilon+d^3\log_2(d)\rceil\Delta^{\rm kernel}$. We assume $P\subset \G^d$ is our input where $n\geq \frac{m}{d+1}$ which asserts that $D_{P}(m)$ is non-empty. (We can simply to do a private count of $|P|$ using additive $\Lap{\frac 1 \epsilon}$ noise and abort if $n<2(d+1)m$.)
\begin{enumerate}
	\setcounter{enumi}{-1}
	\item Preprocess: check $P$ is large enough and the $D_{P}(m)$ isn't of degenerate (in a subspace of dimension $<d$).
	\item (Corollary~\ref{cor:good_kappa}) Apply the $\epsilon$-DP Algorithm~\ref{alg:private-good-kappa} to find a $\kappa$ where $\vol(D_P(\kappa))\geq \frac 1 2 \vol(D_{P}(\kappa-\Delta^{\rm kernel}))$.
	\item (Optional:) Use the $\epsilon$-DP Algorithm~\ref{alg:DP_approx_width_tukey_region} to see if $D_P(\kappa)$ is $\width_{\kappa}\geq c_d$. If so, skip next step.
	\item (Theorem~\ref{thm:dp_bounding_box}) Apply the $(\epsilon,\delta)$-DP Algorithm~\ref{alg:private-bounding-box-approx} to approximate the bounding box of $D_P(\kappa)$ and set $\tilde M$ as the transformation that maps this box to the hypercube $[0,1]^d$ and guarantees that $D_P(\kappa)$ is $c_d$-absolutely fat. 
	\item (Theorem~\ref{thm:private_kernel_abs_fatness}) Apply the $(\epsilon,\delta)$-DP Algorithm~\ref{alg:private-kernel} to get $\calS$~--- a $(\alpha,\Delta^{\rm kernel})$-kernel for $D_P(\kappa)$.
	\item Return $(\kappa,\calS)$.
\end{enumerate}

\paragraph{Applications.} Having outputted a $(\alpha,\Delta)$-kernel for $D_P(\kappa)$, it is possible to post-process the resulting set $\calS$ in various ways, in order to produce multiple estimations regarding the convex bodies $D_P(\kappa)$ and $D_P(\kappa-\Delta)$. We present here a short summary of such applications. 
\begin{itemize}
	\item Minimum enclosing ball:\\
	Denote $r_{\kappa}$ as the radius of the min-enclosing ball of $D_P(\kappa)$ and $r_{\kappa-\Delta}$ as the radius of the min-enclosing ball of $D_P(\kappa-\Delta)$. Then the radius of the min enclosing ball of $\CH(\calS)$ is in the range $[(1-\alpha')r_{\kappa}, (1+\alpha')r_{\kappa-\Delta}]$.\\
	The same holds for the radius of the max-enclosed ball.
	\item Minimum enclosing ellipsoid/box:\\
	Denote $V_{\kappa}$ as the volume of the min-enclosing ellipsoid of $D_P(\kappa)$ and $V_{\kappa-\Delta}$ as the volume of the min-enclosing ellipsoid of $D_P(\kappa-\Delta)$. Then the volume of the min-enclosing ellipsoid of $\CH(\calS)$ is in the range $[(1-2d\alpha')V_{\kappa}, (1+2d\alpha')V_{\kappa-\Delta}]$ (assuming $\alpha'<1/4d$). The same holds for the volume of the min enclosing box (or any other particular convex body). 
	\item Surface area of the convex body:\\
	Denote $A_{\kappa}$ as the area of the facets of $D_P(\kappa)$ and $V_{\kappa-\Delta}$ as the area of the facets of $D_P(\kappa-\Delta)$. Then the area of the facets of $\CH(\calS)$ is in the range $[(1-2(d-1)\alpha')A_{\kappa}, (1+2(d-1)\alpha')A_{\kappa-\Delta}]$ (assuming $\alpha'<1/4(d-1)$). The same holds for the surface area of any min enclosing convex body. 
\end{itemize}
Agrawal et al~\cite{AgarwalHV04} define a function $\mu$ of a dataset as a \emph{faithful  measure} if (i) $\mu$ is non-negative, (ii) for every $P\subset \R^d$ we have $\mu(P)=\mu(\CH(P))$, (iii) $\mu$ is monotone w.r.t containment of convex bodies, and most importantly, that (iv) for some $c\in(0,1)$ a $(1-c\alpha)$-kernel of $P$ yields a $(1-\alpha)$-approximation of $\mu(P)=\mu(\CH(P))$. Obviously, any faithful measure $\mu$ can be approximated by a $(\alpha,\Delta)$-kernel $\calS$ where $(1-\frac{\alpha} c)\mu(D_P(\kappa))\leq \mu(\CH(\calS)) \leq (1+\frac{\alpha} c)\mu(D_P(\kappa-\Delta))$. 

\paragraph{Open Problems.} This work is the first to propose a differentially private approximation algorithms for some key concepts in computational geometry, such as diameter, width, convex hull, min-bounding box, etc.  As such it leaves many more questions unanswered, some of which are described here.

First and far most~--- our algorithm doesn't yield a kernel approximation for any given $\kappa$, but rather it requires $D_P(\kappa)$ to satisfy some properties in order for us to privately transform it into a fat Tukey-region. Granted, we also give heuristics that check whether a given region is fat, or whether we can turn it into a fat region. But the transformation we provide in this work (as well as a private analogue of the transformation in~\cite{BarequetH99} which we didn't detail) requires some notion of proximity between $D_P(\kappa)$ and the shallower $D_P(\kappa-\Delta)$. It is unclear whether there exists a DP algorithm that can yield a kernel for \emph{any} $D_P(\kappa)$. Granted, such a case isn't really interesting~--- what is the purpose of a $(\alpha,\Delta)$-kernel of $D_P(\kappa)$ when $D_P(\kappa)$ and $D_P(\kappa-\Delta)$ are very different? But it would still be interesting to have a complete picture as to our ability (or inability) to privately approximate any $D_P(\kappa)$. Second, it is also interesting to see if we can find a private analogue of the better kernel-algorithm in~\cite{AgarwalHV04}, which requires a discretization involving only $O(\frac{c_d}\alpha)^{\frac{d-1}{2}}$-many directions. Note that devising a private analogue for this better algorithm will have a significant affect on the definition of $\Gamma^{\rm kernel}$ in Theorem~\ref{thm:dp_approx_kernel}. We were unable to implement it and ``save'' on the privacy loss since the algorithm requires nearest-neighbor approximation of $O(\frac{c_d}\alpha)^{\frac{d-1}{2}}$-many points. (The algorithm that sweeps along directions uses, unfortunately, a discretization of directions with size roughly $(\frac{\alpha}{c_d})^{d-1}$, producing no improvement on the current utility guarantee.)
Third, we pose the problem of complimentary lower bounds for our theorems. Even for finding the width of a Tukey region, devising a lower bound based on ``Packing Arguments''~\cite{HardtT10} seems non-trivial, let alone for problems such as min-bounding box / kernel approximations.

Our work only touches upon the very first and very fundamental geometric concepts. But the field of computational geometry has many more algorithms for various other tasks that one may wish to make private. We refer the interested reader to the books~\cite{compgeom:2000,Har-Peled11}. More importantly, the field of computation geometry often deals solely with the given input and avoids any distributional assumptions. Thus, it is unclear whether the results of many algorithms in computational geometry \emph{generalize} should the data be drawn i.i.d from some unknown distribution $\calP$. Tukey depth is, potentially, just one of many possible ways to argue about generalizations of concepts in geometry, as other notions of depths have been proposed and some do generalize~\cite{ZuoS00a}.

\newpage
{
\bibliography{paper}
\bibliographystyle{alpha}
}
\end{document}